\documentclass{sig-alternate-mod}

\newcommand{\mysize}{\fontsize{10pt}{10pt}\selectfont}

\usepackage[boxruled, vlined]{algorithm2e}
\usepackage{graphicx,amssymb,wrapfig,makeidx}

\renewcommand{\P}{\mbox{${\cal P}$}}

\newtheorem{ex}{EXAMPLE}[section]
\newenvironment{example}{\begin{ex} \nopagebreak
  \begin{rm}}{{\hfill$\Box$}\end{rm}\end{ex}} 

\newtheorem{defin}{Definition}[section]
\newenvironment{definition}[1]{\begin{defin}\begin{rm}({\bf #1})}{{\hfill$\Box$}\end{rm}\end{defin}}

\newtheorem{lemm}{Lemma}[section]
\newenvironment{lemma}{\begin{lemm}}{{\hfill$\Box$}\end{lemm}}

\newtheorem{thm}{Theorem}[section]
\newenvironment{theorem}{\begin{thm} \nopagebreak}{{\hfill$\Box$}\end{thm}}

\newtheorem{prop}{Proposition}[section]
\newenvironment{proposition}{\begin{prop}}{{\hfill$\Box$}\end{prop}}

\newtheorem{corol}{Corollary}[section]
\newenvironment{corollary}{\begin{corol} \nopagebreak}{{\hfill$\Box$}\end{corol}}

\newtheorem{conjec}{Conjecture}[section]

\newsavebox{\savepar}

\newcommand{\x}[1]{\mbox{\ensuremath{#1}}}

\newcommand{\squishlist}{
  \begin{list}{$\bullet$}
   {
     \setlength{\itemsep}{0pt}
     \setlength{\parsep}{0pt}
     \setlength{\topsep}{0pt}
     \setlength{\partopsep}{0pt}
     \setlength{\leftmargin}{1.5em}
     \setlength{\labelwidth}{1em}
     \setlength{\labelsep}{0.5em} } }
\newcommand{\squishend}{
   \end{list}  }

\newcommand{\reminder}[1]{{\bf [[[***** #1*****]]] }}
\newcommand{\nop}[1]{}                       

 \SetAlgoSkip{4cm}

\begin{document}

\title
{Equivalence of SQL Queries \\ In Presence of Embedded Dependencies
}





\numberofauthors{2} 

\author{
\alignauthor Rada Chirkova
\\
	\affaddr{Department of Computer Science} \\
	\affaddr{NC State University, Raleigh, NC 27695, USA} \\
	\email{chirkova@csc.ncsu.edu}
\alignauthor Michael R. Genesereth\\
	\affaddr{Department of Computer Science} \\
	\affaddr{Stanford University, Stanford, CA 94305, USA} \\
	\tiny{\email{genesereth@stanford.edu}}
}

\maketitle

{\mysize

\begin{abstract}
{\mysize

 We consider the problem of finding equivalent minimal-size reformulations of SQL queries in presence of embedded dependencies~\cite{AbiteboulHV95}. Our focus is on select-project-join (SPJ) queries with equality comparisons, also known as safe conjunctive {\em (CQ)}  queries, 
 possibly with grouping and aggregation.  For SPJ queries, the semantics of the SQL standard treat query answers as {\em multisets} (a.k.a. {\em bags}), whereas the stored relations may be treated either as sets, which is called {\em bag-set semantics} for query evaluation, or as bags, which is called {\em bag semantics}. (Under {\em set semantics,} both query answers and stored relations are treated as sets.)

In the context of the above Query-Reformulation Problem, we develop a comprehensive framework for equivalence of  CQ queries under bag and bag-set semantics in presence of embedded dependencies, and make a number of conceptual and technical contributions. 
Specifically, we develop equivalence tests for CQ queries in presence of arbitrary sets of embedded dependencies under bag and bag-set semantics, under the condition that  chase~\cite{DeutschPods08} under  set semantics {\em (set-chase)} 
on the inputs terminates. We also present   
equivalence tests for {\em aggregate} CQ queries 
in presence of embedded dependencies. 
We use our equivalence tests to develop sound and complete (whenever set-chase on the inputs terminates) algorithms for solving instances of the Query-Reformulation Problem with CQ queries under each of bag and bag-set semantics, as well as for instances of the problem with aggregate queries. 
%

Some of our results are of independent interest. In particular, it is known that constraints that force some relations to be sets on all instances of a given database schema arise naturally in the context of  sound  (i.e., correct) chase~\cite{DeutschDiss} under bag semantics. We develop a formal framework for defining such constraints as 
embedded dependencies, 
	provided that {\it row (tuple) IDs,} commonly used in commercial database-management systems, are defined for the respective relations. 


We also extend the condition of~\cite{VardiBagsPods93} for bag equivalence of CQ queries, to those cases where some relations are set valued in all instances of the given schema. 
	Our proof of this nontrivial result 
	includes reasoning involving bag (non)containment. In particular, we provide an original proof (adapted to our context) of the result of~\cite{VardiBagsPods93} that CQ query $Q_1$ is bag contained in CQ query $Q_2$ only if, for each predicate used in $Q_1$, $Q_2$ has at least as many subgoals with this predicate as $Q_1$ does. 



	


Our contributions are clearly applicable beyond the Query-Reformulation Problem considered in this paper. Specifically, the results of this paper can be used in developing algorithms for rewriting CQ queries and queries in more expressive languages (e.g., including grouping and aggregation, or   arithmetic comparisons) using views in presence of embedded dependencies, under bag or bag-set semantics for query evaluation. 

} 

This text contains corrections to Sections 2.4 and 4 of \cite{ChirkovaG09}. 

\end{abstract}


{\mysize

\vspace{-0.2cm}

\section{Introduction}

Query containment and equivalence were recognized fairly early as fundamental problems in  database query evaluation and optimization. The reason is, for  conjunctive queries ({\em CQ queries}) --- a broad class of frequently used queries, whose expressive power is equivalent to that of select-project-join queries in relational algebra --- query equivalence can be used as a tool in query optimization. 
Specifically, to find a more efficient {\it and} answer-preserving formulation of a given CQ query, it is enough to ``try all ways'' of arriving at a ``shorter'' query formulation, by removing query subgoals, in a process  called query minimization~\cite{ChandraM77}. A subgoal-removal step succeeds only if equivalence (via containment) of the ``original'' and ``shorter'' query formulations can be ensured. The equivalence test of~\cite{ChandraM77} for CQ queries is known to be NP complete, whereas equivalence of general relational queries is undecidable. 

In recent years, there has been renewed interest in the study of query containment and equivalence, because of their close relationship to the problem of answering queries using views \cite{LevyAquvSurvey}. 
In particular, the problem of rewriting relational queries equivalently using views 
has been the subject of extensive rigorous investigations. Please see~\cite{DeutschPT06,LevyAquvSurvey,ChenLiEncyclop,Ullman00} for discussions of the state of the art and of the numerous practical applications of the problem. A test for equivalence of a CQ query to its candidate CQ rewriting in terms of CQ views uses an equivalent transformation of the rewriting to its CQ {\em expansion,} which (informally speaking)
 replaces references to views in the rewriting by their definitions~\cite{LevyAquvSurvey,Ullman00}. Then the equivalence test succeeds if and only if the expansion of the rewriting is equivalent, via the equivalence test of~\cite{ChandraM77}, to the input query. 

Some of  the investigations discussed in~\cite{DeutschPT06,LevyAquvSurvey,ChenLiEncyclop,Ullman00} focused on view-based query rewriting in presence of integrity constraints (also called {\it dependencies}, see~\cite{AbiteboulHV95} for an overview). For a given query, accounting for the dependencies that hold on the database schema may increase the number of equivalent rewritings of the query using the given views. As a result, for a particular quality metric on the rewritings being generated, one may achieve better quality of the outputs of the rewriting generator, with obvious practical advantages. Similarly, accounting for the existing dependencies in reformulating queries in a query optimizer could result in a larger space of equivalent reformulations. For an illustration, please see Example~\ref{motivating-example} in this paper.  


In the settings of 
query reformulation and view-based query rewriting in presence of dependencies, Deutsch and colleagues have developed an algorithm, called Chase and Backchase {\it (C}\&{\it B}, see~\cite{DeutschPT06}) that, for a given CQ query, outputs equivalent minimal-size CQ reformulations or rewritings of the query. 
The technical restriction on the 
 algorithm is the requirement that the process of  ``chasing'' (see~\cite{AbiteboulHV95} for an overview) the input query under the available dependencies terminate in finite time. Intuitively, the point of the chase in C\&B is to use the available dependencies to derive a new query formulation, which can be used to check ``dependency-aware'' equivalence of the query to any candidate reformulation or rewriting by using any known {\it dependency-free} equivalence test (e.g., that of~\cite{ChandraM77} for CQ queries). 
Under the above restriction, the C\&B algorithm is sound and complete for CQ queries, views, and rewritings/re-\linebreak formulations in presence of {\it embedded dependencies}, which are known to be sufficiently expressive to specify all usual integrity constraints, such as keys, foreign keys, inclusion, join, and multivalued dependencies~\cite{DeutschPods08}. 

The above guarantees of C\&B hold under {\em set semantics} for query evaluation, where both the database (stor-\linebreak ed) relations and query answers are treated as sets. Query answering and rewriting in the set-semantics setting have been studied extensively in the database-theory literature. At the same time, the set semantics are {\it not} the default query-evaluation semantics in database systems in practice. Specifically, the expected semantics of query evaluation in the standard query language SQL~\cite{GarciaMolinaUW02} are  {\it bag-set semantics.} That is, whenever a query does not use the {\tt DISTINCT} keyword, then query answers are treated in the SQL standard as multisets (i.e., sets with duplicates, also called {\it bags}), 
whereas the database relations are assumed to be sets. 

Arguably, the default semantics of SQL are the {\em bag semantics,} where both query answers and stored relations are permitted to be bags. Indeed, by the SQL standard stored relations are bags, rather than sets, whenever the {\tt PRIMARY KEY} and {\tt UNIQUE} clauses (which arise from the best practices but are not required in the SQL standard) are not part of the {\tt CREATE TABLE} statement. 
Using bag semantics in evaluating SQL queries becomes imperative in presence of materialized views~\cite{LevyAquvSurvey}, where the definitions of some of the views may not have included the {\tt DISTINCT} keyword, even assuming that all the original stored relations are required to be sets. 

The problem of developing tests for equivalence of CQ queries under bag and bag-set semantics was solved by Chaudhuri and Vardi in~\cite{VardiBagsPods93}.  
The bag-set-semantics test of~\cite{VardiBagsPods93} is also used in testing equivalence of queries with grouping and aggregation~\cite{CohenNS99,NuttSS98}. 
At the same time, developing tests for equivalence of CQ queries under bag or bag-set semantics in presence of  embedded dependencies has been an open problem until now. To the best of our knowledge, the only efforts in this direction have been undertaken by Deutsch in~\cite{DeutschDiss} and by Cohen in~\cite{Cohen06}, please see Section~\ref{related-work-section} for a more detailed discussion. Neither effort has resulted in equivalence tests for queries in presence of arbitrary sets of embedded dependencies, which may serve as an indication that the problem of  developing tests for equivalence of CQ queries under bag or bag-set semantics in presence of  embedded dependencies is not trivial.


\vspace{-0.2cm}

\paragraph{Our contributions} 

\noindent
 We consider the problem of finding equivalent minimal-size reformulations of SQL queries in presence of embedded dependencies, with a  focus on select-project-join queries with equality comparisons, also known as safe CQ queries, possibly with grouping and aggregation.  
To construct algorithms that would solve instances of this Query-Reformulation Problem (specified in Section~\ref{problem-stmt-sec}), we develop a comprehensive framework for equivalence of CQ queries under bag and bag-set semantics in presence of embedded dependencies, and make a number of conceptual and technical contributions. Specifically: 
\begin{itemize}
\vspace{-0.1cm}

	\item We formulate sufficient and necessary conditions for  correctness {\em (soundness)} of chase 
	for CQ queries and arbitrary sets of embedded dependencies under bag and bag-set semantics, see Section~\ref{new-sound-chase-sec}. 
\vspace{-0.2cm}

	\item It has been shown~\cite{DeutschDiss} that constraints that force some relations to be sets on all instances of a given database schema arise naturally in the context of sound chase under bag semantics. We develop a formal framework for defining such constraints as 
embedded dependencies, 
	provided that {\it row (tuple) IDs} (commonly used in commercial database-management systems) are defined for the respective relations. See  Section~\ref{new-sound-chase-sec} and Appendix~\ref{appendix-a}.  
\vspace{-0.2cm}

	\item We extend the condition of~\cite{VardiBagsPods93} for bag equivalence of CQ queries, to those cases where some relations are set valued in all instances of the given schema, see Section~\ref{new-sound-chase-sec}.   
	Our proof of this nontrivial result 
	includes reasoning involving bag (non)containment. In particular, we provide an original proof (adapted to our context) of the result of~\cite{VardiBagsPods93} that CQ query $Q_1$ is bag contained in CQ query $Q_2$ only if, for each predicate used in $Q_1$, $Q_2$ has at least as many subgoals with this predicate as $Q_1$ does. 
\vspace{-0.2cm}

	\item We show that the  result $Q_n$ of sound chase of a CQ query $Q$ using a finite set $\Sigma$ of embedded dependencies is unique under each of bag and bag-set semantics, whenever set-chase of $Q$ using $\Sigma$ terminates. We also provide a constructive characterization of the maximal subset of $\Sigma$ that is satisfied by the canonical database for $Q_n$. See Section~\ref{un-res-sec}. 

\vspace{-0.2cm}

	\item We provide equivalence tests for CQ queries in presence of embedded dependencies under bag and bag-set semantics, see Section~\ref{equiv-tests-he-subsection}. 
\vspace{-0.2cm}
	
	\item We present  equivalence tests for CQ queries {\it with grouping and aggregation} in presence of embedded dependencies, see Section~\ref{equiv-tests-third-subsection}. 
\vspace{-0.2cm}
	
	\item Finally, we develop sound and complete (whenever {\em set-}chase on the inputs terminates) algorithms for solving instances of the Query-Reformulation Problem with CQ queries under each of bag and bag-set semantics, as well as instances of the problem with aggregate queries, see Section~\ref{equiv-tests-second-subsection}. 
\vspace{-0.1cm}

\end{itemize}

Our contributions are clearly applicable beyond the Query-Reformulation Problem of Section~\ref{problem-stmt-sec}. Specifically, the results of this paper can be used in developing algorithms for rewriting CQ queries and queries in more expressive languages (e.g., including grouping and aggregation, or including arithmetic comparisons~\cite{Klug88}) using views in presence of embedded dependencies, under bag or bag-set semantics for query evaluation. Among other directions, our results could help solve the problem of reformulation for XQueries with bag semantics on XML data. Such queries can be explicitly written using the keyword {\tt unordered}, see~\cite{DeutschDiss} for a discussion.


\vspace{-0.3cm}

\section{Preliminaries}
\label{prelim-section}


\vspace{-0.2cm}

\subsection{The Basics} 
\label{basics-sec}

A database schema $\cal D$ is a finite set of relation symbols 
and their arities. A database (instance) $D$ over $\cal D$ has one finite relation for every 
relation symbol in  $\cal D$, of the same arity. A relation is, in general, {\it bag valued}; that is, it is a bag (also called {\it multiset}) of tuples. A bag can be thought of as a set of elements (the {\it core-set} of the bag) with multiplicities attached to each element. 
We say that a relation is {\it set valued} if its cardinality coincides with the cardinality of its core-set. 
A database instance is, in general, {\it bag valued.} We say that a (bag-valued) database instance is {\it set valued} if all its relations are set valued. 

A {\it conjunctive query (CQ query)} $Q$ over a schema $\cal D$ is an expression of the form $Q(\bar{X}) \ :- \ \phi(\bar{X}, \bar{Y}),$ where $\phi(\bar{X}, \bar{Y})$ is a nonempty conjunction of atomic formulas (i.e., relational atoms, also called {\it subgoals}) over $\cal D$. We follow the usual notation and separate the atoms 
in a query by commas. We call $Q(\bar{X})$ the {\it head} and $\phi(\bar{X}, \bar{Y})$  
the {\it body.} We use a notation such as $\bar{X}$ for a vector of $k$ variables and constants 
$X_1,\ldots,X_k$ (not necessarily distinct). 
Every variable in the head 
must appear in the body (i.e.,  $Q$ must be {\it safe}). The set 
of variables in $\bar{Y}$ is assumed to be existentially quantified.  

Given two conjunctions $\phi(\bar{U})$ and $\psi(\bar{V})$ of atomic formulas, 
a {\it homomorphism} from $\phi(\bar{U})$ to $\psi(\bar{V})$ is a mapping $h$ from the 
set of variables and constants in $\bar{U}$ to the set of variables and constants in $\bar{V}$ such that (1) $h(c) = c$ for each constant $c$, and (2) for every atom $r(U_1,\ldots,U_n)$ of $\phi$,  $r(h(U_1),\ldots,h(U_n))$ is in $\psi$. Given two CQ queries 
$Q_1(\bar{X}) \ :- \ \phi(\bar{X}, \bar{Y})$ and $Q_2(\bar{X}') \ :- \ \psi(\bar{X}', \bar{Y}'),$ a {\it containment mapping} from $Q_1$ to $Q_2$ is a homomorphism $h$ from $\phi(\bar{X}, \bar{Y})$ to $\psi(\bar{X}', \bar{Y}')$ such that $h(\bar{X}) = \bar{X}'$. 

For a conjunction $\phi(\bar{U})$ of atomic formulas, an {\it assignment} $\gamma$ for $\phi(\bar{U})$ is a  mapping of the variables of $\phi(\bar{U})$  to constants, and of the constants of $\phi(\bar{U})$ to themselves. 
We use a notation such as  $\gamma(\bar{X})$ to denote  tuple 
$(\gamma(X_1),\ldots,\gamma(X_k))$. Let relation $P_i$ in database $D$ correspond to predicate $p_i$. Then we say that {\it  atom $p_i(\bar{X})$ is satisfied  by 
assignment $\gamma$ w.r.t. database $D$} if there exists tuple $t \in P_i$ in $D$ such that $t = \gamma(\bar{X})$. Note that  the satisfying assignment 
 $\gamma$ is a homomorphism from $p_i(\bar{X})$ to the ground atom $p_i(\gamma(\bar{X}))$ representing tuple $t$ in $P_i$. 
Both the tuple-based definition of satisfaction and its homomorphism formulation are  naturally extended to define satisfaction of conjunctions of atoms. 


{\bf Query evaluation under set semantics.} For a CQ query  $Q(\bar{X}) \ :- \ \phi(\bar{X}, \bar{Y})$ and for a database $D$,  suppose that there exists an assignment $\gamma$ for the body $\phi(\bar{X}, \bar{Y})$ of $Q$, such that $\phi(\bar{X}, \bar{Y})$ is satisfied by $\gamma$ w.r.t. $D$. Then we say that {\it $Q$ returns a tuple $t = \gamma(\bar{X})$  on $D$.}  Further, the {\it answer $Q(D,S)$ to $Q$ on a set-valued database $D$ under set semantics for query evaluation} is the set of all tuples that $Q$ returns on $D$.


{\bf Query equivalence under set semantics.} Query $Q_1$ {\it is contained in query $Q_2$ under set semantics} ({\em set-contained,} denoted $Q_1 \sqsubseteq_S Q_2$) if $Q_1(D,S) \subseteq Q_2(D,S)$ for every set-valued database $D$.  Query $Q_1$ {\it is equivalent to query $Q_2$ under set semantics}  ({\em set-equivalent,} denoted $Q_1 \equiv_S Q_2$) if   $Q_1 \sqsubseteq_S Q_2$ and  $Q_2 \sqsubseteq_S Q_1$. A classical result~\cite{ChandraM77} states that a necessary and sufficient condition for the set-containment  $Q_1 \sqsubseteq_S Q_2,$ for CQ queries $Q_1$ and $Q_2$,   is the existence of a containment mapping from $Q_2$ to $Q_1.$ This result forms the basis for a sound and complete test for set-equivalence of CQ queries, by definition of set-equivalence.

{\bf Canonical database.} Every CQ query $Q$ can be regarded as a symbolic database $D^{(Q)}$. $D^{(Q)}$ is defined as the result of 
turning each subgoal $p_i(\ldots)$ of $Q$ into a tuple in the relation $P_i$ that corresponds to predicate $p_i$. The procedure is to keep each constant in the body of $Q$, and to replace consistently each variable in the body of $Q$ by a distinct constant different from all constants in  $Q$. The tuples that correspond to the resulting ground atoms are the only tuples in the {\it  canonical database} $D^{(Q)}$ for $Q$, which is unique up to isomorphism. 

\vspace{-0.1cm}

\subsection{Bag and Bag-Set Semantics}
\label{bag-bag-set-defs}

In this section we provide definitions for query evaluation under bag and bag-set semantics. Our definitions are consistent with the semantics of evaluating CQ queries in the SQL standard (see, e.g.,~\cite{GarciaMolinaUW02}), as well as with the corresponding definitions in~\cite{VardiBagsPods93,KolaitisPods06}. 

{\bf Query evaluation under bag-set semantics.} Consider a CQ query  $Q(\bar{X}) \ :- \ \phi(\bar{X}, \bar{Y})$. 
The {\it answer $Q(D,BS)$ to $Q$ on a set-valued database $D$ under bag-set semantics for query evaluation} is the bag of all tuples that $Q$ returns on $D$. That is, for each assignment $\gamma$ for the body $\phi(\bar{X}, \bar{Y})$ of $Q$, such that $\phi(\bar{X}, \bar{Y})$ is satisfied by $\gamma$ w.r.t. $D$, $\gamma$ contributes to the bag  $Q(D,BS)$ a {\it distinct}  tuple  $t = \gamma(\bar{X})$, such that $Q$ returns $t$ on $D$ w.r.t. $\gamma$. (I.e., whenever $Q$ returns $t_1$ on $D$ w.r.t. $\gamma_1$ and $Q$ returns a copy $t_2$ of $t_1$ on $D$ w.r.t. $\gamma_2 \neq \gamma_1$, then each of $t_1$ and $t_2$ is a separate element of the bag  $Q(D,BS)$.) 

{\bf Query evaluation under bag semantics.} For a CQ query  $Q$, the {\it answer $Q(D,B)$ to $Q$ on a bag-valued database $D$ under bag semantics for query evaluation} is a bag of tuples computed as 
follows. Suppose $Q$ is
\begin{tabbing}
$Q(\bar{X}) \ :- \ p_1(\bar{X}_1), p_2(\bar{X}_2),\ldots, p_n(\bar{X}_n).$  
\end{tabbing}
Consider the vector $p_1,\ldots,p_n$ of predicates (not necessarily distinct) occurring in the body of $Q$, and let  $P_1,\ldots,P_n$ be the vector of relations in $D$ such that each $p_i$ corresponds to relation $P_i$.  Whenever two subgoals $p_i(\ldots)$ and $p_j(\ldots)$ of $Q$, with $i \neq j$, have the same predicate, $P_i$ and $P_j$ refer to the same relation in $D$. 

Let $\gamma$ be an assignment for the body of $Q$, such that the body of $Q$ is satisfied by $\gamma$ w.r.t. $D$. Assignment $\gamma$ maps each subgoal $p_i(\bar{X}_i)$ of $Q$ into a tuple $t^{(i)}$ in relation $P_i$. For each $i \in \{ 1,\ldots,n \}$, let $m_i$ be the number of occurrences of tuple $t^{(i)}$ in the bag $P_i$. (I.e., $m_i > 0$ is the multiplicity associated with the (unique copy of) tuple $t^{(i)}$ in the core-set of $P_i$.) Then each distinct $\gamma$ contributes exactly $\Pi_{i = 1}^n m_i$ copies of tuple $t = \gamma(\bar{X})$ to  the bag $Q(D,B)$. (Recall that $\bar{X}$ is the vector of variables and constants in the head of $Q$.) Further, the bag $Q(D,B)$ has no other tuples.

\vspace{-0.1cm}

\subsection{Equivalence Tests for CQ Queries}
\label{bag-equiv-defs}

This subsection outlines equivalence tests for CQ\linebreak queries, for the cases of bag and bag-set semantics. The classical equivalence test~\cite{ChandraM77} for CQ queries for the case of set semantics is described in Section~\ref{basics-sec}. 

{\bf Query equivalence under bag and  bag-set semantics.} Query $Q_1$ {\it is equivalent to query $Q_2$ under bag semantics}  ({\em bag-equivalent,} denoted $Q_1 \equiv_B Q_2$) if for all bag-valued databases $D$ it holds that  $Q_1(D,B)$ and  $Q_2(D,B)$ are the same bags. Query $Q_1$ {\it is equivalent to query $Q_2$ under bag-set semantics}  ({\em bag-set-equivalent,} $Q_1 \equiv_{BS} Q_2$) if for all set-valued databases $D$ it holds that  $Q_1(D,BS)$ and  $Q_2(D,BS)$ are the same bags. 

\vspace{-0.1cm}

\begin{proposition}{\cite{VardiBagsPods93}}
\label{b-bs-s-implic-prop}
Given two CQ queries $Q_1$ and $Q_2$, $Q_1 \equiv_B Q_2$ implies  $Q_1 \equiv_{BS} Q_2$, and $Q_1 \equiv_{BS} Q_2$ implies  $Q_1 \equiv_{S} Q_2$.
\end{proposition}

\vspace{-0.1cm}

For bag and bag-set semantics, the following conditions are known for CQ query equivalence. (Query $Q_c$ is a {\em canonical 
representation} of query $Q$ if $Q_c$ is the result of removing all
duplicate atoms from $Q$.)  
\vspace{-0.1cm}

\begin{theorem}{~\cite{VardiBagsPods93}}
\label{cv-theorem} Let $Q$ and $Q'$ be CQ queries. Then (1)
$Q \equiv_B Q'$ iff
 $Q$ and $Q'$ are isomorphic. (2) $Q \equiv_{BS} Q'$ iff $Q_c \equiv_B Q'_c$, where $Q_c$ and
$Q'_c$ are canonical representations of $Q$ and $Q',$
respectively.
\end{theorem}

\vspace{-0.3cm}

\vspace{-0.1cm}

\subsection{Dependencies and Chase}
\label{chase-prelims}



{\bf Embedded dependencies.} We consider dependencies $\sigma$ of the form 
\vspace{-0.1cm}
\begin{tabbing}
$\sigma: \phi(\bar{U},\bar{W}) \rightarrow \exists \bar{V} \ \psi(\bar{U},\bar{V})$  
\end{tabbing}
\vspace{-0.1cm}
where $\phi$ and $\psi$ are conjunctions of atoms, which may include equations. Such dependencies,  called {\it embedded dependencies,} are sufficiently expressive to specify all usual integrity constraints, such as keys, foreign keys, inclusion, and join dependencies~\cite{DeutschPods08}. 
If $\psi$ consists only 
of equations, then $\sigma$ is an {\it equality-generating dependency (egd).} 
If $\psi$ consists only of relational atoms, then $\sigma$ is a {\it tuple-generating 
dependency (tgd).} Every set $\Sigma$ of embedded dependencies is equivalent to a set of tgds and egds~\cite{AbiteboulHV95}. We write $D \models \Sigma$ if database $D$ satisfies all the dependencies in $\Sigma$. All sets $\Sigma$ we refer to are finite. 

{\bf Query containment and equivalence under dependencies.} We say that query $Q$ is {\it set-equivalent to} query $P$ {\it under a set of dependencies} $\Sigma$, denoted $Q \equiv_{\Sigma,S} P,$ if for every set-valued database $D$ such that $D \models \Sigma$ we have $Q(D,S) = P(D,S)$. The definition of {\it set containment under dependencies,} denoted $\sqsubseteq_{\Sigma,S}$, as well as
the definitions of {\it bag equivalence and bag-set equivalence under dependencies} (denoted by $\equiv_{\Sigma,B}$ and $\equiv_{\Sigma,BS}$ \ , respectively), are analogous modifications of the respective definitions for the dependency-free setting, see Sections~\ref{basics-sec} and~\ref{bag-equiv-defs}. 

{\bf Chase.} Assume a CQ query $Q(\bar{X}) \  :- \ \xi(\bar{X},\bar{Y})$ and a tgd $\sigma$  
of the form $\phi(\bar{U},\bar{W}) \rightarrow \exists \bar{V} \ \psi(\bar{U},\bar{V})$. Assume w.l.o.g. 
that $Q$ has none of the variables $\bar{V}$. The {\it chase of $Q$ with $\sigma$ is applicable} if there is a homomorphism $h$ from $\phi$ to $\xi$ 
and if, moreover, $h$ cannot be extended to a homomorphism $h'$ from  $\phi \wedge \psi$ to $\xi$.  In that case, a {\it chase step} of $Q$ with $\sigma$ and $h$ is a rewrite of $Q$ into $Q'(\bar{X}) \  :- \ \xi(\bar{X},\bar{Y}) \wedge \psi (h(\bar{U}),\bar{V})$. 

We now define a chase step with an egd. Assume a CQ query $Q$ as before and an egd $e$ of the form $\phi(\bar{U}) \rightarrow U_1 = U_2.$ The {\it chase of $Q$ with $e$ is applicable} if there is a homomorphism $h$ from $\phi$ to $\xi$ such that $h(U_1) \neq h(U_2)$ and at least one of $h(U_1)$ and $h(U_2)$ is a variable; assume w.l.o.g. that $h(U_1)$ is a variable. Then a {\it chase step}  
of $Q$ with $e$ and $h$ is a rewrite of $Q$ into a query that results from replacing all occurrences of $h(U_1)$ in $Q$ by $h(U_2)$. 

A {\it $\Sigma$-chase sequence} $C$ (or just {\it chase sequence,} if $\Sigma$ is clear from the context) 
is a sequence of CQ queries $Q_0, Q_1, \ldots$ such that every query $Q_{i+1}$ ($i \geq 0$) in $C$ is obtained from $Q_i$ by a chase step $Q_i \Rightarrow^{\sigma} Q_{i+1}$ using a dependency $\sigma \in \Sigma$. 
A chase sequence $Q = Q_0, Q_1, \ldots, Q_n$ is {\it terminating under set semantics} if $D^{(Q_n)} \models \Sigma$, 
where  $D^{(Q_n)}$ is the canonical database for $Q_n$. In 
this case we say that $(Q)_{\Sigma,S} = Q_n$ is the (terminal) {\it result} of the chase. Chase of CQ queries under set semantics is known to terminate in finite time for a class of embedded dependencies called {\it weakly acyclic dependencies}, see~\cite{FaginKMP05} and references therein. Under set semantics, all chase results for a given CQ query are equivalent in the absence of dependencies~\cite{DeutschPods08}. 


The following result is immediate from~\cite{AbiteboulHV95,DeutschDiss,DeutschPods08}.  
\vspace{-0.1cm}

\begin{theorem}
\label{chase-theorem}
Given CQ queries $Q_1$, $Q_2$ and set $\Sigma$ of embedded dependencies. Then 
	$Q_1 \equiv_{\Sigma,S} Q_2$ iff $(Q_1)_{\Sigma,S} \equiv_S (Q_2)_{\Sigma,S}$ in the absence of dependencies.
\end{theorem}

\vspace{-0.1cm}


\nop{
{\bf Superkeys and keys of relations. Notions that I need to define:

1. Superkey (for Example 4.1)

2. Fd (for explanation under Def 4.1)

3. Keys!

} 
} 


\subsection{Queries with Grouping and Aggregation}
\label{aggr-prelims}

We assume that the data we want to aggregate are real numbers,
 {\bf R}. If $S$ is a set, then ${\cal M}(S)$ denotes the set of finite bags over $S$. A {\it $k$-ary aggregate function} is a function
 $\alpha: \ {\cal M}({\bf R}^{k}) \ \rightarrow \ {\bf R}$ that maps bags 
 of $k$-tuples of real numbers to real numbers.
An {\it aggregate term} is an expression built up using an aggregate function over variables. Every
aggregate term with $k$ variables gives rise to a $k$-ary aggregate function in a natural way.

We use
 $\alpha(y)$ as an
abstract notation for a unary aggregate term, where  $y$ is the
variable in the term.
Aggregate queries that we
 consider have (unary or $0$-ary) aggregate functions $count$, $count(*)$, $sum$, $max$, and  $min$. Note that $count$ is over an argument, whereas $count(*)$
is the only function that we consider here that takes no argument.
(There is a distinction in SQL semantics between $count$ and $count(*)$.)
In the rest of the paper, we will not refer again to the distinction between $count$ and $count(*)$,
as our results carry over.


An {\it aggregate query}~\cite{CohenNS99,NuttSS98} is a conjunctive query augmented
by an aggregate term in its head. For a query with a $k$-ary aggregate function $\alpha$, the syntax is:
\vspace{-0.1cm}
\begin{equation}
\label{aggr-equation}
Q(\bar{S},\alpha(\bar{Y})) \leftarrow A(\bar{S}, \bar{Y}, \bar{Z}) \ .
\end{equation}
\vspace{-0.1cm}
\noindent
$A$ is a conjunction of atoms; 
 $\alpha(\bar{Y})$ is a $k$-ary aggregate term;
 $\bar{S}$ are the {\em
    grouping attributes} of $Q$;
 none of the variables in $\bar{Y}$ appears in $\bar{S}$.
Finally, $Q$ is {\it safe}: all variables in $\bar{S}$ and $\bar{Y}$ occur in $A$.
We consider queries with unary aggregate functions $sum$, $count$, $max$, and $min$. With each aggregate query $Q$ as in Equation (\ref{aggr-equation}), we associate its CQ {\it core} $\breve{Q}$:
$\breve{Q}(\bar{S},\bar{Y}) \leftarrow A(\bar{S}, \bar{Y}, \bar{Z}) .$ 

We define the semantics of an aggregate query as follows:
Let  $D$ be a set-valued database and  $Q$ an  aggregate query as in
Equation (\ref{aggr-equation}). When $Q$ is applied on  $D$ it  yields
a relation $Q({D})$ defined by the following three  steps:
First, we compute the bag ${\bf B} = \breve{Q}(D,BS)$ on $D$. We then form equivalence classes in ${\bf B}$: Two tuples belong to the same equivalence
class if they agree on the values of all the grouping arguments of $Q$. This is the {\em grouping} step.
The third step is {\em aggregation}; it associates with each equivalence
class a value that is the aggregate function computed on a bag that contains
all values of the input argument(s) of the aggregated attribute(s) in this class. For each class, it
returns one tuple, which contains the values of the grouping arguments of $Q$
and the computed aggregated value.


In general,  
queries with different aggregate functions may be equivalent~\cite{CohenNS99}. We follow the approach of~\cite{CohenNS99,NuttSS98} by considering equivalence between queries with the same lists of head arguments, called {\em compatible queries}. 
\vspace{-0.5cm}
\begin{definition}{Equivalence of compatible aggregate queries~\cite{NuttSS98}}
For queries $Q(\bar{X},\alpha(\bar{Y})) \leftarrow A(\bar{S})$ and $Q'(\bar{X},\alpha(\bar{Y})) \leftarrow A'(\bar{S}')$,
$Q \equiv Q'$ if $Q({D}) = Q'({D})$ for every database $D$.
\end{definition} 
\vspace{-0.2cm}

We say that two compatible aggregate queries $Q$ and $Q'$ are {\em equivalent in presence of a set of dependencies} $\Sigma$, $Q \equiv_{\Sigma} Q'$, if  $Q({D}) = Q'({D})$ for every database $D \models \Sigma$. 



\vspace{-0.1cm}
\begin{theorem}{\cite{CohenNS99,NuttSS98}}
\label{equival-aggr-theorem}
(1) Equivalence of $sum$- and of $count$-queries can be
reduced to bag-set equivalence of their cores. (2) Equivalence of
$max$- and of $min$-queries can be reduced to set equivalence of
their cores. 
\end{theorem}
\vspace{-0.1cm}




\vspace{-0.3cm}

\section{Problem Statement}
\label{problem-stmt-sec}

In this section we use the following notation: Let $X$ be the semantics  for query evaluation, with values $S$, $B$, and $BS$, for set, bag, or bag-set semantics, respectively. Let ${\cal L}_1$ and ${\cal L}_2$ 
be two query languages. 
Let $\Sigma$ be a finite set of dependencies on database schema $\cal D$. 

We use the notion of $\Sigma$-minimality~\cite{DeutschPT06}, defined as follows. (Intuitively, reformulation $R$ of query $Q$ 
is not $\Sigma$-minimal if at least one egd in $\Sigma$ is applicable to  $R$.)
\vspace{-0.2cm}

\begin{definition}{Minimality under dependencies~\cite{DeutschPT06}}
A CQ query $Q$ is {\em $\Sigma$-minimal} if there are no queries $S_1$, $S_2$ where $S_1$ is obtained from $Q$ by replacing zero or more variables with other variables of $Q$, and $S_2$ by dropping at least one atom from $S_1$ such that $S_1$ and $S_2$ remain equivalent to $Q$ under $\Sigma$.
\end{definition}
\vspace{-0.2cm}

We extend this definition to $\Sigma$-minimality of CQ\linebreak queries with grouping and aggregation, which is defined as $\Sigma$-minimality of the (unaggregated) core of the query, see Section~\ref{aggr-prelims} for the relevant definitions. 
\nop{

For the purposes of Section~\ref{}, we extend the definition of $\Sigma$-minimality (Section~\ref{c-and-b-section}) naturally to $\Sigma$-minimality of conjunctive queries with arithmetic comparisons (CQAC queries),  which is defined via $\Sigma$-minimality of the CQ part of the query (i.e., of the result of dropping all the inequality/nonequality ACs of the query). \reminder{Remove this preceding CQAC sentence?}
Similarly, we extend the definition to $\Sigma$-minimality of CQ(AC) queries with grouping and aggregation, which is defined via $\Sigma$-minimality of the CQ part of the (unaggregated) core of the query, see Section~\ref{aggr-prelim-section} for the relevant definitions. 

} 

A general statement of {\bf the Query-Reformulation Problem} that we consider in this paper is as follows: The {\it problem input} is $({\cal D}, X, Q, \Sigma, {\cal L}_2)$, where query $Q$ is defined on database schema $\cal D$ in language ${\cal L}_1$. A {\it solution} to the Query-Reformulation Problem, for a problem input $({\cal D}, X, Q, \Sigma, {\cal L}_2)$, is a query $Q'$ defined in language ${\cal L}_2$ on $\cal D$, such that $Q' \equiv_{\Sigma,X} Q$. 


In this paper we consider the Query-Reformulation Problem in presence of embedded dependencies, and focus on (1) the 
{\it CQ class} of the problem, where each of ${\cal L}_1$ and ${\cal L}_2$ is the language of CQ queries, and on (2) the {\it CQ-aggregate class} (see Section~\ref{equiv-tests-second-subsection}), where each of ${\cal L}_1$ and ${\cal L}_2$ is the language of CQ queries with grouping and aggregation, using aggregate functions $sum$, $max$, $min$, and $count$; we refer to this query language as {\it CQ-aggregate}. 
For both classes, we consider only {\it $\Sigma$-minimal solutions} of the Query-Reformulation Problem.


\vspace{-0.1cm}

\section{Sound Chase Under Bag and Bag-Set Semantics}
\label{new-sound-chase-sec}

In this section we 
 show that  under bag and bag-set semantics, it is incorrect to enforce the set-semantics condition of $D^{(Q_n)} \models \Sigma$ (Section~\ref{chase-prelims}) on the terminal chase result $Q_n$ of query $Q$ under dependencies $\Sigma$. The problem is that under this condition, chase may yield a result $Q_n$ that is {\it not} equivalent to the original query $Q$ in presence of $\Sigma$. That is, soundness of chase, understood as $Q_n \equiv_{\Sigma,B} Q$ or $Q_n \equiv_{\Sigma,BS} Q$, may {\it not} hold. We then formulate sufficient and necessary conditions for soundness of chase for CQ queries and embedded dependencies under bag and bag-set semantics. 

In this section we also show that constraints that force certain relations to be sets on all instances of a given database schema can be defined as egds, provided that {\it row (tuple) IDs} are defined for the respective relations. Finally, we extend the condition of Theorem~\ref{cv-theorem} for bag equivalence of CQ queries, to those cases where some relations are required to be set valued in all instances of the given schema. Such requirements can be defined as our set-enforcing egds.  

	
\subsection{Motivating Example}

	Let us conjecture that maybe an analog of Theorem~\ref{chase-theorem} (Section~\ref{chase-prelims}) holds for the case of bag semantics. (In this section we discuss in detail the case of  bag semantics only; analogous reasoning is valid for the case of bag-set semantics.) That is, maybe $Q_1 \equiv_{\Sigma,B} Q_2$ if and only if $(Q_1)_{\Sigma,S} \ \equiv_B (Q_2)_{\Sigma,S}$ in the absence of dependencies, 
	 for a given pair of CQ queries $Q_1$ and $Q_2$ and for a given set $\Sigma$ of embedded dependencies. (We obtain our conjecture by replacing the symbols $\equiv_{\Sigma,S}$ and $\equiv_{S}$ in Theorem~\ref{chase-theorem} by the {\it bag-}semantics versions of these symbols.)
	
	Now consider the C\&B algorithm by Deutsch and colleagues~\cite{DeutschPT06}. 
Under set semantics for query evaluation and  given a CQ query $Q$, C\&B  outputs all equivalent  $\Sigma$-minimal conjunctive reformulations of $Q$ in presence of the given embedded dependencies $\Sigma$ (i.e., C\&B is {\it sound and complete}), whenever chase of $Q$ under $\Sigma$ terminates in finite time. 
See Appendix~\ref{c-and-b-appendix} for the details on C\&B.

	If our conjecture is valid, then a straightforward modification of C\&B gives us a procedure for solving  
	 instances in the CQ class of the Query-Reformulation Problem  for  bag  semantics.\footnote{An analogous extension of C\&B would work for instances $({\cal D}, BS, Q, \Sigma, CQ)$, i.e., under bag-set semantics.}  The only difference between the original C\&B and its proposed modification would be the test for {\it bag}, rather than {\it set}, equivalence (see Theorem~\ref{cv-theorem}) between the universal plan $(Q)_{\Sigma,S}$ of C\&B  for the input query $Q$ and dependencies $\Sigma$, and the terminal result of chasing a candidate reformulation. (These terms are defined in Appendix~\ref{c-and-b-appendix}.) By extension from C\&B, our algorithm would be sound and complete for all problem instances where the universal plan for $Q$ could be computed in finite time. 
	
Unfortunately, this naive extension of C\&B would not be sound for bag semantics (or for bag-set semantics, in the version of C\&B using the bag-set equivalence test of Thm.~\ref{cv-theorem}). We highlight the problems in an example.
\begin{example}
\label{motivating-example}
On database schema ${\cal D} = \{ P, R, S,$ $T, U \}$, consider a set $\Sigma$ that includes four tgds:
\vspace{-0.1cm}

\begin{tabbing}
$\sigma_1: p(X,Y) \rightarrow s(X,Z) \wedge t(X,V,W)$ \\
$\sigma_2: p(X,Y) \rightarrow t(X,Y,W)$ \\
$\sigma_3: p(X,Y) \rightarrow r(X)$ \\
$\sigma_4: p(X,Y) \rightarrow u(X,Z) \wedge t(X,Y,W)$ 
\end{tabbing}
\vspace{-0.1cm}

Suppose $\Sigma$ also includes dependencies enforcing the following constraints: (1) Relations $S$ and $T$ (but not $R$ or $U$) are set valued in all instances of $\cal D$; call these constraints $\sigma_5$ and $\sigma_6$, respectively. (These dependencies are relevant to the bag-semantics case. Under bag-set or set semantics, all relations in all instances of $\cal D$ are set valued by definition.) Please see Section~\ref{making-chase-sound-section} for an approach to expressing such   constraints using egds. (2) The first attribute of $S$ is the key of $S$ (egd $\sigma_7$), and the first two attributes of $T$ are the key of $T$ (egd $\sigma_8$), see Appendix~\ref{key-app} for the definition of keys. 

Consider CQ queries $Q_1$ through $Q_4$, defined as
\vspace{-0.1cm}

\begin{tabbing}
$Q_1(X) \ :- \ p(X,Y), t(X,Y,W), s(X,Z), r(X), u(X,U).$ \\ 
$Q_2(X) \ :- \ p(X,Y), t(X,Y,W), s(X,Z), r(X).$ \\
$Q_3(X) \ :- \ p(X,Y), t(X,Y,W), s(X,Z).$ \\
$Q_4(X) \ :- \ p(X,Y).$ 
\end{tabbing}
\vspace{-0.1cm}
(We disregard queries $Q_2$ and $Q_3$ for the moment.) 

We can show  that $Q_1 \equiv_{\Sigma,S} Q_4$. Thus, $Q_1$ is a reformulation of $Q_4$ in presence of $\Sigma$ under set semantics. At the same time, by~\cite{ChandraM77}  $Q_1$ and $Q_4$ are not equivalent under set semantics in the absence of dependencies.

Our naive modification of C{\em \&}B would return a reformulation $Q_1$ of query $Q_4$. Indeed, each of $(Q_1)_{\Sigma,S}$ and $(Q_4)_{\Sigma,S}$ is isomorphic to $Q_1$, thus by Theorem~\ref{cv-theorem} we have that 
$(Q_1)_{\Sigma,S} \ \equiv_B (Q_4)_{\Sigma,S}$. 

However, 
 even though $(Q_1)_{\Sigma,S} \ \equiv_B (Q_4)_{\Sigma,S}$, it is not true that $Q_1 \ \equiv_{\Sigma,B} Q_4$. The counterexample is a bag-valued database $D$, $D \models \Sigma$, with relations $P = \{\hspace{-0.1cm}\{ (1,2) \}\hspace{-0.1cm}\}$, $R = \{\hspace{-0.1cm}\{ (1) \}\hspace{-0.1cm}\}$, $S = \{\hspace{-0.1cm}\{ (1,3) \}\hspace{-0.1cm}\}$, $T = \{\hspace{-0.1cm}\{ (1,2,4) \}\hspace{-0.1cm}\}$, and $U = \{\hspace{-0.1cm}\{ (1,5), (1,6) \}\hspace{-0.1cm}\}$. On the database $D$, the answer to $Q_4$ under bag semantics is $Q_4(D,B) = \{\hspace{-0.1cm}\{ (1) \}\hspace{-0.1cm}\}$, whereas $Q_1(D,B) = (Q_1)_{\Sigma,S}(D,B) = (Q_4)_{\Sigma,S}(D,B)  = \{\hspace{-0.1cm}\{ (1),$ $ (1) \}\hspace{-0.1cm}\}$. From the fact that $Q_1(D,B)$ and $Q_4(D,B)$ are not the same {\em bags,} we conclude that $Q_1 \ \equiv_{\Sigma,B} \hspace{-0.85cm} / \hspace{0.75cm}  Q_4$.

The same database $D$ (which is {\em set} valued) would disprove $Q_1 \ \equiv_{\Sigma,BS} Q_4$ (i.e., equivalence of $Q_1$ and $Q_4$ under $\Sigma$ and {\em bag-set} semantics), even though it is true by Theorem~\ref{cv-theorem} that $(Q_1)_{\Sigma,S} \ \equiv_{BS} (Q_4)_{\Sigma,S}$.
\end{example}

\subsection{Sound Chase Steps}
\label{making-chase-sound-section}

The problem highlighted in Example~\ref{motivating-example} is unsoundness of {\it set}-semantics chase when applied to query $Q_4$ under bag or bag-set semantics. To rectify this problem, that is to make chase sound under these semantics, we modify the definitions of chase steps.

Given a CQ query $Q$ and a set of embedded dependencies $\Sigma$, let $Q'$ be the result of applying to query $Q$ a dependency $\sigma \in \Sigma$. We say that {\it the chase step $Q \Rightarrow^{\sigma}_B Q'$  is sound under bag semantics}~\cite{DeutschDiss} {\it ($Q \Rightarrow^{\sigma}_{BS} Q'$ is sound under bag-set semantics, respectively)} if it holds that $Q \equiv_{\Sigma,B} Q'$ (that $Q \equiv_{\Sigma,BS} Q'$, respectively). By extension of the above definitions,  all chase steps under embedded dependencies are sound under {\em set} semantics. The definitions of sound chase steps are naturally extended to those of {\em sound chase sequences} under each semantics. We say that a {\em chase result $Q_n$ is sound w.r.t. $(Q,\Sigma)$ under bag semantics (under bag-set semantics, respectively)} whenever there exists a $\Sigma$-chase sequence $C$ that starts with the input query $Q$ and ends with $Q_n$, and such that all chase steps in $C$ are sound under bag semantics (under bag-set semantics, respectively). 

\subsubsection{Regularized Assignment-Fixing Tgds}
\label{regularized-sec}

Toward ensuring soundness of chase under bag and bag-set semantics, we will define key-based chase using tgds, see Section~\ref{key-based-chase-subsec}. For our definition we will need the technical notions of ``regularized tgds'' and ``assign-\linebreak ment-fixing tgds'', which we formally define and characterize in this subsection. 

\paragraph{Regularized tgds} 
Consider a tgd $\sigma: \phi(\bar{X},\bar{Y}) \rightarrow \exists \bar{Z} \ \psi(\bar{X},\bar{Z})$ whose right-hand side $\psi$ has at least two relational atoms. Let $\psi_a$ and $\psi_b$ be a partition of $\psi$ (where $\psi$ is viewed as set of relational atoms) in $\sigma$ into two disjoint nonempty sets, that is $\psi_a \neq \emptyset$, $\psi_b \neq \emptyset$,  $\psi_a \cap \psi_b = \emptyset$, and $\psi_a \cup \psi_b = \psi$. Let $\bar{A}$ be all the variables in $\psi_a$, and let $\bar{B}$ be all the variables in $\psi_b$. We call $\psi_a$ and $\psi_b$ a {\em nonshared partition of} $\psi$ in $\sigma$ whenever $\bar{A} \cap \bar{B} \subseteq \bar{X}$. (Recall that all the variables in $\bar{X}$ are universally quantified in $\sigma$.) In case where $\psi_a$ and $\psi_b$ are two disjoint nonempty sets such that $\psi_a \cup \psi_b = \psi$ and $\bar{A} \cap \bar{B} \cap \bar{Z} \neq \emptyset$, 
 we call  $\psi_a$ and $\psi_b$ a {\em shared partition of} $\psi$ in $\sigma$. 

\begin{definition}{Regularized tgd, regularized\\ set of embedded dependencies}
\label{regularized-def}
A tgd $\sigma: \phi \rightarrow \psi$ 
is a {\em regularized tgd} if there exists no nonshared partition of the set of relational atoms of $\psi$ into two disjoint nonempty sets.\footnote{Trivially, every tgd whose right-hand side has exactly one atom is a regularized tgd.}  
We say that a finite set $\Sigma$ of embedded  dependencies is a {\em regularized set of (embedded) dependencies} if each tgd in $\Sigma$ is regularized. 
\end{definition} 

Sets $\{ u(X,Z) \}$ and $\{ t(X,Y,W) \}$ comprise a nonshared partition of the right-hand side of tgd $\sigma_4$ in Example~\ref{motivating-example}; therefore, the tgd $\sigma_4$ is not regularized. For a tgd $\sigma_1$ in Example~\ref{unregularized-counterex-one}, where $\sigma_1: \ p(X,Y) \rightarrow \exists Z \ \exists W \ r(X,Z) \wedge s(Z,W)$, sets $\{ r(X,Z) \}$ and $\{ s(Z,W) \}$ comprise a shared partition of the right-hand side of $\sigma_1$, because an existential variable $Z$ of $\sigma_1$ is present in both elements of the partition. This tgd is regularized by Definition~\ref{regularized-def}. The set $\Sigma$ in Example~\ref{regul-partapply-ex} is a regularized set of dependencies. 

Consider a tgd $\sigma: \phi \rightarrow \psi$ that is not regularized by Definition~\ref{regularized-def}. The process of  {\em regularizing} $\sigma$ is the process of constructing from $\sigma$  a set $\Sigma_{\sigma} = \{ \sigma_1,\ldots,\sigma_k \}$ of tgds, where $k \geq 2$ and such that for each tgd $\sigma_i$ in $\Sigma_{\sigma}$, (i) the left-hand side of $\sigma_i$ is the left-hand side $\phi$ of $\sigma$; (ii)  the right-hand side of $\sigma_i$ is a nonempty set of atoms $\psi_i \subseteq \psi$ (recall that $\psi$ is the right-hand side of $\sigma$), with all the existential variables (of $\psi$) in $\psi_i$ marked as such in $\sigma_i$; (iii) $\sigma_i$ is regularized by Definition~\ref{regularized-def}; and (iv) $\cup_{i=1}^k \psi_i = \psi$ . It is easy to see that given a non-regularized tgd $\sigma$, the recursive algorithm of  finding nonshared partitions 
of the right-hand side of $\sigma$ (a) regularizes $\sigma$ correctly, (b) results in a unique set $\Sigma_{\sigma}$, and (c) has the complexity $O(m^2 \ log \ m)$, where $m$ is the number of relational atoms in the right-hand side of $\sigma$. (The idea of the algorithm is to (1) give a unique ID $id(a_{\psi})$ to each relational atom $a_{\psi}$ of $\psi$, to then (2) associate with each $id(a_{\psi})$ the set of all those variables of $a_{\psi}$ that are existentially quantified in $\sigma$, and to then (3) recursively sort all the ids, each time by one fixed variable in their associated variable lists, and to either start a new nonshared partition using the sorted list, or to add atoms to an existing nonshared partition, again using the sorted list.) We call $\Sigma_{\sigma}$ the {\em regularized set} of $\sigma$.  

Now given a finite set $\Sigma$ of arbitrary embedded egds and tgds, we {\em regularize} $\Sigma$  by regularizing each tgd in $\Sigma$ as described above. We say that $\Sigma'$ is a {\em regularized version of} $\Sigma$ if (i) for each egd $\sigma$ in $\Sigma$, $\Sigma'$ also has $\sigma$, (ii) for each tgd $\sigma$ in $\Sigma$, $\Sigma'$ has the regularized set of $\sigma$ and, finally, (iii) $\Sigma'$ has no other dependencies. For each $\Sigma$ as above, it is easy to see that $\Sigma'$ is regularized by Definition~\ref{regularized-def} and is unique. 

The following results, in Proposition~\ref{regulariz-thm}, are immediate from Definition~\ref{regularized-def} and from the constructions in this subsection. 

\begin{proposition}
\label{regulariz-thm}
For a finite set $\Sigma$ of embedded egds and tgds defined on schema $\cal D$, let $\Sigma'$ be the regularized version of $\Sigma$. Then 
\begin{itemize}
	\item For every bag-valued database $D$ with schema $\cal D$, $D \models \Sigma$ iff $D \models \Sigma'$; and
	\item For every CQ query $Q$ defined on $\cal D$, chase of $Q$ under set semantics in presence of $\Sigma$ terminates in finite time iff chase of $Q$ under set semantics in presence of $\Sigma'$ terminates in finite time, and  $(Q)_{\Sigma,S} \equiv_S (Q)_{\Sigma',S}$ provided both chase results exist. 
\end{itemize}
\vspace{-0.6cm}
\end{proposition}

\paragraph{Assignment-fixing tgds} 
{\em In the remainder of the paper, whenever we refer to a set of embedded dependencies, we assume that we are discussing (or using) its {\em regularized} version.} We now define assignment-fixing tgds. The idea is to be able to determine easily which tgds ensure sound chase steps under each of bag and bag-set semantics. The intuition is as follows. Suppose chase with tgd $\sigma$ is applicable to a CQ query $Q$ as  defined in Section~\ref{chase-prelims} (i.e., assuming set semantics), but we are looking at the implications of applying the chase under bag or bag-set semantics rather than under set semantics. Suppose further that the right-hand side of $\sigma$ has existential variables. Then we would like to add subgoals to $Q$, that is to perform on $Q$ the chase step $Q \Rightarrow^{\sigma} Q'$, exactly in those cases where each consistent assignment to all body variables of $Q$, w.r.t. any (arbitrary) database $D$ that satisfies the input dependencies, can be extended to {\em one and only one} consistent assignment to all body variables of $Q'$ w.r.t. $D$. Otherwise $Q$ would not be equivalent to $Q'$ in presence of $\sigma$ and under the chosen query-evaluation semantics. It turns out that the characterization we are seeking is, in general, {\em query dependent}. (See Examples~\ref{unregularized-counterex-two} and~\ref{query-dep-ex}.)

We now formalize this intuition of prohibiting, in chase, ``incorrect'' multiplicity of the answer to the given query in presence of the given dependencies, under bag or bag-set semantics. Consider a CQ query $Q(\bar{A}) \ :- \ \zeta(\bar{A},\bar{B})$, 
 and a regularized tgd $\sigma: \phi(\bar{X},\bar{Y}) \rightarrow \exists \bar{Z} \ \psi(\bar{X},\bar{Z})$ that has at least one existential variable, that is $\bar{Z}$ is not empty. Suppose that chase of $Q$ with $\sigma$ is applicable, using a homomorphism $h$ from $\phi$ to $\zeta$. We come up with a substitution $\theta$ of all existential variables $\bar{Z}$ in the right-hand side $\psi$ of the tgd $\sigma$, such that $\theta$ replaces each variable in $\bar{Z}$ by a fresh variable that is not used in any capacity (i.e., neither universally nor existentially quantified) in $\sigma$ or in $\zeta$. (Observe that $\theta$ always exists.) 
 We use $h$ and $\theta$ to define for $Q$ and $\sigma$ an {\em associated test query} $Q^{\sigma,h,\theta}$: 
 

\begin{equation}
\label{theta-eq}
Q^{\sigma,h,\theta}(\bar{A}) \ :- \ \zeta(\bar{A},\bar{B}) \wedge \psi(h(\bar{X}),\bar{Z}) \wedge \psi(h(\bar{X}),\theta(\bar{Z})) \ . 
\end{equation}

Observe that for any pair $(\theta_1, \theta_2)$ of substitutions that satisfy the conditions on $\theta$ above, $Q^{\sigma,h,\theta_1}$ and $Q^{\sigma,h,\theta_2}$ are isomorphic. Hence $Q^{\sigma,h,\theta}$ is unique up to isomorphism w.r.t. $\theta$, and we choose one arbitrary $\theta$ for $Q^{\sigma,h,\theta}$ 
in the remainder of the paper.  

We now treat the case where 
$\sigma$ has no existential variables. In this case $\theta$ is trivially empty, $\theta = \emptyset$, and we define the associated test query $Q^{\sigma,h,\emptyset}$ for $Q$, $\sigma$, and $h$ as above as: 

\begin{equation}
\label{theta-noex-eq}
Q^{\sigma,h,\emptyset}(\bar{A}) :- \zeta(\bar{A},\bar{B}) \wedge \psi(h(\bar{X}),\bar{Z}) \ . 
\end{equation}

That is, $Q^{\sigma,h,\emptyset}$ is the result of applying to the query $Q$ a chase step using $\sigma$, as defined in Section~\ref{chase-prelims} in this paper. 

We stress again that Equation~\ref{theta-noex-eq} is defined only for those cases where $\sigma$ has no existential variables. However, Equation~\ref{theta-noex-eq} can be obtained from Equation~\ref{theta-eq} by setting $\theta = \emptyset$ and by removing duplicate subgoals from the body of the query in Equation~\ref{theta-eq}. Therefore, in what follows we adopt Equation~\ref{theta-eq} as the definition of the associated test query for $Q$ and $\sigma$ regardless of whether $\sigma$ has existential variables. 

\begin{definition}{Associated test query}
\label{asso-test-def}
Given a CQ query $Q$ and a regularized tgd $\sigma$ such that chase using $\sigma$ is applicable to $Q$ using homomorphism $h$,  the {\em associated test query for} $Q$, $\sigma$, {\em and} $h$ is as shown in Equation~\ref{theta-eq}. 
\end{definition}

We now define assignment-fixing tgds, which enable sound chase steps under each of bag and bag-set semantics, under an extra condition under bag semantics that all the subgoals being added in the chase step correspond to set-valued relations. We first ensure correctness of the definition of assignment-fixing tgds, by making a straightforward observation. 

\begin{proposition}
\label{chase-assoc-term-prop}
Given a CQ query $Q$ and a finite set $\Sigma$ of tgds and egds, and for a regularized tgd\footnote{Recall that we assume throughout the paper that $\Sigma$ is the regularized version of any given set of tgds and egds.}  $\sigma \in \Sigma$ such that chase using $\sigma$ applies to $Q$ with a homomorphism $h$. Then the terminal chase result $(Q^{\sigma,h,\theta})_{\Sigma,S}$ exists whenever $(Q)_{\Sigma,S}$ exists. 
\end{proposition} 

\begin{proof}{(Sketch.)} 
Trivial for the case where $\sigma$ has no existential variables. For the remaining case, the proof is by contradiction. Suppose that $(Q^{\sigma,h,\theta})_{\Sigma,S}$ does not exist, that is, the body of $(Q^{\sigma,h,\theta})_{\Sigma,S}$ has an infinite number of relational subgoals, using an infinite number of variable names. We then show that the body of $(Q)_{\Sigma,S}$ also has an infinite number of relational subgoals (using an infinite number of variable names), and thus arrive at the desired contradiction. The procedure is to apply to $Q$ all the chase steps that are applicable to $Q^{\sigma,h,\theta}$. Specifically, for each chase step $S$ that applies on $Q^{\sigma,h,\theta}$ using a homomorphism $\mu$, we apply the same chase step to {\em the result $Q'$ of chase step on $Q$ using $\sigma$ and the $h$ of} $Q^{\sigma,h,\theta}$. In each $S$ we use the homomorphism that is a composition of $\mu$ with a homomorphism that results from putting together the identity mapping (on some of ths subgoals) and $\theta^{-1}$, for the $\theta$ used in defining $Q^{\sigma,h,\theta}$. (By definition, $\theta$ is injective and thus $\theta^{-1}$ exists.) 

Observe that this ``simulation'' on $Q'$ of the infinite chase on  $Q^{\sigma,h,\theta}$ cannot collapse the infinite number of variables in $(Q^{\sigma,h,\theta})_{\Sigma,S}$ into a finite number of variables (and thus into a {\em finite} number of subgoals) in the ``simulation result''. The reason is, the language of embedded dependencies cannot specify the instruction ``generate a new variable name, using the right-hand side of the tgd in question, {\em only} if some variable names are not the same in the left-hand side of the tgd in question''. Q.E.D. 
\end{proof}

We are finally ready to define assignment-fixing gds. 

\begin{definition}{Assignment-fixing tgd}
\label{assgn-fix-def}
Given a CQ query $Q$ and a finite set $\Sigma$ of tgds and egds such that $(Q)_{\Sigma,S}$ exists, 
let $\sigma \in \Sigma$ be a regularized tgd with existential variables $Z_1,\ldots,Z_k$, $k \geq 0$, 
such that chase of $Q$ with $\sigma$ is applicable, with associated test query $Q^{\sigma,h,\theta}$. 
Then $\sigma$ is {\em an assignment-fixing tgd w.r.t. $Q$ and} $h$ if $(Q^{\sigma,h,\theta})_{\Sigma,S}$ has at most one of $Z_i$ and $\theta(Z_i)$ for each $i \in \{ 1,\ldots,k \}$.  
Further, $\sigma$ is an {\em assignment-fixing tgd w.r.t.} $Q$ if $\sigma$ is an assignment-fixing tgd w.r.t. $Q$ and some homomorphism $h$. 
\end{definition}

\begin{proposition}
In the setting of Definition~\ref{assgn-fix-def}, whenever $\sigma$ is a full tgd (i.e., tgd without existential variables), then $\sigma$ is an assignment-fixing tgd w.r.t. all CQ queries $Q$ such that chase using $\sigma$ is applicable to $Q$ and such that $(Q)_{\Sigma}$ exists. 
\end{proposition}

Consider two illustrations of the determination 
whether a tgd with existential variables is assignment fixing w.r.t. a given CQ query. Example~\ref{unregularized-counterex-one} is a positive example, in that it establishes a tgd as assignment fixing,  whereas Example~\ref{unregularized-counterex-two} is a negative example. 

\begin{example}
\label{unregularized-counterex-one}
On database schema ${\cal D} = \{ P,$ $R, S \}$, consider a regularized set of embedded dependencies $\Sigma = \{ \sigma_1, \sigma_2, \sigma_3 \}$, where $\sigma_1$ is a tgd,   
$$\sigma_1: \ p(X,Y) \rightarrow \exists Z \ \exists W \ r(X,Z) \wedge s(Z,W) ,$$
egd $\sigma_2$ 
establishes the first attribute of $R$ as its superkey, and, finally, egd $\sigma_3$ is as follows: 
\begin{tabbing}
He tab lallalalalalalalkakalallala \= lalalala \kill
$\sigma_3: \ r(X,Y) \wedge s(Y,T) \wedge r(X,Z) \wedge s(Z,W) \rightarrow T = W . $
\end{tabbing}

Let CQ query $Q$ be $Q(X) \ :- \ p(X,Y)$. Chase using $\sigma_1$ is applicable to $Q$, using homomorphism $h = \{  X \rightarrow X, Y \rightarrow Y \}$. For the query 
\begin{tabbing}
Hejmnnnnkkkkk\=jjjjjjjjjj\kill
$Q^{\sigma_1,h,\theta}(X) \ :- \ p(X,Y), r(X,Z), s(Z,W),$ \\
\> $r(X,Z_1), s(Z_1,W_1) \ .$ 
\end{tabbing}
constructed using $\theta = \{ Z \rightarrow Z_1, W \rightarrow W_1 \}$, we have
\begin{tabbing}
Hehehkkkkknnnnnnnkkkkk\=jjjjjjjjjj\kill
$(Q^{\sigma_1,h,\theta})_{\Sigma,S}(X) \ :- \ p(X,Y), r(X,Z), s(Z,W) \ .$ 
\end{tabbing}

Thus, $\sigma_1$ is an assignment-fixing tgd w.r.t. $Q$, because the body of $(Q^{\sigma_1,h,\theta})_{\Sigma,S}(X)$ has only one of $Z$ and $Z_1$ and only one of $W$ and $W_1$. 
\nop{
\mbox{}

\mbox{}

\mbox{}

Using substitution $\theta = \{ X \rightarrow X, Z \rightarrow Z_1, W \rightarrow W_1 \}$, we construct $\sigma^+_1$ for $\sigma_1$, as follows: 
\begin{tabbing}
He tab lallalalalalalkakalalalalalala \= lalalala \kill
$\sigma^+_1: \ r(X,Z) \wedge s(Z,W) \wedge r(X,Z_1) \wedge s(Z_1,W_1) \rightarrow $ \\
\> $Z = Z_1 \wedge W = W_1 . $
\end{tabbing}

Consider tableau $T$ that we construct to determine whether $\Sigma \models \sigma^+_1$: 

\mbox{}

\begin{tabular}{ccccc}
$T:$ & $B$ & $C$ & $D$ & $E$ \\
\hline
$r:$ & $X$ & $Z$ & & \\
$s:$ & & & $Z$ & $W$ \\
$r:$ & $X$ & $Z_1$ & & \\
$s:$ & & & $Z_1$ & $W_1$ \\
\hline \\
\end{tabular}

We chase $T$ using egd $\sigma_2$, to obtain tableau $T_1$:

\mbox{}

\begin{tabular}{ccccc}
$T_1:$ & $B$ & $C$ & $D$ & $E$ \\
\hline
$r:$ & $X$ & $Z$ & & \\
$s:$ & & & $Z$ & $W$ \\
$s:$ & & & $Z$ & $W_1$ \\
\hline \\
\end{tabular}

Then we chase $T_1$ using $\sigma_3$, to obtain the final chase result $T_2$:

\mbox{}

\begin{tabular}{ccccc}
$T_2:$ & $B$ & $C$ & $D$ & $E$ \\
\hline
$r:$ & $X$ & $Z$ & & \\
$s:$ & & & $Z$ & $W$ \\
\hline \\
\end{tabular}

$T_2$ is the final result of the chase of $T$ with the egds in $\Sigma$. Thus, implication $\Sigma \models \sigma_1^+$ does hold. 
} 
\end{example}

\nop{
{\em Note 1 on Example~\ref{unregularized-counterex-one}.} Removal of either of $\sigma_2$ or $\sigma_3$ from the given set $\Sigma$ would create (in chase) a chase-result tableau that is a counterexample to $\Sigma' \models \sigma_1^+$. (Here, $\Sigma'$ is the result of removing the egd in question from $\Sigma$.) 

{\em Note 2 on Example~\ref{unregularized-counterex-one}.} The first attribute of $S$ is {\em not} its superkey in the example, that is the functional dependency $s(X,Y) \wedge s(X,Z) \rightarrow Y = Z$ does {\em not} hold in presence of the dependencies $\{  \sigma_1, \sigma_2, \sigma_3 \}$ of Example~\ref{unregularized-counterex-one}. 
}

\begin{example}
\label{unregularized-counterex-two}
Using the database schema 
and dependency $\sigma_2$ of Example~\ref{unregularized-counterex-one}, we replace $\sigma_1$ of that example with a regularized tgd $\sigma_4$: 
$$\sigma_4: \ p(X,Y) \rightarrow \exists Z, W, T \ r(X,Z) \wedge s(Z,W) \wedge s(X,T) .$$

We also replace $\sigma_3$ of the example with egd $\sigma_5$, and add an egd $\sigma_6$: 
\begin{tabbing}
He tab lallalalalalalalkakalallala \= lalalala \kill
$\sigma_5: \ r(X,Z) \wedge s(Z,W) \wedge s(X,T) \rightarrow W = T . $ \\
$\sigma_6: \ p(X,Y) \wedge r(A,X) \wedge s(X,T) \rightarrow X = T . $ \\
\end{tabbing}

We denote by $\Sigma'$ the set of dependencies $\{ \sigma_2, \sigma_4, \sigma_5, \sigma_6 \}$. 

Consider again query $Q(X) \ :- \ p(X,Y)$. Chase using $\sigma_4$ is applicable to $Q$, using the identity homomorphism $h$. For the query 
\begin{tabbing}
Hejmnnnnkkkkk\=jjjjjjjjjj\kill
$Q^{\sigma_4,h,\theta}(X) \ :- \ p(X,Y), r(X,Z), s(Z,W), s(X,T)$ \\
\> $r(X,Z_1), s(Z_1,W_1), s(X,T_1) \ .$ 
\end{tabbing}
constructed using $\theta = \{ Z \rightarrow Z_1, W \rightarrow W_1, T \rightarrow T_1 \}$, we have
\begin{tabbing}
Hejmnnnnkkkkk\=jjjjjjjjjj\kill
$(Q^{\sigma_4,h,\theta})_{\Sigma',S}(X) \ :- \ p(X,Y), r(X,Z),$ \\
\> $s(Z,W), s(X,W), s(Z,W_1), s(X,W_1) \ .$ 
\end{tabbing}

Thus, $\sigma_4$ is {\em not} an assignment-fixing tgd w.r.t. $Q$ by definition, because the body of $(Q^{\sigma_4,h,\theta})_{\Sigma',S}(X)$ has both of $W$ and $W_1$. 
\end{example}



\subsubsection{Motivation for Regularized Assignment-Fixing Tgds}

One may wonder whether the notions introduced in Section~\ref{regularized-sec} are justified. In this subsection we illustrate that whenever a non-regularized tgds or a tgd that is not assignment fixing is used in chase step $Q \Rightarrow^{\sigma} Q'$, then the chase result $Q'$ may be nonequivalent to $Q$ under bag or bag-set semantics. 

Examples~\ref{nonregul-notapply-ex} through \ref{regul-partapply-ex} establish the need for regularized tgds and for the (traditional) definition of the chase step for tgds, see Section~\ref{chase-prelims} in this paper. Example~\ref{regularized-notassgnfix-ex} shows an unsound chase step using a regularized tgd that is not assignment fixing w.r.t. the query. Finally, Example~\ref{regul-allapply-ex} demonstrates a sound chase step using a regularized assignment-fixing tgd, and illustrates how the notion of assignment-fixing tgds is strictly more general than that of key-based dependencies (see Definition~\ref{old-key-based-tgds-def}). 

\begin{example}
\label{nonregul-notapply-ex} 
Consider Example~\ref{motivating-example}, where tgd $\sigma_4$ is not key based in presence of the set $\Sigma$ of embedded dependencies in the example, by the definition of \cite{DeutschDiss}, see Definition~\ref{old-key-based-tgds-def}. For the reader convenience, we provide here the tgd $\sigma_4$ and query $Q_4$ of Example~\ref{motivating-example}. 
\begin{tabbing}
$\sigma_4: p(X,Y) \rightarrow u(X,Z) \wedge t(X,Y,W)$  \\
$Q_4(X) \ :- \ p(X,Y) .$
\end{tabbing}

Now consider the result of removing from $\Sigma$ the tgd $\sigma_2$ of  Example~\ref{motivating-example}; we denote by $\Sigma'$ the set $\Sigma' = \Sigma - \{ \sigma_2 \}$. In presence of $\Sigma'$, the tgd $\sigma_4$ is still not key based. However, if we refrain from applying $\sigma_4$ to $Q_4$ in chase under bag or bag-set semantics, then we will miss the rewriting $Q_3$ (of Example~\ref{motivating-example}) of $Q_4$. Indeed, by the results of this paper it holds that $Q_3 \equiv_{\Sigma',B} Q_4$ and that $Q_3 \equiv_{\Sigma',BS} Q_4$. 
\end{example}

Observe that tgd $\sigma_4$ in Example~\ref{nonregul-notapply-ex} is not regularized, see Definition~\ref{regularized-def}. We miss an equivalent rewriting of the input query $Q_4$ by refraining from applying the tgd. Consider now Example~\ref{nonregul-allapply-ex}, where we do apply the nonregularized tgd $\sigma_4$ in its entirety to the query $Q_4$.  However, instead of the query $Q_3$, which is equivalent to $Q_4$ in presence of $\Sigma'$ under each of bag and bag-set semantics, we obtain a formulation of $Q_4$ that is not equivalent to $Q_4$ (in presence of $\Sigma'$) under either semantics. 

\begin{example}
\label{nonregul-allapply-ex}
Consider the query $Q_4$ and set $\Sigma'$ of dependencies in Example~\ref{nonregul-notapply-ex}. We now attempt to find the rewriting $Q_3$ (of Example~\ref{motivating-example}) that we failed to obtain in Example~\ref{nonregul-notapply-ex}. 
\begin{tabbing}
$Q_3(X) \ :- \ p(X,Y), t(X,Y,W), s(X,Z).$ 
\end{tabbing}

To find the rewriting $Q_3$, specifically to obtain its $T$-subgoal,  we apply the nonregularized dependency $\sigma_4$ 
to the query $Q_4$. We denote by $Q'_4$ the result of the application: 
\begin{tabbing}
$Q'_4(X) \ :- \ p(X,Y), t(X,Y,W), u(X,Z).$ 
\end{tabbing}

Recall from Example~\ref{motivating-example} that in presence of $\Sigma'$, relation $U$ does not have superkeys other than the set of all its attributes. 
Using this information, we construct a database $D$ that is a counterexample to equivalence of $Q_4$ and $Q'_4$ in presence of $\Sigma'$ and under bag-set semantics. (Thus, by definition, $D$ is also a counterexample to the equivalence of the queries in presence of $\Sigma'$ and under bag semantics as well). 

Let $D = \{ P(1,2), T(1,2,3), U(1,4), U(1,5) \}$. (Observe that $D$ is a set-valued database and that $D \models \Sigma'$.) On database $D$, $Q_4(D,BS) = \{ \hspace{-0.25cm} \{ \ (1) \ \} \hspace{-0.25cm} \}$, whereas  $Q'_4(D,BS) = \{ \hspace{-0.25cm} \{ \ (1), (1) \ \} \hspace{-0.25cm} \}$. 
\end{example}

{\em Note 1 on Example~\ref{nonregul-allapply-ex}.} The problem with applying $\sigma_4$ to query $Q_4$ in the example is that $\sigma_4$ is not regularized. The regularized set for $\sigma_4$ is $\{ \sigma'_4, \sigma''_4 \}$, where 
\begin{tabbing}
$\sigma'_4: p(X,Y) \rightarrow t(X,Y,W)$  \\
$\sigma''_4: p(X,Y) \rightarrow u(X,Z)$ 
\end{tabbing}
Observe that tgd $\sigma'_4$ is assignment fixing in presence of (the egds in) $\Sigma'$ (of Example~\ref{nonregul-notapply-ex}), whereas $\sigma''_4$ is not. Thus, $\sigma''_4$ cannot be applied in sound chase of $Q_4$ using $\Sigma'$ under bag or bag-set semantics, by our main results of this section. Using the regularized version of $\Sigma'$ (this version also replaces $\sigma_1$ of Example~\ref{motivating-example} with its regularized set), we can perform sound chase $Q_4$ to obtain the above query $Q_3$, which is equivalent to $Q_4$ in presence of $\Sigma'$ under each of bag and bag-set semantics (with the usual restriction of set-valued relations in the case of bag semantics). 

We now examine the modified definition of chase, see Section 2.4 of \cite{ChirkovaG09}. Indeed, using that definition we obtain correctly the terminal chase results of the query $Q_4$ in Example~\ref{motivating-example}, even though not all input tgds are regularized, see Examples 4.1 and 5.1 of \cite{ChirkovaG09}. However, as we see in the next example, using the modified definition of chase does not result in sound chase (under bag or bag-set semantics) for all problem inputs. 

\begin{example}
\label{regul-partapply-ex}
Consider query $Q$ and set $\Sigma = \{ \nu_1,$ $\nu_2 \}$ of dependencies, where
\begin{tabbing}
$Q(X) \ :- \ p(X,Y), s(X,Z)$ \\
$\nu_1: p(X,Y) \rightarrow \exists Z \ s(X,Z) \wedge t(Z,Y)$  \\
$\nu_2: t(X,Y) \wedge t(Z,Y) \rightarrow X = Z$ 
\end{tabbing}

Observe that $\nu_1$ is a regularized tgd and is also assignment fixing, w.r.t. $Q$, by our definitions in this section. 
 We now apply modified chase as defined in Section 2.4 of \cite{ChirkovaG09} and obtain query $Q'$: 
\begin{tabbing}
$Q'(X) \ :- \ p(X,Y), s(X,Z), t(Z,Y) .$ 
\end{tabbing}

We show nonequivalence of $Q$ to $Q'$ in presence of $\Sigma$ under each of bag and bag-set semantics, by constructing a database $D$ that is a counterexample to either equivalence. 
Indeed, let $D = \{ P(1,2), S(1,1), S(1,3), T(3,2) \}$. (Observe that $D$ is a set-valued database and that $D \models \Sigma$.) On database $D$, $Q(D,BS) = \{ \hspace{-0.25cm} \{ \ (1), (1) \ \} \hspace{-0.25cm} \}$, whereas  $Q'(D,BS) = \{ \hspace{-0.25cm} \{ \ (1) \ \} \hspace{-0.25cm} \}$. 
\end{example}

{\em Note on Example~\ref{regul-partapply-ex}.} The application of $\nu_1$ to $Q$ in the example is sound by the (incorrect) definition of key-based chase steps in \cite{ChirkovaG09}. 
Still, the application of the regularized and assignment-fixing tgd $\nu_1$ using the {\em modified} definition of the chase step does result in unsound chase as shown in Example~\ref{regul-partapply-ex}. 

We now show an example of using a regularized but {\em not} assignment fixing tgd in a (traditional) chase step, see Section~\ref{chase-prelims} in this paper for the definition. 

\begin{example}
\label{regularized-notassgnfix-ex}
Recall the database schema ${\cal D} = \{ P, R, S \}$ 
and dependencies $\Sigma' = \{ \sigma_2,$ $\sigma_4, \sigma_5 \}$ of Example~\ref{unregularized-counterex-two}. 
\begin{tabbing} 
$\sigma_2: r(X,Y) \wedge r(X,Z) \rightarrow Y = Z \ .$ \\
$\sigma_4: p(X,Y) \rightarrow \exists Z, W, T \ r(X,Z) \wedge s(Z,W) \wedge s(X,T) \ .$ \\
$\sigma_5: r(X,Z) \wedge s(Z,W) \wedge s(X,T) \rightarrow W = T \ . $
\end{tabbing}

Recall that $\sigma_4$ is regularized but not assignment fixing w.r.t. query $Q(X) \ :- \ p(X,Y)$; see Example~\ref{unregularized-counterex-two} for the details. We apply the chase step using tgd $\sigma_4$ to $Q$, to obtain the result $Q''$: 
\begin{tabbing} 
$Q(X) \ :- \ p(X,Y) \ . $ \\
$Q''(X) \ :- \ p(X,Y),r(X,Z), s(Z,W), s(X,T) \ . $
\end{tabbing}

To construct a counterexample to equivalence of $Q$ and $Q''$ in presence of $\Sigma'$, under each of bag and bag-set semantics, we use the query $(Q^{\sigma_4,h,\theta})_{\Sigma',S}$ of Example~\ref{unregularized-counterex-two}. Specifically, we use as a counterexample the canonical database, call it $D$, of $(Q^{\sigma_4,h,\theta})_{\Sigma',S}$; we have that $D$ is set valued and that $D \models \Sigma'$ by definition of the query $(Q^{\sigma_4,h,\theta})_{\Sigma',S}$. 

Consider the database $D = \{ P(1,2), R(1,3), S(1,4),$ $S(1,5), S(3,4), S(3,5)  \}$. 
(Recall that the canonical database of a CQ query is isomorphic up to choice of constants.) We have that $Q(D,BS) = \{ \hspace{-0.25cm} \{ \ (1) \ \} \hspace{-0.25cm} \}$, whereas  $Q'(D,BS) = \{ \hspace{-0.25cm} \{ \ (1), (1), (1), (1) \ \} \hspace{-0.25cm} \}$.  
\end{example}

By our main results in this section, for the $Q$, $\Sigma$, and $\nu_1$ of Example~\ref{regul-partapply-ex}, the application of $\nu_1$ to $Q$ (using the traditional definition of chase steps using tgds, see Section~\ref{chase-prelims} in this paper)  is sound in presence of $\Sigma$ under each of bag and bag-set semantics (provided that for the case of bag semantics, both $S$ and $T$ are set-valued relations in all instances of $\{ P, S, T \}$). Example~\ref{regul-allapply-ex}  shows the chase step.


\begin{example}
\label{regul-allapply-ex}
Consider the query $Q$ and set $\Sigma = \{ \nu_1,$ $\nu_2 \}$ of dependencies of Example~\ref{regul-partapply-ex}. Recall that $\nu_1$ is a regularized tgd and is also assignment fixing w.r.t. $Q$ in presence of the egds of $\Sigma$. 
We now apply (traditional) chase as defined in Section~\ref{chase-prelims} in this paper, to obtain query $Q''$: 
\begin{tabbing}
$Q''(X) \ :- \ p(X,Y), s(X,Z), s(X,W), t(W,Y) .$ 
\end{tabbing}

The difference from our application of $\nu_1$ in Example~\ref{regul-partapply-ex} is that we now add a new $S$-subgoal in addition to a new $T$-subgoal. By the definition of chase steps using tgds, the second attribute of $S$ must be denoted by different variables in the two $S$-subgoals in query $Q''$. 
\end{example}

{\em Note on Example~\ref{regul-allapply-ex}.} Recall that $\nu_1$ in the example is assignment fixing w.r.t. the query, and thus by our results can be applied in sound chase under bag and bag-set semantics (provided that for the case of bag semantics, both $S$ and $T$ are set-valued relations in all instances of $\{ P, S, T \}$). At the same time, $\nu_1$ is not key-based by the definition of \cite{DeutschDiss}, see Definition~\ref{old-key-based-tgds-def} in this paper. The problem is with the $S$-atom of $\nu_1$, which is not key based  in presence of $\Sigma$ by Definition~\ref{old-key-based-tgds-def}. 


\subsubsection{Assignment-Fixing Chase}
\label{key-based-chase-subsec}

We begin the exposition of the main results of this section by defining {\em assignment-fixing chase steps using tgds}. 
\nop{

, resulting in chase step $Q \Rightarrow^{\sigma}_X Q'$, where $X$ is  $B$ or $BS$. Let ${\bf S} = \{ s(p_{i1}),\ldots,s(p_{ik}) \}$, $k > 0,$ be the set of all new subgoals, with up to $k$ not necessarily distinct predicates $p_{i1},\ldots,p_{ik}$, that this chase step adds to $Q$, to result in $Q'$.  By definition of the chase step for tgds, there is a 1:1 correspondence between ${\bf S}$ and a nonempty subset $\psi ''$ of $\psi$ in $\sigma$, viz. $\psi '' = \{ p_{i1}(\bar{W}_{i1}),\ldots,p_{ik}(\bar{W}_{ik}) \}$, with $k = |{\bf S}| > 0$. 
} 
\vspace{-0.1cm}
\begin{definition}{Assignment-fixing chase step using tgd}
\label{key-based-tgds-def}
Let 
$\sigma$ 
be a regularized tgd in a finite set $\Sigma$ of embedded dependencies on schema $\cal D$. 
Consider a CQ query $Q$ defined on $\cal D$, such that $(Q)_{\Sigma,S}$ exists and such that $\sigma$ is applicable to $Q$. Then the chase step that applies $\sigma$ to $Q$ is an  {\em assignment-fixing chase step using} $\sigma$ whenever $\sigma$ is an assignment-fixing tgd w.r.t. $Q$. 
\end{definition}
\vspace{-0.1cm}

We now provide necessary and sufficient conditions for soundness of chase steps under bag semantics for query evaluation.

\vspace{-0.1cm}
\begin{theorem}
\label{bag-chase-sound-theorem}
Given a CQ query $Q$ and a set of embedded dependencies\footnote{Recall that we consider only finite  regularized sets of dependencies throughout this paper.} $\Sigma$ on schema $\cal D$. Under bag semantics, 
a chase step $Q \Rightarrow^{\sigma} Q'$ using $\sigma \in \Sigma$ is sound iff 
\begin{enumerate}
\vspace{-0.2cm}
	\item  $Q \Rightarrow^{\sigma} Q'$ is a (tgd) assignment-fixing chase step, and for each subgoal $s(p_{ij})$ that the chase step adds to $Q$, relation $P_{ij}$ is set valued on all databases satisfying $\Sigma$; or
\vspace{-0.2cm}
	\item In  $Q \Rightarrow^{\sigma} Q'$, $\sigma$ is an egd; in this case, duplicates of subgoal $s(p)$ in $Q'$ can be removed only if relation $P$ is set valued in all instances of $\cal D$.
\end{enumerate}
\vspace{-0.6cm}
\end{theorem}
\vspace{-0.1cm}

In Section~\ref{key-based-chase-subsec}, Example~\ref{regularized-notassgnfix-ex}  shows an unsound chase step using a regularized tgd that is not assignment fixing w.r.t. the query. Example~\ref{regul-allapply-ex} in Section~\ref{key-based-chase-subsec} demonstrates a sound (by Theorem~\ref{bag-chase-sound-theorem}) chase step using a regularized assignment-fixing tgd, {\em provided that} both $S$ and $T$ are set-valued relations in all instances of the database schema used in the example. Relaxing this set-valued requirement would result in an unsound chase step using the same tgd, as is easy to demonstrate using a counterexample bag-valued database.   


The requirement that certain stored relations be set valued arises naturally if one seeks soundness of bag-semantics chase, see~\cite{DeutschDiss}. We now show that constraints that force certain relations to be sets on all instances of a database schema can be formally defined as egds, provided that {\it row (tuple) IDs} are defined for the respective relations. In the common practice of using tuple IDs in database systems, each tuple in a (bag-valued) relation is assigned a unique tuple ID.  Then the {\it set-enforcing egd} on relation $P$ can be expressed as a functional dependency ({\em fd,} defined in Appendix~\ref{key-app}), which specifies that whenever two tuples of $P$ agree on everything except the tuple IDs, then the tuples must also agree on the tuple IDs. Please see Appendix~\ref{appendix-a} for  the details of our set-enforcing framework based on tuple IDs. 

We now discuss item 2 of Theorem~\ref{bag-chase-sound-theorem}. Given a database schema $\cal D$, suppose that for some of the relation symbols $\{ P_1,\ldots,P_k \} \subseteq \cal D$ it holds that the relation for each of $P_1,\ldots,P_k$ is required to be {\it set} valued in all instances $D$ over $\cal D$. 
For such scenarios, the bag-equivalence test of Theorem~\ref{cv-theorem} is no longer a necessary condition for bag equivalence of CQ queries. 

\begin{example}
\label{extend-necess-example}
By Theorem~\ref{cv-theorem}, query $Q_3$ of Example~\ref{motivating-example} is not bag equivalent to query $Q_5$:  
\begin{tabbing}
$Q_5(X) \ :- \ p(X,Y), t(X,Y,W), s(X,Z), s(X,Z).$ 
\end{tabbing}
Here, the only difference between $Q_3$ and $Q_5$ is the extra copy of subgoal $s(X,Z)$ in $Q_5$. 
At the same time,  $Q_3$ and $Q_5$ {\em are} bag equivalent on all bag-valued databases where relation $S$ is required to be a set. Please see Theorem~\ref{cv-updated-thm} and Appendix~\ref{set-bag-appendix} for the details. 
\end{example}

We now formulate the extended sufficient and necessary condition. Please see Appendix~\ref{set-bag-appendix} for the proof. 
\begin{theorem}
\label{cv-updated-thm}
Let $\{ P_1,\ldots,P_k \} \subseteq \cal D$ be the maximal set of relation symbols in schema $\cal D$ such that  the relation for each of $P_1,\ldots,P_k$ is required to be {\em set} valued in all instances $D$ over $\cal D$. Given CQ queries $Q_1$, $Q_2$ on $\cal D$, let query $Q'_1$ ($Q'_2$, respectively) be obtained  by removing from $Q_1$ (from $Q_2$, respectively) all duplicate subgoals whose predicates correspond to $P_1,\ldots,P_k$.  Then $Q_1 \equiv_B Q_2$ in the absence of all dependencies other than the set-enforcing dependencies on $P_1,\ldots,P_k$ of the schema $\cal D$ if and only if  $Q'_1$ and $Q'_2$ are isomorphic. 
\end{theorem}}

The correctness of the duplicate-removal rule of item 2 in Theorem~\ref{bag-chase-sound-theorem}  is immediate from Theorem~\ref{cv-updated-thm}.

We now spell out the necessary and sufficient conditions for soundness of chase steps under {\em bag-set} semantics for query evaluation. 

\vspace{-0.1cm}


\begin{theorem}
\label{bag-set-chase-sound-theorem}
Given a CQ query $Q$ and a set of embedded dependencies\footnote{Recall that we consider only finite  regularized sets of dependencies throughout this paper.}  $\Sigma$. Under bag-set semantics, 
a chase step $Q \Rightarrow^{\sigma} Q'$ using $\sigma \in \Sigma$ is sound iff 
\begin{enumerate}
\vspace{-0.2cm}
	\item  $Q \Rightarrow^{\sigma} Q'$ is a (tgd) assignment-fixing chase step; or 
\vspace{-0.2cm}
	\item In  $Q \Rightarrow^{\sigma} Q'$, $\sigma$ is an egd. 
\end{enumerate}
\vspace{-0.5cm}
\end{theorem}



We use Examples~\ref{regularized-notassgnfix-ex} and~\ref{regul-allapply-ex} of Section~\ref{key-based-chase-subsec} to make here the same points as for Theorem~\ref{bag-chase-sound-theorem}. Observe that (unlike the case of bag semantics) the set-valuedness requirement is satisfied by definition of bag-set semantics. See Example~\ref{motivating-example} for query $Q_2$ that is obtained from $Q_4$ by using, among other sound chase steps, a chase step involving dependency $\sigma_3$. By Theorem~\ref{bag-chase-sound-theorem}, $\sigma_3$ may not be used in sound chase under {\em bag} semantics, because relation $S$ is not guaranteed to be set valued in all instances of the database schema of the example. 

\begin{proof}{(Theorems~\ref{bag-chase-sound-theorem} and~\ref{bag-set-chase-sound-theorem}, sketch.)}
We outline here the correctness proof for chase steps using tgds. Please see Appendix~\ref{proof-sound-chase-steps-appendix} for the details of disproving soundness of chase steps under bag semantics whenever chase (using even regularized and assignment-fixing tgds) adds query subgoals whose associated base relations are not set valued in all instances of the given database schema.

Consider a CQ query $Q$ and a set of dependencies $\Sigma$ defined on schema $\cal D$, such that $(Q)_{\Sigma}$ exists. Let $\sigma \in \Sigma$ be a regularized dependency such that chase using $\sigma$ applies to $Q$ (using a homomorphism $h$) and results in query $Q'$. (That is, $Q \Rightarrow^{\sigma} Q'$ is defined.) Further, suppose that for all subgoals that are in $Q'$ but not in $Q$, the respective base relations, call them collectively ${\cal S} \subseteq {\cal D}$, are set valued in all instances of the schema $\cal D$.

Case (1): Let $\sigma$ be an assignment-fixing tgd w.r.t. the query $Q$. We prove that on all instances $D$ of $\cal D$ such that $D \models \Sigma$ and such that at least the relations in $\cal S$ are set valued on $D$, it holds that $Q(D,B) = Q'(D,B)$ and that $Q(D,BS) = Q'(D,BS)$. (Thus, the chase step $Q \Rightarrow^{\sigma} Q'$ is sound under the conditions of Theorems~\ref{bag-chase-sound-theorem} and~\ref{bag-set-chase-sound-theorem}.) 

We fix an arbitrary database $D$ as described above. The idea of the proof  is to establish a 1:1 correspondence between all the assignments satisfied by $Q$ w.r.t. $D$ and all the assignments satisfied by $Q'$ w.r.t. $D$. As a result (and using the fact that the $\cal S$-part of the base relations in $D$ is guaranteed to be set valued), we obtain that for each tuple $t \in Q(D,B)$, such that the multiplicity of $t$ in $Q(D,B)$ is $m > 0$, the multiplicity of $t$ {\em in} $Q'(D,B)$ is also $m$. 

We establish the 1:1 correspondence as follows. 

\begin{itemize}
	\item[(i)] For each assignment $\mu'$ that satisfies $Q'$ w.r.t. $D$, there exists exactly one assignment $\mu$ that (a) satisfies $Q$ w.r.t. $D$, and that (b) coincides with $\mu'$ on the set of body variables of $Q$. (Recall that $\sigma$ is a tgd, and therefore the set of body variables of $Q'$ is a superset of the set of body variables of $Q$.)  
	\item[(ii)] For each assignment $\mu$ that satisfies $Q$ w.r.t. $D$, there exists at least one assignment $\mu'$ that (a) satisfies $Q'$ w.r.t. $D$, and that (b) coincides with $\mu$ on the set of body variables of $Q$. This is immediate from the fact that $D \models \Sigma$. 
	\item[(iii)] From the fact that $\sigma$ is assignment fixing w.r.t. $Q$, we obtain that for each $\mu$ as in (ii) there exists {\em at most} one corresponding $\mu'$ as in (ii). Indeed, suppose that for some such $\mu$ there exist at least two assignments $\mu'_1$ and $\mu'_2$ that satisfy the conditions of (ii). Then we show by obtaining the chase result $(Q^{\sigma,h,\theta})_{\Sigma,S}$, in Definition~\ref{assgn-fix-def}, that $\mu'_1$ and $\mu'_2$ must be identical on all databases satisfying $\Sigma$. 
\end{itemize}

The observation that $D$ is an arbitrary database satisfying the conditions above concludes the proof of $Q \equiv_{\Sigma,B} Q'$ in this case (1). Further, $Q \equiv_{\Sigma,BS} Q'$ is immediate from $Q \equiv_{\Sigma,B} Q'$. 

Case (2): Let $\sigma$ {\em not} be assignment fixing w.r.t. the query $Q$. We construct a set-valued database $D$ (with schema $\cal D$) such that $D \models \Sigma$ and such that $Q(D,BS) \neq Q'(D,BS)$. (As a result, neither of $Q \equiv_{\Sigma,B} Q'$ and $Q \equiv_{\Sigma,BS} Q'$ holds, and therefore the chase step $Q \Rightarrow^{\sigma} Q'$ is not sound in this case under bag or bag-set semantics.) 

As a counterexample database $D$ we use the canonical database of the query $(Q^{\sigma,h,\theta})_{\Sigma,S}$, see Definition~\ref{assgn-fix-def}. Example~\ref{regularized-notassgnfix-ex} illustrates the construction. 

Let $\nu$ be the satisfying (by definition of canonical databases and by definition of chase under set semantics) assignment to the head variables $\bar{X}$ of $Q^{\sigma,h,\theta}$ w.r.t. the database $D$. Observe that the vectors of head variables of all of $Q$, $Q'$, and $Q^{\sigma,h,\theta}$ are the same by definition of $Q^{\sigma,h,\theta}$. By definition of $Q^{\sigma,h,\theta}$, there exists an extension $\nu_Q$ of $\nu$ to all the body variables of $Q$ such that $\nu_Q$ satisfies $Q$ w.r.t. $D$, and there exists an extension $\nu'_{Q'}$ of $\nu$ to all the body variables of $Q'$ such that $\nu'_{Q'}$ satisfies $Q'$ w.r.t. $D$. 

We make the following observations about the answers to $Q$ and $Q'$ under bag-set semantics on the set-valued database $D$. 

\begin{itemize}
	\item[(i)] For each assignment $\mu'$ such that $\mu'|_{\bar{X}} = \nu$ and such that $\mu'$ satisfies $Q'$ w.r.t. $D$ (we have shown that there exists at least one such assignment $\mu'$), there exists exactly one assignment $\mu$ that (a) satisfies $Q$ w.r.t. $D$, and that (b) coincides with $\mu'$ on the set of body variables of $Q$. (See (i) under case (1) of the proof.) Observe that $\mu|_{\bar{X}} = \nu$ by definition of $\mu$. 
	\item[(ii)] For each assignment $\mu$ such that $\mu|_{\bar{X}} = \nu$ and such that $\mu$ satisfies $Q$ w.r.t. $D$  (we have shown that there exists at least one such assignment $\mu$), there exists at least one assignment $\mu'$ that (a) satisfies $Q'$ w.r.t. $D$, and that (b) coincides with $\mu$ on the set of body variables of $Q$. (See (ii) under case (1) of the proof.) Observe that $\mu'|_{\bar{X}} = \nu$ by definition of $\mu'$. 
	\item[(iii)] On our counterexample database $D$, there exists at least one $\mu$ with $\mu|_{\bar{X}} = \nu$ and such that $\mu$ is a satisfying assignment w.r.t. $Q$ and $D$, such that $\mu$ corresponds to {\em at least two distinct} satisfying assignments $\mu'_1$ and $\mu'_2$ w.r.t. $Q'$ and $D$, where each of $\mu'_1$ and $\mu'_2$  coincides with $\mu$ on all the body variables of $Q$. Indeed, we recall that $D$ is the canonical database of  $(Q^{\sigma,h,\theta})_{\Sigma,S}$. 
	If the distinct $\mu'_1$ and $\mu'_2$ as above did not exist, then chase of $Q^{\sigma,h,\theta}$ using $\Sigma$ under set semantics would lead to the ``elimination of the distinction between'' the groups of subgoals $\psi(h(\bar{X}),\bar{Z})$ and $\psi(h(\bar{X}),\theta(\bar{Z}))$ of $Q^{\sigma,h,\theta}$,  
see Equation~\ref{theta-eq} and Definition~\ref{assgn-fix-def},  in the terminal chase result of $Q^{\sigma,h,\theta}$ using $\Sigma$. But if $\psi(h(\bar{X}),\bar{Z})$ and $\psi(h(\bar{X}),\theta(\bar{Z}))$ collapse into the same group in $(Q^{\sigma,h,\theta})_{\Sigma,S}$, then $\sigma$ is an assignment-fixing tgd w.r.t. $Q$ by Definition~\ref{assgn-fix-def}, which is a contradiction with our assumption. 
\end{itemize}

We conclude that in Case (2), the multiplicity of the tuple $\nu(\bar{X})$ is strictly greater in $Q'(D,BS)$ than in $Q(D,BS)$ on our counterexample database $D$. Thus, $Q'(D,BS) \neq Q(D,BS)$. Q.E.D.  
\end{proof}

\section{Unique Result of Sound Chase}
\label{un-res-sec}

In this section we show that the  result of sound chase of CQ queries using arbitrary finite sets of embedded dependencies is unique under each of bag and bag-set semantics for query evaluation. 
Further, we provide an algorithm for constructing, for a given CQ query $Q$ and an arbitrary finite set of embedded depedencies $\Sigma$, the maximal subset $\Sigma^{max}_B(Q,\Sigma)$ of $\Sigma$ such that $D^{(Q_n)} \models \Sigma^{max}_B(Q,\Sigma)$, where $Q_n$ is the result of sound chase of $Q$ under bag semantics. We also outline a version of the algorithm that works for the case of bag-set semantics. 

\subsection{Why Not Key-Based Tgds?}
\label{why-not-key-based-sec}

We begin the discussion by examining the question of why the definition of assignment-fixing chase steps (Definition~\ref{key-based-tgds-def})  cannot be simplified. The intuition behind the notion of assignment-fixing chase steps is that of ensuring that in each assignment-fixing chase step $Q \Rightarrow^{\sigma}_B Q'$, 
using some tgd $\sigma \in \Sigma$, 
each tuple in the bag $Q(D,B)$ would have {\it the same} multiplicity in the bag $Q'(D,B)$, for each database $D \models \Sigma$, in presence of the requisite set-enforcing constraints (of Appendix~\ref{appendix-a}). 
The intuition is the same for bag-set-semantics. 
It appears that a simpler notion, that of {\em key-based tgds,} would suffice. In the definition that follows, we use the notation of Definition~\ref{key-based-tgds-def}. 

 \vspace{-0.1cm}
\begin{definition}{Key-based tgd}
\label{old-key-based-tgds-def}
Let 
$\sigma: \phi(\bar{X},\bar{Y}) \rightarrow \exists{\bar{Z}} \ \psi(\bar{Y},\bar{Z})$ be a tgd 
on database schema $\cal D$. Then $\sigma$ is a {\em key-based tgd} if, for each atom  $p(\bar{Y}'_j,\bar{Z}'_j)$ in $\psi$, 
$\bar{Y}'_j$ is a superkey of relation $P$ in $\cal D$ and, in addition, $P$ is set valued on all instances of $\cal D$. 
\end{definition}
\vspace{-0.1cm}

The notion of key-based tgds is equivalent to that of UWDs of~\cite{DeutschDiss}. 
Note that by Definition~\ref{key-based-tgds-def}, all chase steps using key-based tgds are assignment fixing. However, the class of assignment-fixing tgds (w.r.t. the given CQ query and set of dependencies) includes not just key-based tgds, as illustrated in Example~\ref{regul-allapply-ex}. In addition, unlike assignment-fixing chase steps specified in Definition~\ref{key-based-tgds-def},  a key-based tgd is defined independently of the queries being chased. Deutsch~\cite{DeutschDiss} showed that the result of sound chase of CQ queries under bag semantics is unique up to isomorphism, provided that {\it all} tgds in the given set of dependencies are key based. 

It turns out that the ``key-basedness'' constraints of Definition~\ref{old-key-based-tgds-def} on tgds are not necessary for soundness of chase under either of bag and bag-set semantics. Indeed, consider a modification of Example~\ref{unregularized-counterex-two}:

\begin{example}
\label{query-dep-ex}
In the setting of Example~\ref{unregularized-counterex-two}, we replace the query $Q$ by a query $Q'(X) \ :- \ p(X,Y), r(A,X)$, and keep the set $\Sigma'$ of dependencies of Example~\ref{unregularized-counterex-two}. We can show that tgd $\sigma_4 \in \Sigma'$ is assignment fixing w.r.t. $Q'$. Recall that $\sigma_4$ is {\em not} an assignment-fixing tgd w.r.t. the query $Q$ of  Example~\ref{unregularized-counterex-two}. 
\end{example}



\subsection{Uniqueness of Result of Sound Chase}

We now show that the result of sound chase of CQ queries using {\it arbitrary} sets of embedded dependencies\footnote{Cf. the result of~\cite{DeutschDiss} on  uniqueness of sound bag chase for key-based tgds only; see Section~\ref{why-not-key-based-sec} for the discussion.} is unique under bag and bag-set semantics, up to equivalence in the absence of dependencies (except for the set-enforcing dependencies under bag semantics). (Recall that throughout the paper we assume that all given sets of embedded dependencies are finite and regularized.) We give here a formulation of our result only for the case of bag semantics. The version of Theorem~\ref{uniqueness-theorem} for the case of bag-set semantics (formulated in Appendix~\ref{bag-set-uniqueness-appendix}) is straightforward. 
\begin{theorem}
\label{uniqueness-theorem}
Given a CQ query $Q$ and set $\Sigma$ of embedded dependencies on schema $\cal D$, such that there exists a  {\em set-}chase result $(Q)_{\Sigma,S}$ for $Q$ and $\Sigma$.  Then there exists a result $(Q)_{\Sigma,B}$ of  {\em sound} chase for $Q$ and $\Sigma$ under {\em bag} semantics, unique up to isomorphism after dropping duplicate subgoals that correspond to set-valued relations in $\cal D$.\footnote{See discussion of Theorem~\ref{cv-updated-thm} in Section~\ref{making-chase-sound-section}.} 
That is, for two sound-chase results $(Q)^{(1)}_{\Sigma,B}$ and $(Q)^{(2)}_{\Sigma,B}$ for $Q$ and $\Sigma$, $(Q)^{(1)}_{\Sigma,B} \equiv_B (Q)^{(2)}_{\Sigma,B}$ in the absence of all dependencies other than the set-enforcing dependencies on stored relations. 
\end{theorem}

By Theorem~\ref{bag-chase-sound-theorem}, sound bag chase adds or drops only those subgoals whose predicates correspond to relations required to be sets. Thus, it is natural to use the conditions of Theorem~\ref{cv-updated-thm}, rather than of Theorem~\ref{cv-theorem}, in characterizing bag equivalence of terminal chase results. 

To prove Theorem~\ref{uniqueness-theorem}, we make the following straightforward observation. 
\begin{proposition}
\label{chase-termination-prop}
Given CQ query $Q$ and embedded dependencies $\Sigma$ 
such that there exists a {\em set-}chase result $(Q)_{\Sigma,S}$. 
Then sound chase of $Q$ using $\Sigma$ terminates in finite time under each of bag and bag-set semantics. 
\end{proposition}

This result is immediate from Theorems~\ref{bag-chase-sound-theorem} and~\ref{bag-set-chase-sound-theorem}.

The rest of the proof of Theorem~\ref{uniqueness-theorem} is an adaptation, to {\it sound} chase steps, of the proof of the fact (see~\cite{DeutschPods08}) that all set-chase results (when defined) for a given CQ query are equivalent in the absence of dependencies. Please see Appendix~\ref{bag-set-uniqueness-appendix} for the details. 

We now establish the complexity of sound bag and bag-set chase under weakly acyclic dependencies~\cite{FaginKMP05}. Intuitively, weakly acyclic dependencies cannot generate an infinite number of new variables, hence set-chase under such dependencies terminates in finite time; please see Appendix~\ref{app-chase-compl} for the definition. All sets of dependencies in examples in this paper are weakly acyclic. \begin{theorem}
\label{complexity-thm}
Given a CQ query $Q$ and set $\Sigma$ of weakly acyclic embedded dependencies on schema $\cal D$. 
Then sound chase of $Q$ using $\Sigma$, under each of bag and bag-set semantics, terminates in time polynomial in the size of $Q$ and exponential in the size of $\Sigma$. 
\end{theorem}

The upper  bound is immediate from Proposition~\ref{chase-termination-prop} and from the results in~\cite{AbiteboulHV95,DeutschPT06,FaginKMP05} for  set semantics. For the lower bound, we exhibit an infinite family of pairs $(Q,\Sigma)$, where the size of each of $(Q)_{\Sigma,B}$ and $(Q)_{\Sigma,BS}$ is polynomial in the size of $Q$ and exponential in the size of $\Sigma$. Please see Appendix~\ref{app-chase-compl} for the details. 

\subsection{Satisfiable Dependencies Are Query Based}
\label{semantic-sec-now}

We now provide a constructive characterization of the result of sound chase under bag and bag-set semantics. This  characterization, formulated in Theorem~\ref{sigma-max-theorem} for  bag semantics, settles the problem of which dependencies $\Sigma'$ are satisfied by the canonical database $D^{(Q_n)}$ of $Q_n$. Here, $Q_n$ is the result  of sound chase of  CQ query $Q$ using embedded dependencies $\Sigma$. 
(We assume that {\em set} chase of $Q$ using $\Sigma$ terminates in finite time.) 

Given a CQ query $Q$ and a set of embedded dependencies $\Sigma$, consider  the canonical database $D^{(Q_n)}$ of the result $Q_n = (Q)_{\Sigma,B}$ of sound chase of  $Q$ using $\Sigma$ under bag semantics. 
Clearly, at least some sets  $\Sigma'$ such that $D^{(Q_n)} \models \Sigma'$ do not coincide with the original $\Sigma$.  (We refer here to the discussion in the beginning of Section~\ref{new-sound-chase-sec}.) For instance, in Example~\ref{motivating-example} the canonical database for query $Q_3$ does not satisfy dependency $\sigma_4$. Observe that $Q_3$ is the (unique, by Theorem~\ref{uniqueness-theorem}) result of sound chase of $Q_4$ using $\Sigma$ under bag semantics. 

At the same time, for each pair $(Q,\Sigma)$ there exists a unique maximal-size set $\Sigma^{max}_B(Q,\Sigma) \subseteq \Sigma$, such that  $D^{(Q_n)} \models \Sigma^{max}_B(Q,\Sigma)$. (Appendix~\ref{semant-appendix} has proof of Theorem~\ref{sigma-max-theorem} and the analogous result for bag-set semantics.) 


\vspace{-0.1cm}

\begin{theorem}{{\bf (Unique $\Sigma^{max}_B(Q,\Sigma) \subseteq \Sigma$)}}
\label{sigma-max-theorem}
Given a CQ query $Q$ and set $\Sigma$ of embedded dependencies, such that there exists a  {\em set-}chase result $(Q)_{\Sigma,S}$ for $Q$ and $\Sigma$. 
Let $Q_n$ be the result of  sound chase for $Q$ and $\Sigma$ under {\em bag} semantics, with canonical database $D^{(Q_n)}$.  Then there exists a unique subset $\Sigma^{max}_B(Q,\Sigma)$ of $\Sigma$, such that:
\begin{itemize}
\vspace{-0.1cm}

	\item $D^{(Q_n)} \models \Sigma^{max}_B(Q,\Sigma)$, and 
	\item for each proper superset $\Sigma'$ of $\Sigma^{max}_B(Q,\Sigma)$ such that $\Sigma' \subseteq \Sigma$, $D^{(Q_n)} \models \Sigma'$ does {\em not} hold. 
\end{itemize}
\vspace{-0.52cm}
\end{theorem}

\vspace{-0.1cm}


It turns out that the set $\Sigma^{max}_{B}(Q,\Sigma)$ is the result of removing from $\Sigma$ exactly those tgds $\sigma$ such that the chase step $Q_n \Rightarrow^{\sigma}_B Q'$, with some CQ outcome $Q'$, is not sound under bag semantics. This claim is immediate from the observation that for each dependency $\sigma$ in $\Sigma$ such that $\sigma$ is applicable to $Q_n$, $\sigma$ is {\em unsoundly} applicable to $Q_n$. See Appendix~\ref{semant-appendix} for the details. We make the same observation about the unique set $\Sigma^{max}_{BS}(Q,\Sigma) \subseteq \Sigma$ such that $\Sigma^{max}_{BS}(Q,\Sigma)$ is the maximal set of dependencies satisfied by the canonical database of the result of sound chase of $Q$ using $\Sigma$ under {\em bag-set} semantics. 

Not surprisingly, each of $\Sigma^{max}_{B}(Q,\Sigma)$ and $\Sigma^{max}_{BS}(Q,\Sigma)$  is query dependent. Recall that in Example~\ref{motivating-example} the canonical database  
of the query $Q_3 = (Q_4)_{\Sigma,B}$ 
does not satisfy dependency $\sigma_4$ in the set $\Sigma$ given in the example. At the same time, it is easy to see that for query $Q(X) \ :- \ p(X,Y), u(X,Z),$ 
the canonical database  
of the query $(Q)_{\Sigma,B}$ 
{\it does} satisfy dependency $\sigma_4$ in the {\em same} set $\Sigma$. 

We now establish a relationship between $\Sigma^{max}_{B}(Q,\Sigma)$ and $\Sigma^{max}_{BS}(Q,\Sigma)$ for a fixed pair $(Q,\Sigma)$. This relationship is immediate from Theorems~\ref{bag-chase-sound-theorem}, \ref{bag-set-chase-sound-theorem}, \ref{sigma-max-theorem}, and~\ref{bag-set-sigma-max-theorem}.  
\vspace{-0.4cm}

\begin{proposition}
\label{dep-subset-prop}
For $(Q,\Sigma)$ satisfying conditions of Theorem~\ref{sigma-max-theorem}, $\Sigma^{max}_{B}(Q,\Sigma) \subseteq  \Sigma^{max}_{BS}(Q,\Sigma) \subseteq \Sigma$.  
\end{proposition}
\vspace{-0.2cm}

Query $Q_4$ and dependencies $\Sigma$ of Example~\ref{motivating-example} can be used to show that both subset relationships can be proper: $\Sigma^{max}_{B}(Q,\Sigma) \subset \Sigma^{max}_{BS}(Q,\Sigma) \subset \Sigma$.

We now outline algorithm {\sc Max-Bag-$\Sigma$-Subset}, which accepts as inputs a CQ query $Q$ and a finite set $\Sigma$ of embedded dependencies such that $(Q)_{\Sigma,S}$ exists. The algorithm constructs the set $\Sigma^{max}_B(Q,\Sigma)$ as specified in Theorem~\ref{sigma-max-theorem}. The counterpart of {\sc Max-Bag-$\Sigma$-Subset} for  bag-set semantics can be found in Appendix~\ref{semant-appendix}. 

\nop{

 of $\Sigma$, such that $D^{((Q)_{\Sigma,B})} \models \Sigma^{max}_B(Q,\Sigma)$, where $(Q)_{\Sigma,B}$ is the (unique) 
 result of sound chase for $Q$ and $\Sigma$ under {\em bag} semantics. 
 (The extensions of all the results and discussions in this section to the case of bag-set semantics are straightforward and are omitted.) \reminder{Explain that $\Sigma^{max}_B(Q,\Sigma)$ is, in general, query dependent; and is also chase-step dependent, but in a sense analogous to that of set semantics. That is, if a dependency is at first unsoundly applicable, it may become *soundly* applicable later.}
 } 
 
 \vspace{-0.1cm}

 \begin{algorithm}[htp]

\label{algo: example}

\caption{Max-Bag-$\Sigma$-Subset($Q,\Sigma$)}

\SetKwInOut{Input}{Input}

\SetKwInOut{Output}{Output} 

\SetKwInOut{Feature}{Feature}

\dontprintsemicolon

\Input{CQ query $Q$, set $\Sigma$ of embedded dependencies such that chase result $(Q)_{\Sigma,S}$ exists.}

\Output{$\Sigma^{max}_B(Q,\Sigma) \subseteq \Sigma$ specified in Theorem~\ref{sigma-max-theorem}.}

1. $(Q)_{\Sigma,B} \ := \ soundChase(B,Q,\Sigma); $ 
  
2. $\Sigma^{max}_B(Q,\Sigma) \ := \ \Sigma ;$
  
3. \For{each $\sigma$ in $\Sigma$}
{
4.       \If{$soundChaseStep(\sigma,B,(Q)_{\Sigma,B}) = false$}
      {
5.     $\Sigma^{max}_B(Q,\Sigma) \ := \ \Sigma^{max}_B(Q,\Sigma) - \{ \sigma \} ;$
      }

}

6.  {\bf return} $\Sigma^{max}_B(Q,\Sigma)$;

\end{algorithm}
\vspace{-0.1cm}

The algorithm begins (line 1 of the pseudocode) by computing the result  $(Q)_{\Sigma,B}$ of sound chase of $Q$ using $\Sigma$ under bag semantics ($B$). This result exists and is unique by Theorem~\ref{uniqueness-theorem}. Then the algorithm removes from the set $\Sigma$ all dependencies that are unsoundly applicable to $(Q)_{\Sigma,B}$, see lines 2-5 of the pseudocode. Procedure  $soundChaseStep(\sigma,B,(Q)_{\Sigma,B})$ (line 4) returns $true$ if and only if the bag-chase step using $\sigma$ on $(Q)_{\Sigma,B}$ is sound by Theorem~\ref{bag-chase-sound-theorem}. 

We obtain the following result by construction of algorithm {\sc Max-Bag-$\Sigma$-Subset}. 
\vspace{-0.2cm}

\begin{theorem}{{\bf (Correctness and complexity of  {\sc Max-Bag-$\Sigma$-Subset})}}
\label{correct-max-bag-sigma-theorem}
Given a CQ query $Q$ and set of embedded dependencies $\Sigma$, such that there exists a  {\em set-}chase result $(Q)_{\Sigma,S}$ for $Q$ and $\Sigma$. Then algorithm {\sc Max-Bag-$\Sigma$-Subset}  returns in finite time the set $\Sigma^{max}_B(Q,\Sigma)$ specified in Theorem~\ref{sigma-max-theorem}. If dependencies $\Sigma$ are weakly acyclic, then the  runtime of the algorithm is polynomial in the size of $Q$ and exponential in the size of $\Sigma$. 
\end{theorem}

\vspace{-0.5cm}



\section{$\Sigma$-Equivalence Tests for CQ and CQ-Aggregate Queries}


\nop{
\reminder{
Then theorems on dependency-free equivalence tests.

Then on extending C\&B templates to bag and bag-set semantics.

Then on aggregate queries: aggregate analog of Theorems~\ref{bag-chase-equiv-theorem} and~\ref{bag-set-chase-equiv-theorem}!

Then on CQAC queries? --- my sound extension?

Then on view-based versions???

Then on unchase???
}
} 

We begin this section by providing equivalence tests for CQ queries in presence of embedded dependencies under bag and bag-set semantics, see Section~\ref{equiv-tests-he-subsection}. These results allow us to develop:  (1) Equivalence tests for CQ queries {\it with grouping and aggregation} in presence of embedded dependencies, see Section~\ref{equiv-tests-third-subsection}, and 
(2) Sound and complete (whenever {\em set-}chase on the inputs terminates) algorithms for solving instances of the CQ class of the Query-Reformulation Problem under each of bag and bag-set semantics, as well as for the CQ-aggregate class of the problem, see Section~\ref{equiv-tests-second-subsection}. (Recall that throughout the paper we assume that all given sets of embedded dependencies are finite and regularized.) 


\vspace{-0.1cm}

\subsection{Equivalence Tests for CQ Queries}
\label{equiv-tests-he-subsection}

The main results of this section for CQ queries, Theorems~\ref{bag-chase-equiv-theorem} and~\ref{bag-set-chase-equiv-theorem}, are the analogs, for bag and bag-set semantics, of the dependency-free test of Theorem~\ref{chase-theorem} for equivalence of CQ queries under set semantics and under embedded dependencies. 

\vspace{-0.2cm}

\begin{theorem}
\label{bag-chase-equiv-theorem}
Given CQ queries $Q$ and $Q'$, and a set of embedded dependencies $\Sigma$ 
such that there exist {\em set-}chase results $(Q)_{\Sigma,S}$ for $Q$ and $(Q')_{\Sigma,S}$ for $Q'$.  Then 
$Q \equiv_{\Sigma,B} Q'$ if and only if $(Q)_{\Sigma,B} \equiv_{B} (Q')_{\Sigma,B}$ in the absence of all   dependencies other than the set-enforcing dependencies on stored relations.\footnote{See Theorem~\ref{cv-updated-thm} and discussion of Theorem~\ref{uniqueness-theorem}.}
\end{theorem}

\vspace{-0.4cm}

\begin{theorem}
\label{bag-set-chase-equiv-theorem}
Given CQ queries $Q$ and $Q'$, and a set of embedded dependencies $\Sigma$ 
such that there exist {\em set}-chase results $(Q)_{\Sigma,S}$ for $Q$ and $(Q')_{\Sigma,S}$ for $Q'$.  Then 
$Q \equiv_{\Sigma,BS} Q'$ if and only if $(Q)_{\Sigma,BS} \equiv_{BS} (Q')_{\Sigma,BS}$ in the absence of dependencies.
\end{theorem}

\vspace{-0.2cm}




\nop{
\reminder{
NB! 11/19/08

{\it Note.} Recall that, for the result of sound chase $(Q)_{\Sigma,B}$ of query $Q$ using dependencies $\Sigma$, it is possible for the canonical database of   $(Q)_{\Sigma,B}$ {\em not} to satisfy some of the dependencies in $\Sigma$. (See Theorem~\ref{sigma-max-theorem}.) Thus, ...?

into section on dependency-free equivalence test, put analogy with set-semantics $\equiv_{\Sigma,S}$ (all pre-applicable $\sigma$'s are satisfied vacuously; thus we can change the set of pre-applicable $\sigma$'s while still preserving $\equiv_{\Sigma,S}$!!!)
}
} 

The proofs of Theorems~\ref{bag-chase-equiv-theorem} and~\ref{bag-set-chase-equiv-theorem} follow from Proposition~\ref{chase-termination-prop} 
and from Theorem~\ref{uniqueness-theorem} and its analog for bag-set semantics. See Appendix~\ref{appendix-f} for the details.




We now formulate Proposition~\ref{sigma-b-bs-s-implic-prop}, which is the dep-\linebreak endency-based version of Proposition~\ref{b-bs-s-implic-prop}. The proof of Proposition~\ref{sigma-b-bs-s-implic-prop} can be found in Appendix~\ref{dep-b-bs-implic-appendix}.


\vspace{-0.2cm}

\begin{proposition}
\label{sigma-b-bs-s-implic-prop}
For CQ queries $Q$ and $Q'$ and set of embedded dependencies $\Sigma$, such  that there exists the {\em set-}chase result for each of $Q$ and $Q'$ using $\Sigma$. Then  (1) $(Q)_{\Sigma,B} \equiv_{B} (Q')_{\Sigma,B}$, in the absence of all dependencies other than the set-enforcing constraints on  stored relations, implies $(Q)_{\Sigma,BS} \equiv_{BS} (Q')_{\Sigma,BS}$, and (2) $(Q)_{\Sigma,BS}$ $\equiv_{BS} (Q')_{\Sigma,BS}$ implies $(Q)_{\Sigma,S} \equiv_{S} (Q')_{\Sigma,S}$.
\end{proposition}

\vspace{-0.2cm}

Observe that queries $(Q)_{\Sigma,B}$, $(Q)_{\Sigma,BS}$, and $(Q)_{\Sigma,S}$ may be distinct queries for {\em the same} query $Q$ and set $\Sigma$. For an illustration, please see the chase results $Q_1$ through $Q_3$ of query $Q_4$ in Example~\ref{motivating-example}.   

A corollary of Proposition~\ref{sigma-b-bs-s-implic-prop} establishes a {\em set-cont-} {\em ainment} relationship between a CQ query and the results of its sound chase under a given set of embedded dependencies. Please see Appendix~\ref{dep-b-bs-implic-appendix} for a proof. 

\vspace{-0.2cm}

\begin{proposition}
\label{cor-dep-subset-prop}
For $(Q,\Sigma)$ that satisfy conditions of Thm.~\ref{sigma-max-theorem}, 
$(Q)_{\Sigma,S} \sqsubseteq_S (Q)_{\Sigma,BS} \sqsubseteq_S (Q)_{\Sigma,B} \sqsubseteq_S Q$. 
\end{proposition}

\vspace{-0.2cm}

Queries $Q_4$, $Q_3 = (Q_4)_{\Sigma,B}$, $Q_2 = (Q_4)_{\Sigma,BS}$, and $Q_1 = (Q_4)_{\Sigma,S}$ of Example~\ref{motivating-example} provide an illustration.


\vspace{-0.2cm}

\subsection{Equivalence Tests for Aggregate Queries}
\label{equiv-tests-third-subsection}

We now provide dependency-free tests for equivalence of CQ queries with grouping and aggregation  under embedded dependencies. The results of this subsection are immediate from Theorems~\ref{chase-theorem}, \ref{equival-aggr-theorem},  and~\ref{bag-set-chase-equiv-theorem}. 

\vspace{-0.3cm}

\begin{theorem}
\label{aggr-chase-equiv-theorem}
Given compatible aggregate queries $Q$ and $Q'$, and a set of embedded dependencies $\Sigma$ 
such that there exist {\em set-}chase results $(\breve{Q})_{\Sigma,S}$ for the core $\breve{Q}$ of $Q$ and $(\breve{Q}')_{\Sigma,S}$ for the core $\breve{Q}'$ of $Q'$.  Then 
(1) For $max$ or $min$ queries $Q$ and $Q'$, $Q \equiv_{\Sigma} Q'$ if and only if $(\breve{Q})_{\Sigma,S} \equiv_{S} (\breve{Q}')_{\Sigma,S}$ in the absence of   dependencies. 
(2) For $sum$ or $count$ queries $Q$ and $Q'$, $Q \equiv_{\Sigma} Q'$ if and only if $(\breve{Q})_{\Sigma,BS} \equiv_{BS} (\breve{Q}')_{\Sigma,BS}$ in the absence of   dependencies.
\end{theorem}

\subsection{Sound and Complete Reformulation of\\ CQ and CQ-Aggregate Queries}
\label{equiv-tests-second-subsection}


Theorems~\ref{bag-chase-equiv-theorem} and~\ref{bag-set-chase-equiv-theorem} allow us to extend the algorithm C\&B of~\cite{DeutschPT06} to (a) reformulation of CQ queries in presence of embedded dependencies under bag or bag-set semantics, and to (b) reformulation of CQ queries with grouping and aggregation in presence of embedded dependencies. Our proposed algorithm  {\sc Bag-C\&B} returns $\Sigma$-minimal reformulations $Q'$ of  CQ query $Q$ such that $Q' \equiv_{\Sigma,B} Q$ under the given embedded dependencies $\Sigma$. The only modifications to C\&B that are required to obtain {\sc Bag-C\&B} are (i)  to replace the set-chase procedure by the {\em sound bag-chase} procedure as defined in this paper, and (ii)  to replace the dependency-free equivalence test of Theorem~\ref{chase-theorem} by the test of Theorem~\ref{bag-chase-equiv-theorem}. The algorithm   {\sc Bag-Set-C\&B}  for the case of bag-set semantics is obtained in an analogous fashion. 

We have also developed algorithms that accept sets of embedded dependencies and CQ queries with grouping and aggregation:  {\sc Max-Min-C\&B} accepts CQ queries with aggregate function $max$ or $min$, and {\sc Sum-Count-C\&B} accepts  CQ queries with aggregate function $sum$ or $count$. {\sc Max-Min-C\&B} uses C\&B to obtain all $\Sigma$-minimal reformulations $Q' \equiv_{\Sigma,S} \breve{Q}$ of the core $\breve{Q}$ of the input query $Q$, and for each such query $Q'$ returns a query $Q''$ whose head is the head of $Q$ and whose body is the body of $Q'$. {\sc Sum-Count-C\&B} works analogously, except that it uses {\sc Bag-Set-C\&B} to produce queries $Q' \equiv_{\Sigma,BS} \breve{Q}$. By Theorem~\ref{aggr-chase-equiv-theorem}, for each output $Q''$ of {\sc Max-Min-C\&B} or of {\sc Sum-Count-C\&B}  it holds that $Q'' \equiv_{\Sigma} Q$ whenever set-chase of $Q$ using $\Sigma$ terminates.  

All our algorithms are sound and complete whenever {\em set-}semantics chase of $Q$ using $\Sigma$ terminates.  

\vspace{-0.1cm}

\begin{theorem}
\label{bag-c-and-b-theorem}
Given CQ query $Q$ and set $\Sigma$ of embedded dependencies such that {\em set} chase 
of $Q$ under $\Sigma$ terminates in finite time. Then {\sc Bag-C\&B} returns all $\Sigma$-minimal reformulations $Q'$ such that $Q' \equiv_{\Sigma,B} Q$. 
\end{theorem}

\vspace{-0.1cm}

The analogs of Theorem~\ref{bag-c-and-b-theorem} for (a) CQ queries under bag-set semantics, and for (b) aggregate CQ queries can be found in Appendix~\ref{dep-b-bs-implic-appendix}. All the theorems follow from the soundness and completeness of C\&B of~\cite{DeutschPT06} (see Appendix~\ref{c-and-b-appendix}) and from the results of this paper.





\nop{

\section{Aggregate Queries and Rewriting Templates}
\label{aggr-prelim-section}

\reminder{******* Give here defs of: unaggr core of aggr query; $R^{exp}$ (say that we denote by $R^{exp}$ the $R^{unfolding}$ that was defined in~\cite{CohenNS99})}

We now give a brief overview of the view-based rewriting templates for aggregate queries that we use in this paper; please see~\cite{AfratiRadaIcdt05} for the details. 

Such rewriting templates arise naturally in practical applications, as evidenced by, e.g., the discussion in~\cite{DeHaan05}. 

Note that some of our templates necessitate use of {\it bag-}valued views, hence justification of considering {\it bag} semantics for query evaluation in our work.

\begin{example}
\label{aggr-rewr-template-example}
Consider two aggregate queries, $max$ query $Q_1$ and $sum$ query $Q_2$, as well as views defined on stored relations 
\begin{tabbing}

\end{tabbing}
\end{example}

} 

\vspace{-0.1cm}

\section{Related Work}
\label{related-work-section}

Chandra and Merlin~\cite{ChandraM77} developed the NP-complete containment test of two CQ queries under set semantics. This test has been used in optimization of CQ queries, as well as in developing algorithms for rewriting queries (both equivalently and nonequivalently) using views. Please see~\cite{DeutschPT06,LevyAquvSurvey,ChenLiEncyclop,Ullman00} for discussions of the state of the art and of the numerous practical applications of query rewriting using views. 

The problem of developing tests for equivalence of CQ queries under bag and bag-set semantics was solved by Chaudhuri and Vardi in~\cite{VardiBagsPods93}.  
The results on {\em containment} tests for CQ queries under bag semantics have proved to be more elusive. Please see Jayram and colleagues \cite{KolaitisPods06} for original undecidability results for containment of CQ queries with inequalities under bag semantics. The authors point out that it is not  known whether the problem of bag containment for {\em CQ} queries is even decidable.  On the other hand, the problem of containment of CQ queries under bag-set semantics reduces to the problem of containment of aggregate queries with aggregate function {\tt count(*).} 
The latter problem is solvable using the methods proposed in~\cite{CohenNS03}. 

Studies of dependencies have been motivated by the goal of good database-schema design. 
See~\cite{AbiteboulHV95,DeutschPods08} for overviews and references on dependencies and chase. 
%
%
In~\cite{DeutschDiss}, Deutsch developed chase methods for bag-specific constraints (UWDs), 
and proved  completeness of the view-based version of the Chase and Backchase algorithm (C\&B, \cite{DeutschPT06}) for mixed semantics and for set and bag dependencies, in case where all given tuple-generat-\linebreak ing dependencies are UWDs. In contrast, the algorithm in~\cite{DuschkaG97} is complete in presence of just functional dependencies. Algorithms that are complete in the absence of dependencies are given in~\cite{LevyMSS95} for set semantics, in~\cite{ChaudhuriKPS95} for bag semantics, and in~\cite{GouCK06} for bag-set semantics. 
Finally, Cohen in~\cite{Cohen06} presented an equivalence test for CQ queries in presence of inclusion dependencies,\footnote{An {\em inclusion dependency} is a tgd with a single relational atom on each of the left-hand side and right-hand side.} for the cases of bag-set semantics and of the semantics where queries are evaluated on set-valued databases using both bag-valued and set-valued intermediate results. 


\vspace{-0.3cm}

{\small 
\bibliographystyle{abbrv}
\bibliography{techreport062109}  
}

\appendix

\section{The C\&B Algorithm of  [10]}
\label{c-and-b-appendix}




In this section of the appendix we give an overview of the Chase and Backchase (C\&B) algorithm by Deutsch and colleagues, see~\cite{DeutschPT06} for the details. 
Under set semantics for query evaluation and  given a CQ query $Q$, C\&B  outputs all equivalent  $\Sigma$-minimal conjunctive reformulations of $Q$ in presence of the given embedded dependencies $\Sigma$, whenever chase of $Q$ under $\Sigma$ terminates in finite time. 


C\&B proceeds in two phases. The first phase of C\&B, its {\it chase phase}, does chase of $Q$ using $\Sigma$ under set semantics, to obtain terminal chase result $(Q)_{\Sigma,S}$. This output of  the chase phase is called the {\it universal plan} $U$ for $Q$. Note that by construction of $U$, $Q \equiv_{\Sigma,S} U$. 

The second phase of C\&B, its {\it backchase phase}, proceeds as follows:

\begin{enumerate}
	\item Iterate over all queries $U'$ whose head is $head(U)$ and whose  body is not empty and is $body(U)$ with zero or more atoms dropped. 
 
	\item Chase each $U'$ using $\Sigma$, to obtain terminal chase result $(U')_{\Sigma,S}$. 

	\item C\&B outputs each $U'$ such that for the terminal result $(U')_{\Sigma,S}$  of chasing the candidate reformulation $U'$ under $\Sigma$ (under set semantics), it holds that $(U')_{\Sigma,S} \equiv_{S} U$, that is, each $U'$ for which by Theorem~\ref{chase-theorem} it holds that $U' \equiv_{\Sigma,S} Q$. 

\end{enumerate}

\begin{theorem}{{\bf (C{\&}B is sound and complete)}}
\label{c-and-b-theorem} 
For an arbitrary instance of the Query-Reformulation Problem with a CQ query $Q$, set semantics for query evaluation, and a set of embedded dependencies $\Sigma$ such that chasing $Q$ under $\Sigma$ terminates in finite time, C{\em \&}B outputs all $\Sigma$-minimal conjunctive reformulations $Q'$ of $Q$  such that $Q' \equiv_{\Sigma,S} Q$.
\end{theorem}

The proof of Theorem~\ref{c-and-b-theorem} 
is by construction of C\&B. 

\section{Keys of Relations}
\label{key-app}


This section of the appendix  provides basic definitions for the standard notion of a key of a relation~\cite{GarciaMolinaUW02}. 

\subsection{Attributes and Relations}
\label{attr-app}

Let $\cal U$ be a countably infinite set of {\it attributes.} The {\it universe} $U$ is a finite subset of $\cal U$. A {\it relation schema} $R$ {\it of arity} $k$ is a subset of $U$ of cardinality $k$. A {\it database schema} (or, simply, {\it schema}) $\cal D$ over $U$ is a finite set of relation schemas $\{ R_1,\ldots,R_t \}$ with union $U$, of  arities $k_1,\ldots,k_t$, respectively.

Each attribute $A \in {\cal U}$ has an associated set of {\it values} $\Delta(A)$, called $A$'s {\it domain.} The {\it domain} is the set of values $\Delta = \bigcup_A \Delta(A)$. Let $\cal D$ be a schema over $U$, $R \in {\cal D}$ a relation schema and $X$ a subset of $U$. An {\it X-tuple} $t$ is a mapping from $X$ into $\Delta$, such that each attribute $A \in X$ is mapped to an element of $\Delta(A)$. A (generally bag-valued) {\it relation} $r$ over $R$  is a finite collection of $R$-tuples. 
A {\it database (instance)} $D$ of $\cal D$ is a set of relations, with one relation for each relation schema of $\cal D$.

\subsection{Functional Dependencies and Keys}

Consider a database schema $\cal D$ with $n$-ary relation symbol $P$ such that $n > 1$.  
A {\it functional dependency (fd)}  on relation $P$ in $\cal D$  
is an egd of the form $p(\bar{X},Y,\bar{Z}) \wedge p(\bar{X},Y',\bar{Z}') \rightarrow Y = Y'$, such that predicate $p$ corresponds to relation $P$. 
Here, $Y$ and $Y'$ must be in the same position in the respective atoms, meaning the following. 
Let $Y$ be the $i$th argument of atom $p(\bar{X},Y,\bar{Z})$, for some $ 1 \leq i \leq n$. Then $Y'$ is the $i$th argument of atom $p(\bar{X},Y',\bar{Z}')$. Similarly, we require each element of the vector $\bar{X}$ to be in the same position in each of $p(\bar{X},Y,\bar{Z})$ and $p(\bar{X},Y',\bar{Z}')$.

\begin{definition}{Implied functional dependency}
Let $\sigma$ be an fd on relation $R$, and let $\Sigma$ be a set of fds on $R$. Then $\sigma$ {\em is a functional dependency implied by} $\Sigma$ if $\sigma$ holds on all instances of relation $R$ that satisfy $\Sigma$.
\end{definition}

Standard textbooks (see, e.g.,~\cite{GarciaMolinaUW02}) describe algorithms for solving the problem of finding all fds implied by a given set of dependencies on the schema of a relation. 

Let ${\bf K} = \{ A_{i1},\ldots,A_{ip} \}$ be a nonempty proper subset of the set of attributes of $n$-ary relation $R(A_1,\ldots,A_n)$, with $n > 1$. That is, $1 \leq p < n$  and $A_{ij} \in \{ A_1,\ldots,A_n \}$ for each $j \in \{ 1,\ldots,p \}.$ In the definitions that follow, we will use the following notation: Let $\sigma({\bf K} | A_i)$, for some $i \in \{ 1,\ldots,n \}$ such that  $A_i \notin {\bf K}$, 
denote an fd that equates the values of attribute $A_i$ of $R$  whenever the two $r$-atoms in the left-hand side of $\sigma({\bf K} | A_i)$ agree on the values of all and only attributes in ${\bf K}$. For example, if the schema of $R$ is $R(A,B,C,D)$ and ${\bf K} = \{ A,C \},$ then $\sigma(A,C | B)$ is defined as 
\begin{tabbing}
$\sigma(A,C | B): r(A,B_1,C,D_1) \wedge r(A,B_2,C,D_2) \rightarrow B_1 = B_2.$
\end{tabbing}

\begin{definition}{Superkey of relation}
${\bf K}$ is a {\em superkey of relation} $R$ if for each attribute $A$ in the set $\{ A_1,\ldots,$ $A_n \} - {\bf K},$ it holds that fd $\sigma({\bf K} | A)$ is implied by the set $\Sigma$ of fds on $R$.
\end{definition}

The set of all attributes of $R$ is also a superkey of $R$. 

\begin{definition}{Key of relation}
${\bf K}$ is a {\em key of relation} $R$ if (1) ${\bf K}$ is a superkey of $R$, and (2) for each nonempty proper subset ${\bf K'}$ of ${\bf K},$ ${\bf K'}$ is not a superkey of $R$.
\end{definition}

\section{Tuple IDs for Relations}
\label{appendix-a}



In this section of the appendix  we present a solution to the problem of ensuring, under {\it bag} semantics, that certain base relations are {\it sets} in all database instances. To this end, we provide here a formal framework for {\it tuple IDs}, which are unique tuple identifiers commonly used in implementations of real-life database-management systems~\cite{GarciaMolinaUW02}. 
Our approach to ensuring that some relations are always set valued is to use functional dependencies (Appendix~\ref{key-app}) to force certain relations to be set valued, by restricting tuples with the same ``contents'' (that is, all values with the exception of the tuple ID) to have the same tuple ID.


Assume bag semantics for query evaluation and consider relation symbol $R_i$ in database schema $\cal D$. (Section~\ref{attr-app} has the relevant definitions.) 
We follow the approach taken in implementations of real-life database-management systems~\cite{GarciaMolinaUW02} by incrementing the arity of $R_i$. As a result, the arity of each relation $R_i$ becomes $k_i + 1$ instead of the original $k_i$ as defined in Section~\ref{prelim-section}.\footnote{We emulate the standard implementation practice that tuple IDs be invisible to the users of the database system; that is, in our approach the user assumes that the arity of each relation $R_i$ is still $k_i$.} 

Let ${\cal D}'$ be the schema resulting from such arity modification in $\cal D$ for each relation $R_i$. 
By $D'$ we denote instances of ${\cal D}'$. In the schema of $R_i$ in ${\cal D}'$, let the last attribute of $R_i$ be the attribute for the tuple ID. The values of all tuple IDs are required to be distinct in all instances of ${\cal D}'$, which is formally specified as follows.

\begin{definition}{Tuple ID.}
\label{tuple-id-def}
For a relation symbol $R_i$ of arity $k_i + 1$ in database schema ${\cal D}'$, let queries $Q^{R_i}_{tid}$ and   $Q^{R_i}_{vals}$ be as follows:
\begin{tabbing}
$Q^{R_i}_{tid}(X_{k_i + 1}) \ :- \ R_i(X_1,\ldots,X_{k_i},X_{k_i + 1}).$ \\
$Q^{R_i}_{vals}(X_1,\ldots,X_{k_i}) \ :- \ R_i(X_1,\ldots,X_{k_i},X_{k_i + 1}).$
\end{tabbing} 
Then  the $(k_i + 1)$st attribute of $R_i$ in ${\cal D}'$ is {\em the tuple ID for} $R_i$ if in all instances $D'$ of ${\cal D}'$, the following relationship holds between the relations $Q^{R_i}_{tid}(D',B)$ and $Q^{R_i}_{vals}(D',B)$: 
\begin{tabbing}
$| coreSet(Q^{R_i}_{tid}(D',B)) | = | Q^{R_i}_{vals}(D',B) | .$
\end{tabbing} 

Here, $coreSet({\bf B})$ denotes the core-set of bag ${\bf B}$,   
and $|{\bf B}|$ denotes cardinality of ${\bf B}$. 
\end{definition} 


We now study the relationship between instances $D'$ of ${\cal D}'$ and instances $D$ of $\cal D$. Suppose that for relation $R_i$ of arity $k_i+1$ in ${\cal D}'$, the last attribute of $R_i$ is the tuple ID of $R_i$. By definition of tuple IDs, for each instance $D$ of $\cal D$, relation $R_i$ in $D$ can be obtained from some instance $D'$ of ${\cal D}'$, by evaluating query $Q^{R_i}_{vals}$ under bag semantics on $R_i$ in $D'$: 
\begin{displaymath}
Q^{R_i}_{vals}(X_1,\ldots,X_{k_i}) \ :- \ R_i(X_1,\ldots,X_{k_i},X_{k_i + 1}).
\end{displaymath} 

Now suppose that in (the original) schema $\cal D$, a relation with symbol $R_i$ and arity $k_i$ is required to be set valued in all instances of $\cal D$. We enforce this requirement by the functional dependency 
\begin{tabbing}
$\sigma^{R_i}_{tid}: \ R_i(X_1,\ldots,X_{k_i},X_{k_i + 1}) \wedge$ \\
$ \ \ \ \ \ \ \ \  R_i(X_1,\ldots,X_{k_i},Y_{k_i + 1}) \rightarrow X_{k_i + 1} = Y_{k_i + 1} $
\end{tabbing} 
on $R_i$ in schema ${\cal D}'$. This functional dependency enforces the same tuple ID for each pair of tuples that agree on the values of all other attributes of $R_i$. In conjunction with Definition~\ref{tuple-id-def}, which ensures uniqueness of each tuple ID within each instance of ${\cal D}'$, $\sigma^{R_i}_{tid}$ enforces that the answer to query $Q^{R_i}_{vals}$ (i.e., $R_i$ in schema $\cal D$) be set valued when computed under bag semantics.

In the context of Example~\ref{motivating-example}, in presence of tuple IDs we could formally define dependency $\sigma_6$ as an egd:
\begin{tabbing}
$\sigma_6: \ t(X,Y,Z,U) \wedge t(X,Y,Z,W) \rightarrow U = W.$
\end{tabbing} 
Here, the fourth attribute of relation $T$ is the tuple-ID attribute.

\section{Proof of Theorem 4.2}
\label{set-bag-appendix}


This section of the appendix provides a proof of Theorem~\ref{cv-updated-thm}. We first supply the details of Example~\ref{extend-necess-example}.  

\begin{example}
\label{ext-extend-necess-example}
To show that query $Q_3$ of Example~\ref{motivating-example} is not bag equivalent to query $Q_5$, 
\begin{tabbing}
$Q_3(X) \ :- \ p(X,Y), t(X,Y,W), s(X,Z).$ \\
$Q_5(X) \ :- \ p(X,Y), t(X,Y,W), s(X,Z), s(X,Z).$ 
\end{tabbing}
we construct a bag-valued database $D$, with the following relations:  $P = \{\hspace{-0.1cm}\{ (1,2) \}\hspace{-0.1cm}\}$, $R = \emptyset$, $S = \{\hspace{-0.1cm}\{ (1,3), (1,3) \}\hspace{-0.1cm}\}$,  $T = \{\hspace{-0.1cm}\{ (1,2,5) \}\hspace{-0.1cm}\}$, and $U = \emptyset$. 
On this database $D$, the answer to $Q_3$ is $Q_3(D,B) = \{\hspace{-0.1cm}\{ (1), (1) \}\hspace{-0.1cm}\}$, whereas $Q_5(D,B) = \{\hspace{-0.1cm}\{ (1), (1), (1), (1) \}\hspace{-0.1cm}\}$, by rules of bag semantics. From the fact that $Q_3(D,B)$ and $Q_5(D,B)$ are not the same bags, we conclude that 
bag equivalence $Q_3 \equiv_B Q_5$ does not hold. 

At the same time, by Theorem~\ref{cv-updated-thm} it holds that $Q_3$ and $Q_5$ are bag equivalent on all databases where relation $S$ is required to be a set. 
\end{example}

We now prove Theorem~\ref{cv-updated-thm}. The If part of the proof is straightforward. For the Only-If part, we argue that the only way for $Q_1$ and $Q_2$ to be bag equivalent under the set-enforcing constraints of database schema $\cal D$ is for $Q_1$ and $Q_2$ to satisfy the conditions of Lemma~\ref{bag-stubborn-lemma}. The proof of Lemma~\ref{bag-stubborn-lemma} completes the proof of  Theorem~\ref{cv-updated-thm}, by showing by contrapositive that bag equivalence of  $Q_1$ and $Q_2$ under the set-enforcing constraints of database schema $\cal D$ has to entail isomorphism of the queries $Q'_1$ and $Q'_2$ defined in the statement of Theorem~\ref{cv-updated-thm}. 

\begin{proof}{(Theorem~\ref{cv-updated-thm})}

{\bf If.} 
Let  database schema $\cal D$ have a relation symbol $P$, such that  the relation for $P$ is  set valued in all (bag-valued) instances $D$ over $\cal D$. (Appendix~\ref{appendix-a} 
provides an approach to enforcing this set-valuedness constraint using functional dependencies that involve tuple IDs.) Consider an arbitrary CQ query  $Q_1$  that  has a subgoal with predicate $p$ corresponding to relation $P$; w.l.o.g. let the subgoal be $p(\bar{W})$. Let $Q_2$ be a CQ query obtained by adding to the body of $Q_1$ a duplicate of $p(\bar{W})$. 

We argue that for $Q_1$ and $Q_2$ as described above, it holds that $Q_1 \equiv_B Q_2$ under the set-enforcing dependencies of the schema $\cal D$. (The claim of the If direction of the theorem is immediate from this observation.) 
Indeed, consider an arbitrary instance $D$ of database schema $\cal D$, such that $D$ satisfies the set-enforcing dependencies of the schema $\cal D$. From the definition of bag semantics for query evaluation it follows that each assignment satisfying the body of $Q_1$ w.r.t. $D$ is also a satisfying assignment for the body of $Q_2$ w.r.t. $D$, and vice versa. Further, each such satisfying assignment $\gamma$ maps $p(\bar{W})$, in the body of $Q_1$, into a {\it single} tuple $t$ in relation $P$ in $D$, and similarly  $\gamma$ maps {\it both copies} of $p(\bar{W})$, in the body of $Q_2$, into the same single tuple $t$, due to relation $P$ being set valued in the database $D$. It follows that each such satisfying assignment $\gamma$ contributes to each of $Q_1(D,B)$ and $Q_2(D,B)$ {\it the same} number of tuples under bag semantics for query evaluation. The claim of the If direction of Theorem~\ref{cv-updated-thm} is immediate from the above observation. 

\mbox{}

{\bf Only-If.} The proof is by contrapositive. For two CQ queries $Q_1$ and $Q_2$, let $Q_1 \equiv_B Q_2$ hold in the absence of all dependencies other than the set-enforcing dependencies of the schema $\cal D$. Consider queries $Q'_1$ and $Q'_2$ defined in the statement of Theorem~\ref{cv-updated-thm}. We assume that $Q'_1$ and $Q'_2$ are not isomorphic, and obtain from this assumption that $Q_1$ and   $Q_2$ are not bag equivalent on at least one database that satisfies the set-enforcing dependencies of schema $\cal D$, in contradiction with what we are given. 

W.l.o.g., let $s$ be a subgoal of query $Q'_1$ such that either $Q'_2$ has no subgoals with the predicate of $s$, or $Q'_2$ has fewer (but still a positive number of) subgoals with the predicate of $s$ than $Q'_1$  does. Consider first the case where $Q'_2$ has no subgoals with the predicate of $s$; it follows from the construction of queries $Q'_1$ and $Q'_2$ that $Q_2$ does not have subgoals with the predicate of $s$ either, whereas $Q_1$ has at least one occurrence of subgoal with the predicate of $s$. 
%
%
Observe that in this case, {\em set} equivalence between $Q_1$ and $Q_2$ does not hold by the results of~\cite{ChandraM77}. From the result of~\cite{VardiBagsPods93} (see Proposition~\ref{b-bs-s-implic-prop} in this current paper) that bag equivalence implies set equivalence, it follows immediately that bag equivalence of $Q_1$ and $Q_2$ cannot hold either, {\em in presence} of the set-enforcing dependencies in the schema $\cal D$. (This follows from the fact that $Q_2 \sqsubseteq_S \hspace{-0.5cm / \hspace{0.5cm}} Q_1$ implies existence of a {\em set-valued} database on which $Q_2$ under {\em set} semantics produces a tuple $t$, such that $t$ is {\em not} in the {\em set-}semantics answer to $Q_1$ on the same database.) Thus, we have arrived at a contradiction with our assumption that $Q_1 \equiv_B Q_2$ on all databases satisfying the set-enforcing dependencies of the schema $\cal D$. 

We now consider the remaining case concerning the number in $Q'_2$ of subgoals with the predicate of $s$, that is the case where $Q'_2$ has fewer (but still a positive number of) subgoals with the predicate of $s$ than $Q'_1$  does. Suppose first that there is no {\em bag-set} equivalence between $Q_1$ and $Q_2$. That is, by Theorem~\ref{cv-theorem} we assume that the canonical representations of $Q'_1$ and of $Q'_2$ (which are the same as the canonical representations of $Q_1$ and of $Q_2$, respectively) are not isomorphic. Then similarly to the previous case considered in this proof, from Proposition~\ref{b-bs-s-implic-prop} we obtain immediately the contradiction to $Q_1 \equiv_B Q_2$ under the set-enforcing dependencies of schema $\cal D$. (Similarly to the case above, $Q_1 \equiv_{BS} Q_2$ would have to be violated on a {\em set-}valued database, therefore the set-enforcing dependencies of the schema $\cal D$ would be satisfied in that $Q_1 \equiv_{B} Q_2$ would be violated on the {\em same} database.) 

Thus, for the rest of this proof we assume that (1) $Q_1 \equiv_{BS} Q_2$, and (2) $Q'_1 \equiv_{BS} Q'_2$ (from $Q_1 \equiv_{BS} Q_2$ and by construction of $Q'_1$ and $Q'_2$). That is, for both pairs of queries the canonical representations are isomorphic. 
Under these restrictions, the only way $Q'_1$ and $Q'_2$ can be nonisomorphic is the case where $Q'_1$ (w.l.o.g.) has more subgoals (than $Q'_2$) whose predicate corresponds to a relation, say $R$, that is {\em not} required to be a set in all instances of schema $\cal D$. (Indeed, if $Q_1$ and $Q_2$ have this number-of-subgoals discrepancy for a predicate whose relation {\em is} required to be a set in all instances of $\cal D$, then $Q'_1$ and $Q'_2$ must have the same number of such subgoals by $Q_1 \equiv_{BS} Q_2$ and by construction of $Q'_1$ and $Q'_2$.) Note that in this case, relation symbol $R$ must belong to ${\cal D} - \{ P_1,\ldots,P_k \}$ (``$-$'' is set difference), and thus the subset relationship  $\{ P_1,\ldots,P_k \} \subseteq {\cal D}$ is proper in this case, that is $\{ P_1,\ldots,P_k \} \subset {\cal D}$. Recall that $\{ P_1,\ldots,P_k \}$ is the maximal subset of $\cal D$ such that all symbols in $\{ P_1,\ldots,P_k \}$ correspond to relations required to be set valued in all instances of $\cal D$. 

We finish the proof of Theorem~\ref{cv-updated-thm} by proving Lemma \ref{bag-stubborn-lemma}, which constructs a database $D$ satisfying the set-enforcing dependencies of schema $\cal D$. By construction, database $D$ is a counterexample to $Q_1 \equiv_B Q_2$ (on databases satisfying the set-enforcing dependencies of schema $\cal D$), whenever  $Q'_1$ has more subgoals (than $Q'_2$) whose predicate corresponds to a relation that is not required to be a set in all instances of schema $\cal D$. 
\end{proof}

\begin{lemma}
\label{bag-stubborn-lemma}
Let $\cal D$, $\{ P_1,\ldots,P_k \} \subset {\cal D}$, $Q_1$, $Q_2$, $Q'_1$, and $Q'_2$ be as specified in Theorem~\ref{cv-updated-thm}, and let $Q_1 \equiv_{BS} Q_2$. Let $R$ be a relation symbol in the set ${\cal D} - \{ P_1,\ldots,P_k \}$; that is, relation $R$ is not required to be a set in all instances of $\cal D$. Suppose that $Q'_1$ has strictly more subgoals whose predicate corresponds to $R$ than $Q'_2$ does. Then there exists an instance $D$ of $\cal D$ such that all of relations  $P_1,\ldots,P_k$ are set valued in $D$, and such that $Q_1(D,B)$ is not the same bag as $Q_2(D,B)$. 
\end{lemma}

By the above characterization, database $D$ is a counterexample to queries $Q_1$ and $Q_2$ being bag equivalent on all instances of $\cal D$ that satisfy the set-enforcing restrictions of schema $\cal D$. 

The intuition for the proof of Lemma~\ref{bag-stubborn-lemma} is as follows. Let query $Q_1$ have $n_1 > 1$ subgoals whose predicate corresponds to relation $R$, such that $R$ is not required to be set valued in instances of schema $\cal D$. (Part of the proof is to show that by the properties of this relation symbol $R$ and by construction of $Q'_1$ from $Q_1$, it holds that $Q_1$ and $Q'_1$ have exactly the same number of subgoals whose predicate corresponds to $R$. We make the same observation about $Q_2$ and $Q'_2$.) Further, let query $Q_2$ have a positive number (by proof of Theorem~\ref{cv-updated-thm}) $n_2 < n_1$ of subgoals whose predicate corresponds to $R$. We build a database $D$ on which $Q_1$ produces at least $m^{(n_1)}$ copies of some  (distinct) tuple $t^*$, with the positive integer value of $m$ to be determined. We then ``let'' $Q_2$ have as many satisfying assignments for the body of $Q_2$ w.r.t. this database $D$ as possible. That is, we assume the best case for $Q_2$ of producing as many tuples on database $D$ as possible. We then show that if the value of $m$ is chosen in a certain way, then the number $m^{(n_1)}$ of copies of tuple $t^*$ in the bag $Q_1(D,B)$ is greater than the maximal (i.e., best-case) number $N$ of {\em all} tuples (counting all duplicate tuples as separate tuples) that can be contributed by $Q_2$ to the bag $Q_2(D,B)$. The reason that we can make such a choice of the value of $m$ is that this maximal number $N$ grows asymptotically as $m^{(n_2)}$, with $0 < n_2 < n_1$, whereas the number  of copies of tuple $t^*$ in the bag $Q_1(D,B)$ is $m^{(n_1)}$.

\begin{proof}{(Lemma~\ref{bag-stubborn-lemma})}
%
Let $n_1$ be the number of subgoals in $Q'_1$ whose predicate corresponds to $R$, and let $n_2$ be the number of subgoals in $Q'_2$ whose predicate corresponds to $R$; $n_2 > 0$ by $Q_1 \equiv_{BS} Q_2$. By our assumption, $n_1 \geq n_2 + 1$.   By construction of $Q'_1$, $Q_1$ has the same number $n_1$ of subgoals whose predicate corresponds to $R$ as $Q'_1$ does; we make the same observation about the relationship between the number $n_2$ of subgoals in $Q_2$ whose predicate corresponds to $R$ and the (same) number $n_2$ of subgoals of the same type in $Q'_2$. (See proof of Theorem~\ref{cv-updated-thm} for the details of the argument.)

Let $D'$ be the (set-valued by definition, see Section~\ref{basics-sec}) canonical database for the canonical representation of $Q'_1$. (From the proof of Theorem~\ref{cv-updated-thm}, we have that $Q_1$, $Q'_1$,  $Q_2$, and $Q'_2$ all have the same canonical representation.) We construct from $D'$ our counterexample database $D$ as follows.

1. For each relation symbol $S$ in ${\cal D} - \{ R \}$, the relation $S$ in $D$ is the same as the relation $S$ in $D'$. By construction of $D'$, all the relations in $\{ P_1,\ldots,P_k \}$ are {\em set valued} in database $D$. Thus, database $D$ satisfies the set-enforcing restrictions of the schema $\cal D$. 

2. We build relation $R$ in $D$ by ``putting together'' $m> 0$ copies of relation $R$ in $D'$, with the value of $m$ to be determined shortly. That is, for each tuple $t$ such that $t$ is in the {\em set-valued} relation $R$ in $D'$, relation $R$ in $D$ has $m$ copies of tuple $t$; further, $R$ in $D$ has no other tuples. 


By definition of bag semantics for query evaluation, see Section~\ref{bag-bag-set-defs}, the bag $Q_1(D,B)$ has at least $m^{(n_1)}$ copies of some individual tuple. Indeed, consider the assignment mapping $\gamma$ from  $Q_1$ to $D$ such that $\gamma$ was used to {\em generate} the canonical database $D'$ of the canonical representation of $Q_1$. (See Section~\ref{basics-sec} for the description of the process of construction of a canonical database for a CQ query.) Observe that $\gamma$ is a satisfying assignment for the body of $Q_1$ w.r.t. database $D$. The assignment $\gamma$ maps  each of the $R$-subgoals of $Q_1$ to at least $m$ tuples of $R$, by construction of relation $R$ in $D$, and $\gamma$ maps each non-$R$ subgoal (if any) of $Q_1$ to exactly one tuple. Thus, for the tuple $t^* = \gamma(\bar{X}) \in Q_1(D,B)$, where $Q_1(\bar{X})$ is the head of the query $Q_1$, the multiplicity of $t^*$ in $Q_1(D,B)$ is at least $m^{(n_1)}$. (The ``at least'' part comes from the possibility that extra copies of the tuple $t^*$ could be contributed to the bag $Q_1(D,B)$ by one or more satisfying assignments $\gamma'$ for the body of $Q_1$ w.r.t. database $D$, such that each such $\gamma'$ is {\em not} identical to the assignment $\gamma$.)

At the same time, we show that the total size of the bag $Q_2(D,B)$ cannot exceed  
\begin{equation}
\label{qtwo-tuples-eq}
n_1^{(2n_2)} \times n_4^{(n_3-n_2)} \times m^{(n_2)}
\end{equation}
tuples, in case the total number $n_3$ of subgoals of query $Q_2$ is greater than $n_2$; $n_4$ is the number of  subgoals of $Q_1$ whose (subgoals') predicate does not correspond to relation symbol $R$.  (By $Q_1 \equiv_{BS} Q_2$ we have that $n_4 > 0$ whenever $n_3 > n_2$.) In this case, we set the value $m^*$ of $m$ to 
\begin{equation}
\label{qtwo-tuples-third-eq}
m^* \ := \ 1 + n_1^{(2n_2)} \times n_4^{(n_3-n_2)}  .
\end{equation}
It follows that 
\begin{equation}
\label{qtwo-tuples-fourth-eq}
(m^*)^{(n_1-n_2)} > n_1^{(2n_2)} \times n_4^{(n_3-n_2)} .
\end{equation}
That is (recall that $0 < n_2 < n_1$), 
\begin{equation}
\label{qtwo-tuples-fifth-eq}
(m^*)^{(n_1)} > n_1^{(2n_2)} \times n_4^{(n_3-n_2)} \times (m^*)^{(n_2)} .
\end{equation}
We conclude that on the database $D$ where the value of $m$ is fixed at $m^*$, the number of copies of tuple $t^*$ in the bag $Q_1(D,B)$ exceeds the number of {\em all} tuples in the bag $Q_2(D,B)$. Therefore, the bag $Q_1(D,B)$ is not the same bag as   $Q_2(D,B)$. 

(In case the total number $n_3$ of subgoals of query $Q_2$ is equal to $n_2$, we show that the bag $Q_2(D,B)$ cannot have more than 
\begin{equation}
\label{qtwo-tuples-second-eq}
n_1^{(2n_2)} \times m^{(n_2)}
\end{equation}
tuples. In this case, we set the value $m^*$ of $m$ to 
\begin{equation}
\label{qtwo-tuples-thirds-eq}
m \ := \ 1 + n_1^{(2n_2)}  .
\end{equation}
It follows that at this value $m^*$ of $m$, we have that 
\begin{equation}
\label{qtwo-tuples-fourths-eq}
(m^*)^{(n_1-n_2)} > n_1^{(2n_2)} .
\end{equation}
That is (recall that $0 < n_2 < n_1$), 
\begin{equation}
\label{qtwo-tuples-fifths-eq}
(m^*)^{(n_1)} > n_1^{(2n_2)} \times (m^*)^{(n_2)} .
\end{equation}
We conclude that on the database $D$ where the value of $m$ is fixed at $m^*$, the number of copies of tuple $t^*$ in the bag $Q_1(D,B)$ exceeds the number of {\em all} tuples in the bag $Q_2(D,B)$. Therefore, the bag $Q_1(D,B)$ is not the same bag as   $Q_2(D,B)$. The proof of this case is straightforward from the proof, see below, of Equation~\ref{qtwo-tuples-eq} for the case where the total number $n_3$ of subgoals of query $Q_2$ is greater than $n_2$.)


We now explain why  the bag $Q_2(D,B)$ cannot be of greater cardinality than the number of tuples specified in Equation~\ref{qtwo-tuples-eq}, in the case where the total number $n_3$ of subgoals of query $Q_2$ is greater than $n_2$. 
The idea of the proof is to ``let'' $Q_2$ have as many satisfying assignments for the body of $Q_2$ w.r.t. database $D$ as possible. That is, we assume the best case for $Q_2$ of producing as many tuples on database $D$ as possible. We take the following specific steps in building the upper bound: 
\begin{enumerate}

	\item We assume the best case for $Q_2$ of the number of satisfying assignments, w.r.t. database $D$, 
for the $n_3-n_2$ subgoals of $Q_2$ whose (subgoals') predicates do not correspond to $R$. The maximal number of such assignments cannot exceed 
\begin{equation}
\label{thru-eq}
n_4^{(n_3-n_2)} . 
\end{equation}
That is, the best case for $Q_2$ is to assume that all of the $n_3-n_2$ subgoals of $Q_2$ have {\em the same} predicate, say predicate $s$ corresponding to the relation symbol $S$, where $S$ may or may not be one of the relation symbols $P_1,\ldots,P_k$ specified in the formulation of this Lemma. We also assume that the $n_4 > 0$ non-$R$ subgoals of $Q_1$ also have the same predicate $s$. Database $D$ has at most $n_4$ tuples in relation $S$ (by construction of canonical databases). We assume the best case for $Q_2$ that each of the $n_3-n_2$ subgoals of $Q_2$ can map independently into each of the (at most) $n_4$ tuples, hence the formula of Equation~\ref{thru-eq}.   

	\item For each of the above $n_4^{(n_3-n_2)}$  assignments, $Q_2$ may have at most  
\begin{equation}
\label{thru-two-eq}
n_1^{(n_2)}  
\end{equation}
satisfying assignments, w.r.t. database $D$, for all the $n_2$ subgoals of $Q_2$ whose predicate corresponds to the relation symbol $R$. The computations are similar to those that we used in explaining Equation~\ref{thru-eq}. 

	\item For each of the $n_1^{(n_2)}$  satisfying assignments, w.r.t. database $D$, for all the $n_2$ subgoals of $Q_2$ whose predicate corresponds to the relation symbol $R$, $Q_2$ can produce on database $D$ at most 
\begin{equation}
\label{thru-three-eq}
(n_1 \times m)^{(n_2)}  
\end{equation}
tuples. We obtain the formula of Equation~\ref{thru-three-eq} by assuming that the evaluation of $Q_2$ admits a Cartesian product of $n_2$ copies of the relation $R$, where relation $R$ has at most $n_1 \times m$ tuples on $D$. 

	\item We combine Equations~\ref{thru-eq}, \ref{thru-two-eq}, and \ref{thru-three-eq}, to obtain that the total number of satisfying assignments for the body of $Q_2$ w.r.t. database $D$ cannot exceed
	\begin{equation}
\label{thru-four-eq}
n_1^{(n_2)} \times n_4^{(n_3-n_2)}  
\end{equation}
(satisfying assignments); and that, further, for each one of these assignments $Q_2$ produces on database $D$ at most 
	\begin{equation}
\label{thru-five-eq}
(n_1 \times m)^{(n_2)} 
\end{equation}
tuples (where each duplicate is counted separately) in the bag $Q_2(D,B)$. (Recall that all relations except $R$ are set valued in database $D$.) We conclude that the total number of tuples (including duplicates) that query $Q_2$ produces on database $D$ is at most
	\begin{equation}
\label{thru-six-eq}
(n_1)^{2(n_2)} \times n_4^{(n_3-n_2)}  \times m^{(n_2)}
\end{equation}
tuples. Equation~\ref{thru-six-eq} gives us an upper bound on the size of the bag $Q_2(D,B)$. Q.E.D.

\end{enumerate}
\vspace{-0.5cm}
\end{proof}


Consider an illustration to the proof of Lemma~\ref{bag-stubborn-lemma}. 


\begin{example}
\label{next-ext-extend-necess-example}
Let CQ queries $Q_7$ and $Q_8$ be defined as 
\begin{tabbing}
$Q_7(X) \ :- \ p(X,Y), r(X), r(X).$ \\
$Q_8(X) \ :- \ p(X,Y), r(X).$ 
\end{tabbing}
in the setting of Example~\ref{motivating-example}. 
To illustrate the proof of Lemma~\ref{bag-stubborn-lemma}, we construct a counterexample database to the claim that $Q_7$ and $Q_8$ are bag equivalent on all databases that satisfy just the set-enforcing dependencies of Example~\ref{motivating-example}. We use the fact that query $Q_7$ has two copies of subgoal $r(X)$, whereas $Q_8$ has just one copy of that subgoal.  (Recall that relation $R$ is not required to be a set on all instances of the database schema $\cal D$ of Example~\ref{motivating-example}.)  

Queries $Q_7$ and $Q_8$, as well as the database schema $\cal D$ of Example~\ref{motivating-example} together with its set-enforcing constraints, satisfy all the conditions of Lemma~\ref{bag-stubborn-lemma}. Observe that query $Q'_7$ (defined in the statement of Theorem~\ref{cv-updated-thm}) is isomorphic to $Q_7$, because relation $R$ is not required to be a set. Similarly, query $Q'_8$ is isomorphic to $Q_8$. Further, the canonical representation of each of $Q_7$, $Q_8$, $Q'_7$, and $Q'_8$ is isomorphic to query $Q_8$. 

Consider query $Q'_8$ and its canonical database $D'$, with $P = \{\hspace{-0.1cm}\{ (1,2) \}\hspace{-0.1cm}\}$ and $R = \{\hspace{-0.1cm}\{ (1) \}\hspace{-0.1cm}\}$. From $D'$, we construct a bag-valued database $D$, with relations $P = \{\hspace{-0.1cm}\{ (1,2) \}\hspace{-0.1cm}\}$ (same as $P$ in $D'$) and with $m > 0$ copies of tuple $(1)$ in relation $R$. That is, $R = \{\hspace{-0.1cm}\{ (1),\ldots, (1) \}\hspace{-0.1cm}\}$, with cardinality $m$ of bag $R$ in $D$. 
Let relations $S$, $T$, $U$  be empty sets in $D$. 
Then $D$ satisfies all the set-enforcing dependencies of Example~\ref{motivating-example}. 

Now using the notation of the proof of Lemma~\ref{bag-stubborn-lemma}, we have $n_1 = 2$. Here, $n_1$ is the number of subgoals of $Q'_7$ -- and thus also of $Q_7$ -- whose predicate corresponds to $R$. At the same time, $n_2 = 1 < n_1$,  where $n_2$ is the number of  subgoals of $Q'_8$ -- and thus also of $Q_8$ -- whose predicate corresponds to $R$. Further, the total number $n_3$ of subgoals of $Q_8$ is $n_3 = 2$, and the number $n_4$ of non-$R$ subgoals of $Q_7$ is $n_4 = 1$. 

It is easy to verify that the bag $Q_7(D,B)$ has $m^{(n_1)} = m^2$ copies of tuple $(1)$. At the same time, by the argument justifying Equation~\ref{qtwo-tuples-eq} in the proof of Lemma~\ref{bag-stubborn-lemma}, the total number of tuples (where each duplicate is counted separately) in the bag $Q_8(D,B)$ cannot exceed 
\begin{displaymath}
n_1^{(2n_2)} \times n_4^{(n_3-n_2)} \times m^{(n_2)} = 2^2 \times 1^{(2-1)} \times m^1 = 4m
\end{displaymath}
tuples. It is easy to see that for any value $m^*$ of $m$ such that $m^* > 4$, the number of copies of tuple $(1)$ in the bag $Q_7(D,B)$ is always going to be greater than the cardinality of the bag $Q_8(D,B)$.  

In fact, the upper bound of  Equation~\ref{qtwo-tuples-eq} is not tight for this example, as can be observed from the facts that 
\begin{itemize}
	\item the total number of copies of tuple $(1)$ in bag $Q_8(D,B)$ is $m$, and
	\item the core-set of the bag $Q_8(D,B)$ has no tuples other than $(1)$; therefore, the cardinality of the bag $Q_8(D,B)$ is $m$ as well. 
\end{itemize}
\nop{

\mbox{}

On this database $D$, the answer to $Q_3$ is $Q_3(D,B) = \{\hspace{-0.1cm}\{ (1) \}\hspace{-0.1cm}\}$, whereas $Q_7(D,B) = \{\hspace{-0.1cm}\{ (1), (1) \}\hspace{-0.1cm}\}$, by rules of bag semantics. From the fact that $Q_3(D,B)$ and $Q_7(D,B)$ are not the same bags, we conclude that 
bag equivalence $Q_3 \equiv_B Q_7$ on databases that satisfy the set-enforcing dependencies of Example~\ref{motivating-example} does not hold. 
} 
\vspace{-0.3cm}
\end{example}

%



\section{Proofs of the Theorems on\\ Sound Chase Steps}
\label{proof-sound-chase-steps-appendix}


We provide here representative parts of proofs for Theorems~\ref{bag-chase-sound-theorem} and~\ref{bag-set-chase-sound-theorem}. The idea of the complete proofs is to show, for an arbitrary embedded dependency, one of the following two things: 

(1) Either using the dependency results in sound chase steps, under the appropriate semantics, for all CQ queries, in case the format of the dependency  is described in the applicable theorem (i.e., either Theorem~\ref{bag-chase-sound-theorem} or Theorem~\ref{bag-set-chase-sound-theorem}).  Please see Proposition~\ref{sound-key-chase-prop}  in Section~\ref{prop-app-he} for an example of such a claim. 

(2) Or using the dependency results in unsound chase, in case the format of the dependency  is not  described in the theorem for the respective query-evaluation semantics (i.e., either Theorem~\ref{bag-chase-sound-theorem} or Theorem~\ref{bag-set-chase-sound-theorem}). Please see Propositions~\ref{one-unsound-key-chase-prop} and~\ref{two-unsound-key-chase-prop}  in Section~\ref{prop-app-he} for examples of such claims. 

All the remaining proofs for  Theorems~\ref{bag-chase-sound-theorem} and~\ref{bag-set-chase-sound-theorem} are analogous to the proofs of Propositions~\ref{sound-key-chase-prop} through~\ref{two-unsound-key-chase-prop}.

\subsection{Bag Projection}
\label{bag-proj-app}

This subsection of the appendix defines bag projection. We use the definition in the proof of Proposition~\ref{sound-key-chase-prop} in Section~\ref{prop-app-he}. 

Given positive integers $m$, $k$ and $i(1),$ $\ldots,i(k)$, such that for  each $j \in \{ 1,\ldots, k \}$ it holds that $1 \leq i(j) \leq m$. Then for  an $m$-tuple $t = (a_1,\ldots,a_m)$, we say that a $k$-tuple $t' = (a_{i(1)},\ldots,a_{i(k)})$  
is a {\it projection of $t$ on attributes in positions} $i(1),\ldots,i(k)$, denoted 
\vspace{-0.2cm}
\begin{tabbing}
$t' = $ $t[i(1),\ldots,i(k)]$.
\end{tabbing}
\vspace{-0.2cm}
Further, for the $m$, $k$ and $i(1),$ $\ldots,i(k)$ as above and for an $m$-ary relation $P$, 
a bag of tuples $B$ is a {\em bag projection of $P$ on attributes in  positions} $i(1),\ldots,i(k)$, denoted $B = \pi^{bag}_{i(1),\ldots,i(k)}(P)$, if each tuple $t \in P$ contributes to $B$ a separate tuple  $t' = t[i(1),\ldots,i(k)]$, and if $B$ has no other tuples. $B$ can be interpreted as  the answer $Q(D,B)$ on database $\{ P \}$ to query 
\vspace{-0.2cm}
\begin{tabbing}
$Q(X_{i(1)},\ldots,X_{i(k)}) \ :- \ p(X_1,\ldots,X_m)$, 
\end{tabbing}
\vspace{-0.2cm}
where the predicate $p$ corresponds to relation $P$.

\subsection{The Proofs}
\label{prop-app-he}

\begin{proposition}
\label{sound-key-chase-prop}
Given a CQ query $Q$ and a set of embedded dependencies $\Sigma$. Under bag semantics for query evaluation, a chase step $Q \Rightarrow^{\sigma}_B Q'$ using tgd $\sigma \in \Sigma$ is sound if  $Q \Rightarrow^{\sigma}_B Q'$ 
is (tgd) key-based, and for each subgoal $s(p_{ij})$ that the chase step adds to $Q$, relation $P_{ij}$ is set valued on all databases satisfying $\Sigma$. 
%
\end{proposition}

\begin{proof}
Let $\sigma$ 
be of the form 
\vspace{-0.1cm}
\begin{tabbing}
$\sigma: \phi(\bar{X},\bar{Y}) \rightarrow \exists \bar{Z}$ $p_1(\bar{Y}_1,\bar{Z}_1) \wedge \ldots \wedge p_n(\bar{Y}_n,\bar{Z}_n)$, 
\end{tabbing}
\vspace{-0.1cm}
with $n > 0$. Here, the set of variables in each $\bar{Y}_i$, $i \in \{ 1,\ldots,n \},$ is the maximal subset, in the set of variables in $\bar{Y}_i \bigcup \bar{Z}_i$, of the set of variables in $\bar{Y}$. (We abuse the notation by treating $\bar{Y}_i \bigcup \bar{Z}_i$ as a set of variables and constants.) We show that the chase step $Q \Rightarrow^{\sigma}_B Q'$ using $\sigma$ is sound whenever for all $i \in \{ 1,\ldots,n \}$ such that $p_i(\bar{Y}_i,\bar{Z}_i)$ corresponds to a subgoal in $Q'$ that is not a subgoal of $Q$, it holds that (1)  $\bar{Y}_i$ is a superset of the key of relation symbol $P_i$ in $\cal D$, and (2) $P_i$ is set valued in all databases with schema $\cal D$. 

By our assumption that $\sigma$ is applicable to $Q$, (1) there exists a mapping $\mu$ from a (not necessarily proper) superset $\xi$ of $\phi$ to a subset of subgoals of $Q$. By the same assumption, (2) there does not exist a mapping $\mu'$ such that $\mu'$ is an extension of $\mu$ and such that $\mu'(\psi)$ is also a subset of subgoals of $Q$. Here, $\psi$ is the right-hand side of the tgd $\sigma$. 

Consider a mapping $\nu$ from $\phi$ to the body of $Q$, such that $\nu$ agrees with $\mu$ on all the variables in $\xi$ (note that  all of $\bar{Y}$ are in $\xi$), and such that $\nu$ maps the subset of  variables  $\bar{Z}$ in $\psi - \xi$ (here, ``$\psi - \xi$'' is read as set difference between sets of conjuncts $\psi$ and $\xi$) into distinct fresh variables. 
By definition of chase step for tgds, $\nu(\psi)$ adds at least one subgoal  to $Q$, which results in query $Q'$. Let one such new subgoal $S$ be the result of applying $\nu$ to atom $p_i(\bar{Y}_i,\bar{Z}_i)$ in the $\psi$ part of $\sigma$, for some $i \in \{ 1,\ldots,n \}.$ 

Consider  an arbitrary database $D$ with schema $\cal D$, such that $D$ satisfies the dependencies $\Sigma$. To finalize our proof, it remains to show that on $D$, the following two relations are the same as bags: $Q(D,B)$ and $Q''(D,B)$, where $Q''$ results from adding the subgoal $S$ to the body of $Q$. Here, each of $Q(D,B)$ and $Q''(D,B)$ is to be computed under bag semantics for query evaluation. 

Let $bQ(D,B)$ be the relation, on $D$, for the body of $Q(D,B)$, and let $bQ''(D,B)$ be the relation, on $D$, for the body of $Q''(D,B)$. Note that if $bQ(D,B)$ and $bQ''(D,B)$ are the same bags modulo the columns of $bQ(D,B)$, then $Q(D,B)$ and $Q''(D,B)$ are the same bags as well. (Recall that the heads of $Q$ and $Q''$ are the same by definition of $Q''$.) When we say ``$bQ(D,B)$ and $bQ''(D,B)$ are the same bags modulo the columns of $bQ(D,B)$'', the meaning is as follows: If we do bag projection on $bQ''(D,B)$ on just the columns of $bQ(D,B)$, then we will obtain precisely $bQ(D,B)$. (Please see Appendix~\ref{bag-proj-app} for the definition of bag projection.)

We now show that $bQ(D,B)$ and $bQ''(D,B)$ are the same bags modulo the columns of $bQ(D,B)$, which finalizes our proof. The case where $bQ(D,B)$ is empty is trivial, thus we assume for the remainder of the proof that $bQ(D,B)$ is not an empty bag. Consider an assignment mapping $\lambda$ that was used to obtain a tuple $t$ in bag $bQ(D,B)$. By definition of (tgd) key-based chase step for $\sigma$, there is exactly one way (up to duplicates of stored tuples) to extend $\lambda$, to obtain a (distinct) tuple $t' \in P_i$, such that $t'$ ``matches'' $t$ according to the join conditions between the body of $Q$ and the new subgoal $S$ in $Q''$.\footnote{That is, the extension of $\lambda$ is a satisfying assignment for the body of $Q''$ w.r.t database $D$. In this and other proofs, we can  use ``procedural'' evaluation of queries under each of bag and bag-set semantics. The correctness of this usage stems from the fact that our definitions for query evaluation under bag and bag-set semantics, see Section~\ref{bag-bag-set-defs}, are consistent with the operational semantics of evaluating CQ queries in the SQL standard, as shown in~\cite{VardiBagsPods93}.} Further, from the fact that the relation $P_i$ is a set on $D$, we obtain that $t'$ is a unique tuple (i.e., it has no duplicates in $P_i$) that ``matches'' $t$ in the above sense. As a result, each single tuple in $bQ(D,B)$ corresponds, for the purposes of computing $bQ''(D,B)$ from $bQ(D,B)$, to exactly one tuple in $P_i$. 

Observe that the above procedure for computing\linebreak $bQ''(D,B)$ from $bQ(D,B)$ corresponds to a valid plan for computing $bQ''(D,B)$ from only the stored relations in $D$. (This plan is a left-linear plan, such that $P_i$ is the right input of the top join-operator node in the tree. For the basics on query-evaluation plans, please see~\cite{GarciaMolinaUW02}.) We conclude that $bQ(D,B)$ and $bQ''(D,B)$ are the same bags modulo the columns of $bQ(D,B)$.
\end{proof}

\begin{proposition}
\label{one-unsound-key-chase-prop}
Given a CQ query $Q$ and a set of embedded dependencies $\Sigma$. Consider a key-based chase step $Q \Rightarrow^{\sigma}_B Q'$ using tgd $\sigma \in \Sigma$, 
\begin{tabbing}
$\sigma: \phi(\bar{X},\bar{Y}) \rightarrow \exists \bar{Z} \ \psi(\bar{Y},\bar{Z})$. 
\end{tabbing}
Suppose that at least one relation $P_i$ used in $\psi$ is not set valued. Further, suppose that in the chase step $Q \Rightarrow^{\sigma}_B Q'$ using $\sigma,$ $Q'$ is obtained by adding to the body of $Q$ a new $P_i$-subgoal $s(P_i)$ (possibly alongside other subgoals).\footnote{I.e., in the chase step $Q \Rightarrow^{\sigma}_B Q'$, applying $\sigma$ to $Q$ may generate other new subgoals besides the $P_i$-subgoal.} Then under bag semantics for query evaluation, the chase step $Q \Rightarrow^{\sigma}_B Q'$ using $\sigma$ is not sound.
\end{proposition}

\begin{proof}
Let $\sigma$ 
be of the form 
\begin{tabbing}
$\sigma: \phi(\bar{X},\bar{Y}) \rightarrow \exists \bar{Z} \ p_1(\bar{Y}_1,\bar{Z}_1) \wedge \ldots \wedge p_n(\bar{Y}_n,\bar{Z}_n)$, 
\end{tabbing}
with $n > 0$. Here, for each $j \in \{ 1,\ldots,n \},$ $\bar{Y}_j$  is the maximal subset of $\bar{Y}$ in the set $\bar{Y}_j \bigcup \bar{Z}_j$, please see proof of Proposition~\ref{sound-key-chase-prop} for the notation. In addition, the relation $P_i$ for $p_i(\bar{Y}_i,\bar{Z}_i)$ is not a set-valued relation for at least one $i \in \{ 1,\ldots,n \}$. Given this assumption on $P_i$, the proof of the claim of Proposition~\ref{one-unsound-key-chase-prop} is by providing a bag-valued database $D$, such that $D$ satisfies $\Sigma$ and such that $Q(D,B)$ and $Q'(D,B)$ are not the same bags. 

We build the database $D$ as follows. Let $D'$ be the canonical database for query $Q'$. We obtain $D$ by adding to $D'$ a single duplicate of the tuple for the subgoal $s(P_i)$ of $Q'$. We now follow the reasoning in the proof of Proposition~\ref{sound-key-chase-prop}, to observe that the bag $Q'(D,B)$ has {\it at least one more tuple than} the bag $Q(D,B)$, due to the fact that the two identical tuples of relation $P_i$ add to $Q'(D,B)$ an extra copy of at least one tuple in $Q(D,B)$. This observation concludes the proof. 
\end{proof}

We now provide an illustration that shows the main points of the proof of Proposition~\ref{one-unsound-key-chase-prop}. 
\begin{example}
Consider a set $\Sigma = \{ \sigma_1, \sigma_2 \}$ of embedded dependencies, where
\begin{tabbing}
$\sigma_1: p(X,Y) \wedge p(X,Z) \rightarrow Y = Z.$ \\
$\sigma_2: r(X,Y) \rightarrow p(X,Y).$
\end{tabbing}
Observe that chase steps using $\sigma_2$  are (tgd) key-based in presence of the egd $\sigma_1$, and that $\Sigma$ does not include dependencies that would restrict the relation $P$, in the right-hand side of $\sigma_2$, to be set valued.

Consider a CQ query $Q$ defined as
\begin{tabbing}
$Q(A) \ :- \ r(A,B).$
\end{tabbing}
Applying $\sigma_2$ to the query $Q$, in chase step $Q \Rightarrow^{\sigma_2}_B Q'$, results in query $Q'$ defined as
\begin{tabbing}
$Q'(A) \ :- \ r(A,B), p(A,B).$
\end{tabbing}

We now illustrate the construction of the database $D$ in the proof of Proposition~\ref{one-unsound-key-chase-prop}.  First, the canonical database $D'$ of $Q'$ has relations $R = \{\hspace{-0.1cm}\{ (a,b) \}\hspace{-0.1cm}\}$ and $P = \{\hspace{-0.1cm}\{ (a,b) \}\hspace{-0.1cm}\}$. $D$ is constructed from $D'$ by adding to relation $P$ a duplicate of the tuple $(a,b)$, that is $D$ has relations $R = \{\hspace{-0.1cm}\{ (a,b) \}\hspace{-0.1cm}\}$ and $P = \{\hspace{-0.1cm}\{ (a,b), (a,b) \}\hspace{-0.1cm}\}$. 
Note that database $D$ is bag valued and satisfies all the dependencies in $\Sigma$.

Now by the bag semantics for query evaluation, $Q(D,B)$ $=  \{\hspace{-0.1cm}\{ (a) \}\hspace{-0.1cm}\}$, while $Q'(D,B) =  \{\hspace{-0.1cm}\{ (a), (a) \}\hspace{-0.1cm}\}$. Thus, database $D$ is a counterexample to $Q \equiv_{\Sigma,B} Q'$, which proves that the chase step $Q \Rightarrow^{\sigma_2}_B Q'$ using $\sigma_2$ is not sound. 
\end{example}

\begin{proposition}
\label{two-unsound-key-chase-prop}
Given a CQ query $Q$ and a set of embedded dependencies $\Sigma$. Let $\sigma \in \Sigma$ be a tgd, 
\begin{tabbing}
$\sigma: \phi(\bar{X},\bar{Y}) \rightarrow \exists \bar{Z} \ p_1(\bar{Y}_1,\bar{Z}_1) \wedge \ldots \wedge p_n(\bar{Y}_n,\bar{Z}_n)$,
\end{tabbing}
with $n > 0$.\footnote{For the notation, please see proof of Proposition~\ref{sound-key-chase-prop}.}  Consider a chase step $Q \Rightarrow^{\sigma}_{BS} Q'$ using $\sigma,$ such that $Q'$ is obtained by adding to the body of $Q$ a new $P_i$-subgoal $s(P_i)$ (possibly alongside other subgoals), where $s(P_i)$ corresponds to conjunct  $p_i(\bar{Y}_i,\bar{Z}_i)$ in the consequent $\psi$ of $\sigma$. Suppose that $\bar{Y}_i$ is not a superkey of $P_i$. 
Then under bag-set semantics, the chase step $Q \Rightarrow^{\sigma}_{BS} Q'$ using $\sigma$ is not sound.
\end{proposition}

(Proposition~\ref{two-unsound-key-chase-prop} is formulated for the case of bag-set semantics, which allows us to show the flavor of the proofs that are required to establish Theorem~\ref{bag-set-chase-sound-theorem}.)

\begin{proof}{(Proposition~\ref{two-unsound-key-chase-prop})}
Given the assumption that for the conjunct $p_i(\bar{Y}_i,\bar{Z}_i)$ used in $\psi$, it holds that $\bar{Y}_i$ is not a superkey of $P_i$, the proof of the claim of Proposition~\ref{two-unsound-key-chase-prop} is by providing a {\em set}-valued database $D$, such that $D$ satisfies $\Sigma$ and such that $Q(D,BS)$ and $Q'(D,BS)$ are not the same bags. 

Fix $i$ such that $\bar{Y}_i$ in $p_i(\bar{Y}_i,\bar{Z}_i)$ is not a superkey of $P_i$ and such that $p_i(\bar{Y}_i,\bar{Z}_i)$ corresponds to a subgoal in $Q'$ that (subgoal) is not in $Q$, in chase step $Q \Rightarrow^{\sigma}_{BS} Q'$.  We begin the construction of the database $D$ by building  the canonical database $D'$ for query $Q'$. We obtain $D$ by adding to $D'$ a single extra {\em (nonduplicate)} tuple for the subgoal $s(P_i)$ of $Q'$, as follows. 

Without loss of generality, let $p_i(\bar{Y}_i,\bar{Z}_i)$ be of the form $p_i(\bar{Y}_i,\bar{Z}'_i,\bar{Z}''_i)$, where $\bar{Z}'_i$ is not empty, $\bar{Y}_i \bigcup \bar{Z}'_i$ is a superkey of  $P_i$, and no proper subset of $\bar{Y}_i \bigcup \bar{Z}'_i$ is a superkey of $P_i$. Now suppose $\nu$ was the mapping  used to generate $s(P_i)$ from $p_i(\bar{Y}_i,\bar{Z}'_i,\bar{Z}''_i)$ in the chase step $Q \Rightarrow^{\sigma}_{BS} Q'$. (See proof of Proposition~\ref{sound-key-chase-prop} for the details on $\nu$.) Then $s(P_i)$ is of the form $p_i(\bar{A},\bar{C},\bar{E})$, where $\bar{A}$ ($\bar{C}$, $\bar{E}$, respectively) is the image of $\bar{Y}_i$ (of $\bar{Z}'_i$, of $\bar{Z}''_i$, respectively)  under $\nu$.  By construction of the canonical database $D'$ of $Q'$, the tuple for $s(P_i)$ in relation $P_i$ in $D'$ is $(\bar{a},\bar{c},\bar{e})$. We construct the database $D$ from $D'$ by adding to $P_i$ of $D'$ a tuple $(\bar{a},\bar{c}',\bar{e})$, such that at least one constant in $\bar{c}'$ is not equal to the same-position constant in $\bar{c}$. By construction, database $D$ is set valued and satisfies the dependencies $\Sigma$. (Recall that no proper subset of $\bar{Y}_i \bigcup \bar{Z}'_i$ in $p_i(\bar{Y}_i,\bar{Z}'_i,\bar{Z}''_i)$    is a superkey of $P_i$.) 

We now follow the reasoning in the proof of Proposition~\ref{sound-key-chase-prop}, to observe that the bag $Q'(D,BS)$ has {\it at least one more tuple than} the bag $Q(D,BS)$. The reason is, tuples $(\bar{a},\bar{c},\bar{e})$ and $(\bar{a},\bar{c}',\bar{e})$ in relation $P_i$ add to $Q'(D,BS)$ an extra copy of at least one tuple in $Q(D,BS)$. This observation concludes the proof. 
\end{proof}

We now provide an illustration that shows the main points of the proof of Proposition~\ref{two-unsound-key-chase-prop}. 
\begin{example}
Consider a set $\Sigma = \{ \sigma \}$ of embedded dependencies, where
\begin{tabbing}
$\sigma: r(X,Y) \rightarrow p(X,Z).$
\end{tabbing}

Given a CQ query $Q$,
\begin{tabbing}
$Q(A) \ :- \ r(A,B).$
\end{tabbing}
applying $\sigma$ to the query $Q$ results in query $Q'$, 
\begin{tabbing}
$Q'(A) \ :- \ r(A,B), p(A,C).$
\end{tabbing}
Observe that chase step $Q \Rightarrow^{\sigma}_{BS} Q'$ using   $\sigma$ is not key-based, as the set of all attributes of $P$ is the only key of $P$.

We now illustrate the construction of the database $D$ in the proof of Proposition~\ref{two-unsound-key-chase-prop}.  First, the  canonical database $D'$ of $Q'$ has relations $R = \{ (a,b) \}$ and $P = \{ (a,c) \}$. $D$ is constructed from $D'$ by adding to relation $P$ a new tuple $(a,d)$, that is $D$ has relations $R = \{ (a,b) \}$ and $P = \{ (a,c), (a,d) \}$. 
Note that database $D$ is set valued and satisfies the dependency $\sigma$.

Now by the bag-set semantics for query evaluation, $Q(D,BS) =  \{\hspace{-0.1cm}\{ (a) \}\hspace{-0.1cm}\}$, whereas $Q'(D,BS) =  \{\hspace{-0.1cm}\{ (a), (a) \}\hspace{-0.1cm}\}$. Thus, database $D$ is a counterexample to $Q \equiv_{\{ \sigma \},BS} Q'$, which proves that the chase step $Q \Rightarrow^{\sigma}_{BS} Q'$ using $\sigma$ is not sound. 
\end{example}

\section{Counterexample Database for Example 5.1}
\label{real-motiv-appendix}


This section of the appendix provides the counterexample database for Example~\ref{real-he-motivating-example}. Database $D$ is a counterexample to  soundness of chase step $Q_4 \Rightarrow^{\sigma_1}_B Q_4^{(1)}$. In $D$, let the relations be as follows: $P = \{\hspace{-0.1cm}\{ (1,2) \}\hspace{-0.1cm}\}$, $R = \emptyset$, $S = \{\hspace{-0.1cm}\{ (1,3) \}\hspace{-0.1cm}\}$, $T = \{\hspace{-0.1cm}\{ (1,4,5),$ $(1,6,7) \}\hspace{-0.1cm}\}$, and $U = \emptyset$. Note that $D \models \Sigma$, for the set of dependencies $\Sigma$ in Example~\ref{real-he-motivating-example}. 

On this database $D$, the answer to $Q_4$ is $Q_4(D,B) = \{\hspace{-0.1cm}\{ (1) \}\hspace{-0.1cm}\}$, whereas $Q_4^{(1)}(D,B) = \{\hspace{-0.1cm}\{ (1), (1) \}\hspace{-0.1cm}\}$, by rules of bag semantics. From the fact that $Q_4(D,B)$ and $Q_4^{(1)}(D,B)$ are not the same {\em bags,} we conclude that the chase step $Q_4 \Rightarrow^{\sigma_1}_B Q_4^{(1)}$ is not sound under bag semantics.

\section{Uniqueness Theorems for Chase Results}
\label{bag-set-uniqueness-appendix}

We begin this section of the appendix by formulating the version of Theorem~\ref{uniqueness-theorem} for the case of bag-set semantics. 
\begin{theorem}
\label{bag-set-uniqueness-theorem}
Given a CQ query $Q$ and a set $\Sigma$ of embedded dependencies on database schema $\cal D$, such that there exists a  chase result $(Q)_{\Sigma,S}$ for $Q$ and $\Sigma$ under {\em set} semantics.  Then there exists a result $(Q)_{\Sigma,BS}$ of {\em sound} chase  for $Q$ and $\Sigma$ under {\em bag-set} semantics, unique up to isomorphism of its canonical representation.\footnote{See Theorem~\ref{cv-theorem} in Section~\ref{bag-equiv-defs}.} 
That is, for two sound-chase results $(Q)^{(1)}_{\Sigma,BS}$ and $(Q)^{(2)}_{\Sigma,BS}$ for $Q$ and $\Sigma$, $(Q)^{(1)}_{\Sigma,BS} \equiv_{BS} (Q)^{(2)}_{\Sigma,BS}$ in the absence of dependencies. 
\end{theorem}

We now provide a proof for Theorem~\ref{uniqueness-theorem}. An adaptation of the proof to the statement of Theorem~\ref{bag-set-uniqueness-theorem} is straightforward. 
\begin{proof}{(Theorem~\ref{uniqueness-theorem})}
We first establish that, by the definition of soundness of the chase result $(Q)^{(1)}_{\Sigma,B}$, there exists a chase sequence $C_1$ using $\Sigma$, such that $C_1$ starts with $Q$ and ends with $(Q)^{(1)}_{\Sigma,B}$, and such that all chase steps in $C_1$ are sound under bag semantics. Similarly, we establish that there exists a  chase sequence $C_2$ using $\Sigma$, such that $C_2$ starts with $Q$ and ends with $(Q)^{(2)}_{\Sigma,B}$, and such that all chase steps in $C_2$ are sound under bag semantics. 

The proof of Theorem~\ref{uniqueness-theorem} is by contrapositive. Assume, toward contradiction, that $(Q)^{(1)}_{\Sigma,B}$ and $(Q)^{(2)}_{\Sigma,B}$ are not isomorphic after removal of duplicate subgoals that correspond to set-valued relations in the database schema $\cal D$. Let us denote by $(\bar{Q})^{(1)}_{\Sigma,B}$ the result of removing such ``set-valued'' duplicate subgoals from $(Q)^{(1)}_{\Sigma,B}$, and let us use the analogous notation $(\bar{Q})^{(2)}_{\Sigma,B}$ for $(Q)^{(2)}_{\Sigma,B}$. 

Suppose, w.l.o.g., that  $(\bar{Q})^{(1)}_{\Sigma,B}$ has a nonempty set of subgoals $p_1(\bar{X}_1),\ldots,p_m(\bar{X}_m)$ such that this set of subgoals does not have a counterpart in the image of any 
injective homomorphism from $(\bar{Q})^{(1)}_{\Sigma,B}$ to $(\bar{Q})^{(2)}_{\Sigma,B}$. It is clear that $p_1(\bar{X}_1),\ldots,p_m(\bar{X}_m)$ cannot be a subset of all the subgoals in the body of the original query $Q$. (By definition of sound chase steps, no chase steps using embedded dependencies ever remove original query subgoals.) Then, from the sound chase sequence $C_1$, we can form a sequence $C'_1$ of sound chase steps that (i) uses a {\em subsequence} of the sequence of dependencies applied in $C_1$, and (ii) starts with $Q$ and ends with adding all the subgoals in $p_1(\bar{X}_1),\ldots,p_m(\bar{X}_m)$ to $Q$. By definition of sound chase, there must exist a nonempty suffix subsequence $C''_1$ of $C'_1$ such that all chase steps in $C''_1$ apply to the chase result $(Q)^{(2)}_{\Sigma,B}$ of chase sequence $C_2$, and such that applying the respective (to $C''_1$) dependencies in $\Sigma$ to $(Q)^{(2)}_{\Sigma,B}$  would result in adding to $(Q)^{(2)}_{\Sigma,B}$ a set of subgoals that would be an image of $p_1(\bar{X}_1),\ldots,p_m(\bar{X}_m)$ in some injective homomorphism from $(\bar{Q})^{(1)}_{\Sigma,B}$ to $(\bar{Q})^{(2)}_{\Sigma,B}$. We thus arrive at a contradiction with the condition of Theorem~\ref{uniqueness-theorem}, which states that $(Q)^{(2)}_{\Sigma,B}$ is a (terminal) result of sound chase for $Q$ using $\Sigma$ under bag semantics. (That is, the contradiction is with the assumption that no sound chase steps of the form $(Q)^{(2)}_{\Sigma,B} \Rightarrow^{\sigma}_B Q'$ are possible, where $\sigma \in \Sigma$.)

The case where some of $p_1(\bar{X}_1),\ldots,p_m(\bar{X}_m)$ were eliminated in $(Q)^{(2)}_{\Sigma,B}$ by use of one or more egds in $\Sigma$ is  analogous to the above tgd case, except that the contradiction in the case of egds is with our assumption that $(Q)^{(1)}_{\Sigma,B}$ is a (terminal) result of sound chase for $Q$ using $\Sigma$ under bag semantics. That is, those same egds can be applied to $(Q)^{(1)}_{\Sigma,B}$, hence $(Q)^{(1)}_{\Sigma,B}$ is not a result of sound chase under bag semantics. 
\end{proof}

\section{Complexity of Sound Chase}
\label{app-chase-compl}

In this section of the appendix, we establish for Theorem~\ref{complexity-thm} the lower bound on the complexity of sound chase under each of bag and bag-set semantics, using sets of weakly acyclic dependencies. 

\subsection{Weakly Acyclic Dependencies}

We provide here the definition and discussion of~\cite{DeutschPT06} for weakly acyclic dependencies. 

The chase-termination property under set semantics is in general undecidable for CQ queries 
and dependencies given by tgds and egds. However, the notion of {\it weak acyclicity} of a set of dependencies is sufficient to guarantee that any chase sequence terminates. This is the least restrictive sufficient termination condition that has been generally studied  in the literature (but see~\cite{DeutschPods08} for a generalization). The weak acyclicity condition appears to hold in all practical scenarios. 

\vspace{-0.2cm}

\begin{definition}{Weakly acyclic set of dependencies} 
Let $\Sigma$ be a set of tgds over a fixed schema. Construct a directed graph, called 
the {\em dependency graph,} as follows: (1) there is a node for every 
pair $(R, A)$, with $R$ a relation symbol of the schema and $A$ an 
attribute of $R$; call such pair $(R, A)$ a {\em position;} (2) add edges 
as follows: for every tgd $\phi(\bar{X}) \rightarrow \exists \bar{Y} \ \psi(\bar{X},\bar{Y})$ in $\Sigma$ and for every $X$ in $\bar{X}$ that occurs in $\psi$: 

For every occurrence of $X$ in $\phi$ in position $(R, A_i )$: 
\begin{itemize}
	\item[(a)] for every occurrence of $X$ in $\psi$ in position $(S, B_j )$, add 
an edge $(R, A_i ) \rightarrow (S, B_j)$; 
	\item[(b)] in addition, for every existentially quantified variable 
$Y$ and for every occurrence of $Y$ in $\psi$ in position $(T , C_k )$, 
add a special edge $(R, A_i ) \rightarrow^{*}  (T , C_k )$. 
\end{itemize}

Note that there may be two edges in the same direction between two nodes, if exactly one of the two edges is special. Then $\Sigma$ is {\em weakly acyclic} if the dependency graph has no cycle going 
through a special edge. We say that a set of tgds and egds is 
{\em weakly acyclic} if the set of all its tgds is weakly acyclic. 
\end{definition}

\vspace{-0.4cm}

\begin{theorem}{\cite{DeutschT03,FaginKMP05}}
If $\Sigma$ is a weakly acyclic set of tgds and egds, then the chase with $\Sigma$ of any CQ query $Q$ under set semantics terminates in finite time. 
\end{theorem}

{\bf The complexity of the chase.} For a fixed database schema and set $\Sigma$ 
of dependencies, if $\Sigma$ is weakly acyclic then under set semantics any chase sequence 
terminates in polynomial time in the size of the query being 
chased (as shown in \cite{DeutschT03,FaginKMP05}). The fixed-size assumption about 
schemas and dependencies is often justified in practice, where 
one is usually interested in repeatedly reformulating incoming queries for the same setting with schemas and dependencies. Nonetheless, the degree of the polynomial depends on the size 
of the dependencies and care is needed to implement the chase 
efficiently. Successive implementations have shown that in practical situations the chase is eminently usable~\cite{DeutschPT06}. 

\mbox{}

{\bf The complexity of reformulation under set semantics (in C\&B).} Assume that under set semantics the chase of any query with $\Sigma$ terminates in polynomial time (for fixed database schema). Then checking whether a CQ query $Q$ admits a reformulation is NP-complete in the size of $Q$. Checking whether a given query $Q'$  
is a $\Sigma$-minimal reformulation of $Q$ is NP-complete in the sizes of $Q$ and $Q'$. For arbitrary sets of dependencies (for which the chase may not even terminate), the above problems are undecidable.

\subsection{The Lower Complexity Bound}

We now establish for Theorem~\ref{complexity-thm} the lower bound on the complexity of sound chase using weakly acyclic dependencies under each of bag and bag-set semantics, as follows.

\begin{example}
\label{expon-example}
On a database schema ${\cal D} =$\linebreak $\{ P_1, \ P_2, \ \ldots, \ P_m \}$ where each relation symbol has arity 2, consider a query $Q$ with a single subgoal $p_1$: 
\vspace{-0.1cm}
\begin{tabbing}
$Q(X,Y) \ :- \ p_1(X,Y)$. 
\end{tabbing}
\vspace{-0.1cm}

Suppose the database schema $\cal D$ satisfies a set $\Sigma$ of tgds of the following form:

\vspace{-0.1cm}
\begin{tabbing}
$\sigma^{(1)}_{i,j}: \ p_i(X,Y) \rightarrow \ \exists Z \ p_{j}(Z,X)$ \\
$\sigma^{(2)}_{i,j}: \ p_i(X,Y) \rightarrow \ \exists W \ p_{j}(Y,W)$ 
\end{tabbing}
\vspace{-0.1cm}

$\Sigma$ has one tgd $\sigma^{(1)}_{i,j}$ and one tgd $\sigma^{(2)}_{i,j}$ for each pair $(i,j)$, where $i \ \in \ \{ 1,\ldots,m-1 \}$ and $j \ \in \ \{ i+1,\ldots,m \}$. Thus, the number of dependencies in $\Sigma$ is quadratic in $m$. 

We show one partial chase result (under set semantics) of the query $Q$ under dependencies $\Sigma$, for $m \geq 2$:
\vspace{-0.1cm}
\begin{tabbing}
This is my \= kill mark \kill
$Q'(X,Y) \ :- \ p_1(X,Y), \ p_2(Z_1,X), \ p_2(Y,Z_2)$.
\end{tabbing}
\vspace{-0.1cm}

$Q'$ is the result of applying to $Q$ tgds $\sigma^{(1)}_{1,2}$ and $\sigma^{(2)}_{1,2}$. Observe that $Q'$ has a self-join of the relation $P_2$.
\end{example}

For the terminal result $(Q)_{\Sigma,S}$ of chase of the query $Q$ using the tgds $\Sigma$ under set semantics in Example~\ref{expon-example}, we can show that the size of $(Q)_{\Sigma,S}$   is exponential in the size of $Q$ and $\Sigma$. Specifically, the size of $(Q)_{\Sigma,S}$ is exponential in the size $m$ of the database schema $\cal D$. Intuitively, just as $Q'$ has two subgoals for predicate $p_2$, the query $(Q)_{\Sigma,S}$ has two subgoals for $p_2$, four subgoals for $p_3$, and so on.  

\begin{example}
\label{cont-expon-example}
We continue Example~\ref{expon-example}. We build a set $\Sigma'$ of dependencies from the set $\Sigma$ of Example~\ref{expon-example} by adding $3m$ functional dependencies (fds): For each $i \in \{ 1,\ldots,m \}$, we add the following three fds for the relation $P_i$ in $\cal D$:
\vspace{-0.1cm}
\begin{tabbing}
$\sigma^{(1)}_{i}: \ p_i(X,Y) \ \wedge \ p_i(X,Z) \rightarrow \ Y = Z$ \\
$\sigma^{(2)}_{i}: \ p_i(Y,X) \ \wedge \ p_i(Z,X) \rightarrow \ Y = Z$ \\
$\sigma^{(3)}_{i}: \ p_i(X,Y,Z_1) \ \wedge \ p_i(X,Y,Z_2) \rightarrow \ Z_1 = Z_2$ 
\end{tabbing}
\vspace{-0.1cm}

That is, in all databases that satisfy the first two fds for $i$ in $\Sigma'$, the core-set of $P_i$ does not have repeated values of either attribute. The third fd for $P_i$ guarantees that relation $P_i$ is {\em set} valued in all instances of the database schema $\cal D$. Here, the {\em third} attribute of $P_i$ is its {\em tuple-id} attribute. Please see Appendix~\ref{appendix-a} for the details on using egds for enforcing set-valuedness of relations in all  instances of a given database schema. 

Note that the addition of these fds transforms the tgds $\Sigma$ of Example~\ref{expon-example} into key-based tgds $\Sigma'$ (see Definition~\ref{old-key-based-tgds-def}). Thus, 
 for the terminal result $(Q)_{\Sigma',B}$ of sound chase of the query $Q$ under the dependencies $\Sigma'$ under bag semantics, the size of $(Q)_{\Sigma',B}$  is exponential in the size of $Q$ and $\Sigma'$. The same relationship holds under bag-set semantics between the size of $(Q)_{\Sigma',BS}$ and the sizes of $Q$ and $\Sigma$.  
\end{example}

By the results of Section~\ref{new-sound-chase-sec}, chase of CQ query $Q$ under key-based tgds $\Sigma$ results in a query that is equivalent to $Q$ under $\Sigma$ under each of bag and bag-set semantics for query evaluation. Observing that the dependencies $\Sigma'$ of Example~\ref{cont-expon-example} are weakly acyclic (and, in fact, strictly acyclic), completes the construction of the infinite family of pairs $(Q,\Sigma')$, one pair for each natural-number value of $m$, such that the size of each of $(Q)_{\Sigma,B}$ and $(Q)_{\Sigma,BS}$ (both constructed using sound  chase) is polynomial in the size of $Q$ and exponential in size of $\Sigma$. 

\section{Satisfiable Dependencies Are\\ Query Based}
\label{semant-appendix}

In this section of the appendix we provide Theorem~\ref{bag-set-sigma-max-theorem}, which is the analog of Theorem~\ref{sigma-max-theorem} for the case of bag-set semantics. We  then supply a proof of Theorem~\ref{sigma-max-theorem}; the proof of Theorem~\ref{bag-set-sigma-max-theorem} is similar. 
Finally, we outline the counterpart of algorithm {\sc Max-Bag-$\Sigma$-Subset} (of Section~\ref{semantic-sec-now}) for the case of bag-set semantics. 

\vspace{-0.1cm}

\begin{theorem}{{\bf (Unique $\Sigma^{max}_{BS}(Q,\Sigma) \subseteq \Sigma$)}}
\label{bag-set-sigma-max-theorem}
Given a CQ query $Q$ and set $\Sigma$ of embedded dependencies, such that there exists a  {\em set-}chase result $(Q)_{\Sigma,S}$ for $Q$ and $\Sigma$. 
Let $Q_n$ be the result of  sound chase for $Q$ and $\Sigma$ under {\em bag-set} semantics, with canonical database $D^{(Q_n)}$.  Then there exists a unique subset $\Sigma^{max}_{BS}(Q,\Sigma)$ of $\Sigma$, such that:
\begin{itemize}
	\item $D^{(Q_n)} \models \Sigma^{max}_{BS}(Q,\Sigma)$, and 
	\item for each proper superset $\Sigma'$ of $\Sigma^{max}_{BS}(Q,\Sigma)$ such that $\Sigma' \subseteq \Sigma$, $D^{(Q_n)} \models \Sigma'$ does {\em not} hold. 
\end{itemize}
\vspace{-0.52cm}
\end{theorem}

\vspace{-0.1cm}

We now turn to the proof of Theorem~\ref{sigma-max-theorem}. We first observe that the process of sound chase of a CQ query using a set $\Sigma$ of  embedded dependencies under bag semantics can be modeled as state transitions for $\Sigma$, with certain conditions on the final state, which corresponds to obtaining the result of the chase. The termination conditions are formalized in Proposition~\ref{state-trans-prop}; we first set up the terminology required to formulate Proposition~\ref{state-trans-prop}. 

Suppose we are given a CQ query $Q$ and a finite set $\Sigma$ of embedded dependencies, such that there exists a  {\em set-}chase result $(Q)_{\Sigma,S}$ for $Q$ and $\Sigma$. Consider an arbitrary chase sequence ${\bf C} = Q_0, Q_1, \ldots$, such that (i) $Q_0 = Q$, and (ii) every query $Q_{i+1}$ ($i \geq 0$) in ${\bf C}$ is obtained from $Q_i$ by a sound chase step $Q_i \Rightarrow^{\sigma}_B Q_{i+1}$ using a dependency $\sigma \in \Sigma$. By Proposition~\ref{chase-termination-prop}, the chase sequence ${\bf C}$ is finite, that is,  ${\bf C} = Q_0, Q_1, \ldots, Q_n$, such that $n \in {\bf N} \cup \{ 0 \}$ and such that query $Q_n = (Q)_{\Sigma,B}$. Moreover, by Theorem~\ref{uniqueness-theorem} we have that the query $Q_n$ is bag-equivalent in the absence of dependencies\footnote{Other than the set-enforcing dependencies on stored relations, see Theorem~\ref{uniqueness-theorem}.} to the terminal queries in all sound-chase sequences for $Q$ and $\Sigma$ under bag semantics.  

Given a chase sequence ${\bf C}$ as defined above, with chase result $Q_n = (Q)_{\Sigma,B}$, we assign a unique ID to each subgoal of $Q_n$. We then ``propagate the IDs back'' to all the queries in ${\bf C}$, so that the enumeration of the subgoals is consistent across all the elements of ${\bf C}$. If extra subgoals are encountered in non-terminal elements of ${\bf C}$, we assign unique IDs to those subgoals as well. (The only case when a query $Q_i$, $i < n$, in ${\bf C}$ could have an extra subgoal compared to $Q_n$ is when the procedure of dropping duplicate subgoals has been applied to either $Q_i$ or its successors in ${\bf C}$. See Theorems~\ref{cv-theorem}, \ref{bag-chase-sound-theorem}, and~\ref{cv-updated-thm}.) In what follows, we refer to the $j$th subgoal of query $Q_i$ as $s^{(i)}_{j}$. 

Fix an arbitrary $i \in \{ 0,\ldots,n \}$, and consider query $Q_i$ in the chase sequence ${\bf C}$. Given an arbitrary dependency $\sigma \in \Sigma$, of the form $\sigma: \phi(\bar{U},\bar{W}) \rightarrow \exists \bar{V} \ \psi(\bar{U},\bar{V})$, we define the state of $\sigma$ w.r.t. $Q_i$ in ${\bf C}$ as follows:
\begin{itemize}
	\item Dependency $\sigma$ is {\em pre-applicable} to $Q_i$ if the chase of none of $Q_0, Q_1, \ldots, Q_i$ with $\sigma$ is applicable; that is, for each $j \in \{ 0,\ldots,i \},$ there does not exist a homomorphism from the left-hand side $\phi$ of $\sigma$ to the body of the query $Q_j$. 

	\item Dependency $\sigma$ is {\em soundly applicable} to set of subgoals $S = \{ s^{(i)}_{j1},\ldots,s^{(i)}_{jk} \}$ of query $Q_i$, for some $k > 0$, if there exists  a proper subset $\theta$, of size $k' \geq k$, of $\phi \wedge \psi$ (of $\sigma$), with the following properties:
	
	\begin{itemize} 
	
		\item $\theta$ is a superset of $\phi$; 
		
		\item there exists a homomorphism $h$ from $\theta$ to exactly the subgoals $s^{(i)}_{j1},\ldots,s^{(i)}_{jk}$ of query $Q_i$, such that $h$ cannot be extended to a homomorphism from 
$\phi \wedge \psi$ to the body of the query $Q_i$ (see Section~\ref{chase-prelims} for further details on this definition); and 
		
		\item chase step $Q_i \Rightarrow^{\sigma}_B Q'$, where $Q'$ is a CQ query, is sound; that is, $Q' \equiv_{\Sigma,B} Q_i$. 
	 
	 \end{itemize}
	

	\item Dependency $\sigma$ is {\em unsoundly applicable} to set of subgoals $S = \{ s^{(i)}_{j1},\ldots,s^{(i)}_{jk} \}$ of query $Q_i$, for some $k > 0$,  if there exists  a proper subset $\theta$, of size $k' \geq k$,  of $\phi \wedge \psi$ (of $\sigma$), with the following properties:
	
	\begin{itemize} 
	
		\item $\theta$ is a superset of $\phi$; 
		
		\item there exists a homomorphism $h$ from $\theta$ to exactly the subgoals $s^{(i)}_{j1},\ldots,s^{(i)}_{jk}$ of  query $Q_i$, such that $h$ cannot be extended to a homomorphism from 
$\phi \wedge \psi$ to the body of the query $Q_i$ (see Section~\ref{chase-prelims} for further details on this definition); and 
		
		\item chase step $Q_i \Rightarrow^{\sigma}_B Q'$, where $Q'$ is a CQ query, is {\em unsound;} that is, $Q' \equiv_{\Sigma,B} Q_i$ does not hold. 
	 
	 \end{itemize}

	\item Finally, dependency $\sigma$ is {\em post-applicable} to $Q_i$ (assuming $i > 0$) if (a) $\sigma$ is neither soundly applicable nor unsoundly applicable to $Q_i$, and (b) there exists a $j \in \{ 0,\ldots,i-1 \}$ such that $\sigma$ has been used in a sound chase step $Q_j \Rightarrow^{\sigma}_B Q_{j+1}$. 
	Observe that in this case, by definition of (sound) chase steps there exists a homomorphism from the conjunction of the left-hand side $\phi$ of $\sigma$ with the right-hand side $\psi$ of $\sigma$  to the body of the query $Q_i$.

\end{itemize}

In the above definition of the state of $\sigma \in \Sigma$ w.r.t. $Q_i$ in ${\bf C}$, the only difference between the states ``soundly applicable'' and ``unsoundly applicable'' is the soundness property of the chase step in question. Specifically, in the state ``$\sigma$ is soundly applicable to $Q_i$'' the chase step $Q_i \Rightarrow^{\sigma}_B Q'$ is sound under bag semantics, whereas in the state ``$\sigma$ is unsoundly applicable to $Q_i$'', the chase step $Q_i \Rightarrow^{\sigma}_B Q'$ is unsound.  

We now define the {\em state of the set of embedded dependencies $\Sigma$ w.r.t. $Q_i$ in} ${\bf C}$, as a total mapping $s_i^{\bf C}$ from $\Sigma$ to the set of the four above states (pre-applicable, post-applicable, soundly-applicable, and unsoundly-applicable), where the state $s_i^{\bf C}(\sigma)$ of each $\sigma \in \Sigma$ w.r.t. $Q_i$ in ${\bf C}$ is as follows:
\begin{itemize}
	\item $s_i^{\bf C}(\sigma)$ = ``soundly-applicable'' if and only if there exists a set $S$ of subgoals of $Q_i$ such that  $\sigma$ is soundly applicable to $S$ in $Q_i$; 
	\item $s_i^{\bf C}(\sigma)$ = ``unsoundly-applicable'' if and only if there exists {\em no} subset $S$ of subgoals of $Q_i$ such that  $\sigma$ is soundly applicable to $S$ in $Q_i$, {\em and} there exists a set $S'$ of subgoals of $Q_i$ such that  $\sigma$ is unsoundly applicable to $S'$ in $Q_i$; 
	\item $s_i^{\bf C}(\sigma)$ = ``post-applicable'' if $\sigma$ and $Q_i$ satisfy the conditions (a) and (b) of post-applicability, see above; and  
	\item $s_i^{\bf C}(\sigma)$ = ``pre-applicable'' if $\sigma$ and $Q_i$ satisfy the conditions of pre-applicability, see above.
\end{itemize} 

We now establish straightforward facts about the states of $\Sigma$ w.r.t. particular queries in the sound-chase sequence ${\bf C} = Q_0,\ldots,Q_n$, where $Q_n$ is the result of the sound chase of $Q$ using $\Sigma$ under bag semantics. The proofs of all the claims in Proposition~\ref{state-trans-prop} are immediate from the definitions in this section of the appendix and from the definitions of chase steps, see Section~\ref{chase-prelims}.  

\begin{proposition}
\label{state-trans-prop}
For a CQ query $Q$ and a set of embedded dependencies $\Sigma$ such that there exists a set-chase result $(Q)_{\Sigma,S}$ for $Q$ and $\Sigma$. Let ${\bf C} = Q_0,\ldots,Q_n$ be a sound-chase sequence for $Q$ and $\Sigma$ under bag semantics. In ${\bf C}$, $Q_0 = Q$, and $Q_n$ is the result $(Q)_{\Sigma,B}$ of the sound chase of $Q$ using $\Sigma$ under bag semantics. Then the following holds about the states of $\Sigma$ w.r.t. queries in ${\bf C}$. 
\begin{enumerate}
	\item Suppose that in the state $s_0^{\bf C}$ of $\Sigma$ w.r.t. query $Q_0$ in chase sequence ${\bf C}$, for all $\sigma \in \Sigma$ it holds that $s_0^{\bf C}(\sigma)$ is either  ``pre-applicable'' or ``unsoundly-applicable''. Then ${\bf C} = Q_0$. That is, $Q$ is isomorphic to $(Q)_{\Sigma,B}$. 
	\item Consider the state $s_n^{\bf C}$ of $\Sigma$ w.r.t. query $Q_n$ in chase sequence ${\bf C}$. 
Then for all $\sigma \in \Sigma$ it must hold that $s_n^{\bf C}(\sigma)$ is one of  ``pre-applicable'', ``post-applicable'', and ``unsoundly-applicable''. 
	\item For an arbitrary $i \in \{ 0,\ldots,n-1 \}$ (assuming $n > 0$), consider the state $s_i^{\bf C}$ of $\Sigma$ w.r.t. query $Q_i$ and  the state $s_{i+1}^{\bf C}$ of $\Sigma$ w.r.t. query $Q_{i+1}$. Then 
	\begin{enumerate}
		\item there must exist a $\sigma^* \in \Sigma$ such that $s_i^{\bf C}(\sigma^*)$ is ``soundly applicable'', and 
		
		\item for each $\sigma \in (\Sigma - \{ \sigma^* \})$, 
		$s_i^{\bf C}(\sigma) = s_{i+1}^{\bf C}(\sigma)$. 
	
	\end{enumerate}
	
\end{enumerate}
\vspace{-0.5cm}
\end{proposition}


We are now ready to prove Theorem~\ref{sigma-max-theorem}. 

\begin{proof}{(Theorem~\ref{sigma-max-theorem})}
Consider a fixed pair $(Q,\Sigma)$ satisfying the conditions of Theorem~\ref{sigma-max-theorem}, and let $Q_n$ be the result of  sound chase for $Q$ and $\Sigma$ under bag semantics, with canonical database $D^{(Q_n)}$.  We show that  the set $\Sigma^{max}_{B}(Q,\Sigma)$ is the result of removing from $\Sigma$ exactly those tgds $\sigma$ such that the chase step $Q_n \Rightarrow^{\sigma}_B Q'$, with some CQ query $Q'$ being the outcome of the chase step, is unsound under bag semantics. This claim is, in fact, immediate from Proposition~\ref{state-trans-prop}, in which it is shown that, for each dependency $\sigma$ in $\Sigma$ such that $\sigma$ is applicable to $Q_n$, $\sigma$ is {\em unsoundly} applicable to $Q_n$. 
\end{proof}

Finally, we outline the counterpart of algorithm {\sc Max-Bag-$\Sigma$-Subset} (of Section~\ref{semantic-sec-now}) for the case of bag-set semantics. 

 \begin{algorithm}[htp]

\label{algo: example}

\caption{Max-Bag-Set-$\Sigma$-Subset($Q,\Sigma$)}

\SetKwInOut{Input}{Input}

\SetKwInOut{Output}{Output} 

\SetKwInOut{Feature}{Feature}

\dontprintsemicolon

\Input{CQ query $Q$, set $\Sigma$ of embedded dependencies such that chase result $(Q)_{\Sigma,S}$ exists}

\Output{$\Sigma^{max}_{BS}(Q,\Sigma) \subseteq \Sigma$ s. t. \\ (1) $D^{((Q)_{\Sigma,BS})} \models \Sigma^{max}_{BS}(Q,\Sigma)$, and \\ (2) $\forall \ \Sigma'$ such that $\Sigma^{max}_{BS}(Q,\Sigma) \subset \Sigma' \subseteq \Sigma$, \\  $D^{((Q)_{\Sigma,BS})} \models \hspace{-0.27cm} / \hspace{0.2cm} \Sigma'$}

1. $(Q)_{\Sigma,BS} \ := \ soundChase(BS,Q,\Sigma); $ 
  
2. $\Sigma^{max}_{BS}(Q,\Sigma) \ := \ \Sigma ;$
  
3. \For{each $\sigma$ in $\Sigma$}
{
4.       \If{$soundChaseStep(\sigma,BS,(Q)_{\Sigma,BS}) = false$}
      {
5.     $\Sigma^{max}_{BS}(Q,\Sigma) \ := \ \Sigma^{max}_{BS}(Q,\Sigma) - \{ \sigma \} ;$
      }

}

6.  {\bf return} $\Sigma^{max}_{BS}(Q,\Sigma)$;

\end{algorithm}

The correctness and complexity results for {\sc Max-Bag-Set-$\Sigma$-Subset} are the same as their counterparts for algorithm {\sc Max-Bag-$\Sigma$-Subset}, see Theorem~\ref{correct-max-bag-sigma-theorem} and  Section~\ref{semantic-sec-now} for the details.

\section{Proofs of $\Sigma$-Equivalence-Tests \\ for CQ Queries}
\label{appendix-f}


To prove Theorems~\ref{bag-chase-equiv-theorem} and Theorem~\ref{bag-set-chase-equiv-theorem}, we first make a straightforward observation, as follows.
\begin{proposition}
\label{intro-sigma-prop}
Given two queries $Q$ and $Q'$ and a set of embedded dependencies $\Sigma$. Let $X$ be one of $B$, $BS$, $S$, which stand for bag, bag-set, and set semantics, respectively. Then $Q \equiv_X Q'$ implies $Q \equiv_{\Sigma,X} Q'$. 
\end{proposition}

The proof of Proposition~\ref{intro-sigma-prop} is straightforward from the definitions of query equivalence in presence and in the absence of  dependencies.

The proof of Theorem~\ref{bag-chase-equiv-theorem} is immediate from Propositions~\ref{chase-termination-prop} and~\ref{intro-sigma-prop}, from Theorem~\ref{uniqueness-theorem},  and from Lemmas~\ref{bag-chase-first-lemma} and~\ref{bag-chase-second-lemma}. Similarly, the proof of Theorem~\ref{bag-set-chase-equiv-theorem} is immediate from Propositions~\ref{chase-termination-prop} and~\ref{intro-sigma-prop}, from the analog of Theorem~\ref{uniqueness-theorem} for bag-set semantics (see Theorem~\ref{bag-set-uniqueness-theorem}),  and from straightforward analogs of Lemmas~\ref{bag-chase-first-lemma} and~\ref{bag-chase-second-lemma} for the case of bag-set semantics for query evaluation.


\begin{lemma}
\label{bag-chase-first-lemma}
Given CQ queries $Q$ and $Q'$, and given a set of embedded dependencies $\Sigma$ on schema $\cal D$ 
such that there exist {\em set-}chase results $(Q)_{\Sigma,S}$ for $Q$ and $(Q')_{\Sigma,S}$ for $Q'$.  Then $Q \equiv_{\Sigma,B} Q'$ implies $(Q)_{\Sigma,B} \equiv_{B} (Q')_{\Sigma,B}$ in the absence of all dependencies other than the set-enforcing dependencies on $\cal D$. 
%
\end{lemma}

\begin{proof}
First, from Proposition~\ref{chase-termination-prop} we obtain that sound chase of each of $Q$ and $Q'$ using $\Sigma$ is guaranteed to terminate under bag semantics. Further, from Theorem~\ref{uniqueness-theorem} it follows that there exist (1) a unique result $(Q)_{\Sigma,B}$ of sound chase  for $Q$, and (2) a unique result $(Q')_{\Sigma,B}$ of sound chase  for $Q'$ . Both results are unique  in the absence of all dependencies  other than the set-enforcing dependencies on $\cal D$, call these set-enforcing dependencies $\Sigma' \subseteq \Sigma$. 

From $Q \equiv_{\Sigma,B} Q'$ and by the soundness of chase in obtaining $(Q)_{\Sigma,B}$ and $(Q')_{\Sigma,B}$, we have $(Q)_{\Sigma,B} \equiv_{\Sigma,B} (Q')_{\Sigma,B}$. That is, on each bag-valued database $D$ that satisfies $\Sigma$, we have that $Q(D,B)$ and $Q'(D,B)$ are the same as bags.

To show that $(Q)_{\Sigma,B} \equiv_{B} (Q')_{\Sigma,B}$ in the absence of all dependencies  other than $\Sigma'$, it remains to prove that  $Q(D,B)$ and $Q'(D,B)$ are also the same as bags on each database $D$ that does not satisfy $\Sigma$ but does satisfy $\Sigma'$. There are two cases:

Case 1: Suppose $D$ violates only those dependencies that are not relevant in sound chase to either $Q$ or $Q'$. (In the terminology of Section~\ref{semant-appendix}, those would be exactly the dependencies that are pre-applicable to each of $(Q)_{\Sigma,B}$ and $(Q')_{\Sigma,B}$.) In this case, $D$ does not violate any dependencies as far as $(Q)_{\Sigma,B}$ or $(Q')_{\Sigma,B}$ are concerned, as formalized in Theorem~\ref{sigma-max-theorem}. 
Thus from $(Q)_{\Sigma,B} \equiv_{\Sigma,B} (Q')_{\Sigma,B}$ we obtain that $Q(D,B)$ and $Q'(D,B)$ are the same as bags on $D$.

Case 2: Suppose $D$ violates at least one dependency that is relevant in {\em sound} chase to either $Q$ or $Q'$. 
 (In the terminology of Section~\ref{semant-appendix}, those would be exactly the dependencies that are post-applicable to each of $(Q)_{\Sigma,B}$ and $(Q')_{\Sigma,B}$.) Still, by Theorem~\ref{sigma-max-theorem} the definitions of $(Q)_{\Sigma,B}$ and of $(Q')_{\Sigma,B}$ 
 ensure that all such relevant dependencies are enforced (i.e., do not fail) on all assignments $\gamma$ that satisfy each of $(Q)_{\Sigma,B}$ and  $(Q')_{\Sigma,B}$ w.r.t. $D$. Let $D_Q$ be the union of all tuples in all such satisfying assignments for $(Q)_{\Sigma,B}$ w.r.t $D$; $D_{Q'}$ is defined analogously  for $(Q')_{\Sigma,B}$.  Then $D' = D_Q \bigcup D_{Q'}$ satisfies all the dependencies of $\Sigma$ that are relevant in chase to either $Q$ or $Q'$. Thus, from  $(Q)_{\Sigma,B}$ $\equiv_{\Sigma,B} (Q')_{\Sigma,B}$ we obtain that $Q(D',B)$ and $Q'(D',B)$ are the same as bags. From the fact that none of the tuples of $D$ that are not in $D'$  participates in forming either $Q(D,B)$ or $Q'(D,B)$, it follows that $Q(D,B)$ and $Q'(D,B)$ are the same as bags  {\it on database $D$}. 
\end{proof}

\begin{lemma}
\label{bag-chase-second-lemma}
Given CQ queries $Q$, $Q'$, and given embedded dependencies $\Sigma$ on schema $\cal D$ 
such that there exist {\em set-}chase results $(Q)_{\Sigma,S}$ for $Q$ and $(Q')_{\Sigma,S}$ for $Q'$.  Then  $Q \equiv_{\Sigma,B} Q'$ holds whenever $(Q)_{\Sigma,B} \equiv_{B} (Q')_{\Sigma,B}$ in the absence of all dependencies  other than the set-enforcing dependencies on $\cal D$. 
\end{lemma}
The proof of Lemma~\ref{bag-chase-second-lemma} is immediate from the fact that each of $(Q)_{\Sigma,B}$ and $(Q')_{\Sigma,B}$ was obtained using sound chase steps under bag semantics (which implies $(Q)_{\Sigma,B} \equiv_{\Sigma,B} Q$ and $(Q')_{\Sigma,B} \equiv_{\Sigma,B} Q'$), as well as from  Propositions~\ref{chase-termination-prop}  and~\ref{intro-sigma-prop} and from transitivity of bag equivalence in presence of dependencies. 


\section{$\Sigma$-Based Version of Prop. 2.1}
\label{dep-b-bs-implic-appendix} 


In this appendix we provide the proof of Proposition~\ref{sigma-b-bs-s-implic-prop}, which is the dependency-based version of Proposition~\ref{b-bs-s-implic-prop} (\cite{VardiBagsPods93}, see Section~\ref{bag-equiv-defs} of this current paper). By Theorems~\ref{bag-chase-equiv-theorem} and~\ref{bag-set-chase-equiv-theorem}, the proof works both for the formulation of Proposition~\ref{sigma-b-bs-s-implic-prop} and for the formulation that parallels Proposition~\ref{b-bs-s-implic-prop} (see Proposition~\ref{he-b-bs-s-implic-prop} below.) We also provide a proof of Proposition~\ref{cor-dep-subset-prop}. Finally, we provide the analogs of Theorem~\ref{bag-c-and-b-theorem}  for (a) CQ queries under bag-set semantics, and for (b) CQ queries with grouping and aggregation.


\begin{proof}{(Proposition~\ref{sigma-b-bs-s-implic-prop})}

\mbox{}

{\it Proof of (1):} Assume 
\begin{equation}
\label{eq-one}
Q \equiv_{\Sigma,B} Q' . 
\end{equation}
or, equivalently (by Theorem~\ref{bag-chase-equiv-theorem}), assume 
\begin{equation}
\label{eq-two}
(Q)_{\Sigma,B} \equiv_{B} (Q')_{\Sigma,B}  
\end{equation}
in the absence of all dependencies other than the set-enforcing dependencies of the given database schema. 
Then Equation~\ref{eq-three} 
\begin{equation}
\label{eq-three}
(Q)_{\Sigma,B} \equiv_{BS} (Q')_{\Sigma,B} . 
\end{equation}
follows from Equation~\ref{eq-two} by Proposition~\ref{b-bs-s-implic-prop}. Equation~\ref{eq-four} 
\begin{equation}
\label{eq-four}
(Q)_{\Sigma,B} \equiv_{\Sigma,BS} (Q')_{\Sigma,B} . 
\end{equation}
follows from Equation~\ref{eq-three} by Proposition~\ref{intro-sigma-prop}. Equation~\ref{eq-five} 
\begin{equation}
\label{eq-five}
((Q)_{\Sigma,B})_{\Sigma,BS} \equiv_{BS} ((Q')_{\Sigma,B})_{\Sigma,BS} . 
\end{equation}
follows from Equation~\ref{eq-four} by Theorem~\ref{bag-set-chase-equiv-theorem}. Equation~\ref{eq-six} 
\begin{equation}
\label{eq-six}
(Q)_{\Sigma,BS} \equiv_{BS} (Q')_{\Sigma,BS} . 
\end{equation}
follows from Equation~\ref{eq-five} for the following reasons:
\begin{itemize}
	\item By Proposition~\ref{dep-subset-prop} (also see Theorem~\ref{bag-chase-sound-theorem} and the definitions of chase steps), the set $\Sigma_1 \subseteq \Sigma$ of dependencies that are soundly applicable  to a query under bag semantics is a subset of  the set $\Sigma_2 \subseteq \Sigma$ of dependencies that are soundly applicable  to the same query under bag-set semantics.
	\item From Theorem~\ref{uniqueness-theorem} and its analog for bag-set semantics (Theorem~\ref{bag-set-uniqueness-theorem}), it follows that\linebreak $((Q)_{\Sigma,B})_{\Sigma,BS} \equiv_{BS} (Q)_{\Sigma,BS}$, and similarly\linebreak $((Q')_{\Sigma,B})_{\Sigma,BS} \equiv_{BS} (Q')_{\Sigma,BS}$ . 
	\item By transitivity of $\equiv_{BS}$, we obtain Equation~\ref{eq-six}.
\end{itemize} 

Finally,  Equation~\ref{eq-seven} 
\begin{equation}
\label{eq-seven}
Q \equiv_{\Sigma,BS} Q' . 
\end{equation}
follows from Equation~\ref{eq-six} by Theorem~\ref{bag-set-chase-equiv-theorem}.

\mbox{}


{\it Proof of  (2):}  
Assume 
\begin{equation}
\label{eq-eight}
Q \equiv_{\Sigma,BS} Q' . 
\end{equation}
or, equivalently (by Theorem~\ref{bag-set-chase-equiv-theorem}), assume 
\begin{equation}
\label{eq-nine}
(Q)_{\Sigma,BS} \equiv_{BS} (Q')_{\Sigma,BS} . 
\end{equation}
Then Equation~\ref{eq-ten} 
\begin{equation}
\label{eq-ten}
(Q)_{\Sigma,BS} \equiv_{S} (Q')_{\Sigma,BS} . 
\end{equation}
follows from Equation~\ref{eq-nine} by Proposition~\ref{b-bs-s-implic-prop}. Equation~\ref{eq-eleven} 
\begin{equation}
\label{eq-eleven}
(Q)_{\Sigma,BS} \equiv_{\Sigma,S} (Q')_{\Sigma,BS} . 
\end{equation}
follows from Equation~\ref{eq-ten} by Proposition~\ref{intro-sigma-prop}. Equation~\ref{eq-twelve} 
\begin{equation}
\label{eq-twelve}
((Q)_{\Sigma,BS})_{\Sigma,S} \equiv_{S} ((Q')_{\Sigma,BS})_{\Sigma,S} . 
\end{equation}
follows from Equation~\ref{eq-eleven} by Theorem~\ref{chase-theorem}. Equation~\ref{eq-thirteen} 
\begin{equation}
\label{eq-thirteen}
(Q)_{\Sigma,S} \equiv_{S} (Q')_{\Sigma,S} . 
\end{equation}
follows from Equation~\ref{eq-twelve} for the following reasons:
\begin{itemize}
	\item By Proposition~\ref{dep-subset-prop} (also see Theorem~\ref{bag-set-chase-sound-theorem} and the definitions of chase steps), the set $\Sigma_1 \subseteq \Sigma$ of dependencies that are soundly applicable  to a query under bag-set semantics is a subset of  the set $\Sigma_2 \subseteq \Sigma$ of dependencies that are (always soundly) applicable  to the same query under set semantics.
	\item From the analog of Theorem~\ref{uniqueness-theorem} for bag-set semantics (Theorem~\ref{bag-set-uniqueness-theorem}) and from the definitions of chase steps, it follows that $((Q)_{\Sigma,BS})_{\Sigma,S} \equiv_{S} (Q)_{\Sigma,S}$, and similarly $((Q')_{\Sigma,BS})_{\Sigma,S} \equiv_{S} (Q')_{\Sigma,S}$ . 
	\item By transitivity of $\equiv_{S}$, we obtain Equation~\ref{eq-thirteen}.
\end{itemize} 

Finally,  Equation~\ref{eq-fourteen} 
\begin{equation}
\label{eq-fourteen}
Q \equiv_{\Sigma,S} Q' . 
\end{equation}
follows from Equation~\ref{eq-thirteen} by Theorem~\ref{chase-theorem}. 
\end{proof}

\vspace{-0.1cm}
\begin{proposition}
\label{he-b-bs-s-implic-prop}
Given two CQ queries $Q_1$ and $Q_2$, and a set of embedded dependencies $\Sigma$, such that there exists the {\em set-}chase result in chase of each of $Q_1$ and $Q_2$ using $\Sigma$. Then (1) $Q_1 \equiv_{\Sigma,B} Q_2$ implies  $Q_1 \equiv_{\Sigma,BS} Q_2$, and (2) $Q_1 \equiv_{\Sigma,BS} Q_2$ implies  $Q_1 \equiv_{\Sigma,S} Q_2$.
\end{proposition}

\vspace{-0.1cm}
We next provide a proof of Proposition~\ref{cor-dep-subset-prop}. 

\begin{proof}{(Proposition~\ref{cor-dep-subset-prop})}
Consider a pair $(Q,\Sigma)$ that satisfies conditions of Theorem~\ref{sigma-max-theorem}. By definition of chase steps (see Section~\ref{chase-prelims}), in an arbitrary {\em set-}chase sequence ${\bf C} = Q,Q_1,\ldots$ for $Q$ and $\Sigma$, for each element $Q_i$ of ${\bf C}$ such that  $Q_{i+1}$ is also an  element of ${\bf C}$, it holds that $Q_{i+1} \sqsubseteq_S Q_i$ in the absence of dependencies. (Also, trivially, for each CQ query $Q$ it holds that $Q \sqsubseteq_S Q$.) By transitivity and reflexivity of $\sqsubseteq_S$, for an arbitrary pair $(Q_i,Q_{i+j})$ (for $j \geq  0$) of elements of ${\bf C}$, it holds that $Q_{i+j} \sqsubseteq_S Q_i$. By definition of sound chase under bag and bag-set semantics (see Section~\ref{new-sound-chase-sec}), the same {\em set-containment} relationship $Q_{i+j} \sqsubseteq_S Q_i$ holds for an arbitrary pair $(Q_i,Q_{i+j})$ (for $j \geq 0$) of elements of a sound-chase sequence ${\bf C'}$ under bag or bag-set semantics. The rest of the proof of Proposition~\ref{cor-dep-subset-prop} is immediate from the result of Proposition~\ref{dep-subset-prop} that 
$\Sigma^{max}_{B}(Q,\Sigma) \subseteq  \Sigma^{max}_{BS}(Q,\Sigma) \subseteq \Sigma$ for the above fixed pair $(Q,\Sigma)$ and from Proposition~\ref{sigma-b-bs-s-implic-prop}. 
\end{proof}

We now provide the analog of Theorem~\ref{bag-c-and-b-theorem} for CQ queries under bag-set semantics.  

\vspace{-0.1cm}

\begin{theorem}
\label{bag-set-c-and-b-theorem}
Given CQ query $Q$ and set $\Sigma$ of embedded dependencies 
such that {\em set} chase 
of $Q$ under $\Sigma$ terminates in finite time. Then {\sc Bag-Set-C\&B} returns all $\Sigma$-minimal reformulations $Q'$ such that $Q' \equiv_{\Sigma,BS} Q$. 
\end{theorem}

\vspace{-0.1cm}
Finally, we provide the analog of Theorem~\ref{bag-c-and-b-theorem} for CQ queries with grouping and aggregation.  

\vspace{-0.1cm}

\begin{theorem}
\label{aggr-c-and-b-theorem}
Given CQ query $Q$ with aggregate\linebreak function $max$, $min$, $sum$, or $count$, and set $\Sigma$ of embedded dependencies 
such that {\em set} chase 
of the core of $Q$ under $\Sigma$ terminates in finite time. Then (1) If the aggregate function of $Q$ is $max$ or $min$, then {\sc Max-Min-C\&B} returns all $\Sigma$-minimal reformulations $Q'$ of $Q$ such that $Q' \equiv_{\Sigma} Q$;  (2) If the aggregate function of $Q$ is $sum$ or $count$, then {\sc Sum-Count-C\&B} returns all $\Sigma$-minimal reformulations $Q'$ of $Q$ such that $Q' \equiv_{\Sigma} Q$. 
\end{theorem}

\vspace{-0.1cm}


} 

\end{document}